\documentclass[12pt]{article}

\RequirePackage{amsthm,amsmath,amsfonts,amssymb}
\RequirePackage{bm}
\RequirePackage{url}

\usepackage{multirow}
\usepackage{colortbl}
\usepackage{adjustbox}
\usepackage{bbm,enumerate}

\RequirePackage{natbib}
\usepackage{comment}
\usepackage{appendix}
\usepackage{natbib}
\usepackage{enumitem}
\usepackage{graphicx}
\usepackage{subfigure}
\usepackage{booktabs}
\usepackage{array}
\usepackage{algorithm}
\usepackage{algpseudocode}

\usepackage{dsfont}
\usepackage{setspace}

\usepackage{bigints}

\usepackage[colorlinks,linkcolor=myred,anchorcolor=myblue,citecolor=myblue]{hyperref}
\usepackage{fullpage}
\usepackage[font={footnotesize,it}]{caption}
\numberwithin{equation}{section}

\usepackage{smile}
\usepackage{tikzstyle}

\newcommand{\calH}{\mathcal{H}}
\newcommand{\calP}{\mathcal{P}}
\newcommand{\calX}{\mathcal{X}}
\newcommand{\yihong}[1]{{\color{myred} \colorbox{myred!30}{YG} #1}}

\begin{document}

\title{\bf Neural Generative  Distributional Regression}
\author{Jinhang Chai\thanks{Department of Operations Research and Financial Engineering, Princeton University, jhchai@princeton.edu},\   Jianqing Fan\thanks{Department of Operations Research and Financial Engineering, Princeton University,  jqfan@princeton.edu.  His research is supported by NSF Grant DMS-2412029 and the ONR Grant N00014-25-1-2317},\  and Yihong Gu\thanks{Department of Biomedical Informatics, Harvard University, yihong\_gu@hms.harvard.edu}}
\maketitle

\setstretch{1} 
\begin{abstract}
Any continuous conditional distribution of $Y$ given $X$ can be generated from a transform of a known noise distribution $U$ such as the uniform or normal distribution via $Y = g(X, U)$. This paper provides an estimator of such a generative transformation $g$ by minimizing the empirical energy distance between distributions of $Y$ and $g(X, U)$, and implements it via neural networks. The estimated distribution can then be readily applied to downstream tasks such as conditional moment estimation, predictive interval construction, and conditional density estimation. By leveraging the representation power of neural networks, the estimator can adaptively exploit low-dimensional structures in a purely algorithmic manner. Theoretically, we establish an oracle inequality attaining the adaptive optimal nonparametric rates. Numerical simulations and real data analysis further demonstrate the practical effectiveness of the proposed method.

\end{abstract} 
\setstretch{1.5} 


\section{Introduction}
Deep learning has achieved tremendous success across a wide range of domains, including computer vision \citep{voulodimos2018deep}, game playing \citep{silver2016mastering}, and natural language processing \citep{brown2020language}. The success of deep neural network models is often attributed to their remarkable ability to represent and approximate complex, high-dimensional functions. From a theoretical standpoint, this perspective aligns naturally with statistical learning theory, where deep networks have been shown to possess powerful approximation and generalization capabilities. Notably, \citet{schmidt2020nonparametric} demonstrated that neural networks, with carefully designed architectures, can achieve the minimax-optimal rate for nonparametric regression while adaptively exploiting compositional structures in the target function.

Traditional nonparametric regression \citep{fan1996local,gyorfi2002distribution,tsybakov2009nonparametric} focuses on estimating the conditional mean function. However, the conditional mean captures only a limited view of the relationship between covariates and responses~\citep{shaked2007stochastic,hastie2009elements}. In many applications, richer aspects of the conditional distribution are of primary interest. For instance, conditional quantiles are crucial in heavy-tailed settings or risk-sensitive decision-making~\citep{wang2012estimation,gardes2010functional}; prediction intervals are indispensable for uncertainty quantification in fields such as economics~\citep{chudy2020long} and meteorology~\citep{umlauf2018primer}; and conditional score functions or higher-order moments play key roles in semi-parametric inference~\citep{chernozhukov2018double} and financial risk modeling~\citep{harvey2013dynamic}. 

Generative distributional regression directly estimates the conditional distribution of the response $Y$ given the covariate $X$ instead of a fixed point estimate. There is a considerable literature proposing parametric or semi-parametric methods for distributional regression \citep{rigby2005generalized, hothorn2014conditional}. Although the effects of the covariates on the conditional distribution can be modeled nonparametrically, for example, using neural networks \citep{klein2024distributional,kneib2023rage}, these methods do assume the conditional distribution belongs to certain parametric families. This may risk model misspecification and lack scalability under complex data-generating processes. 

With the development of modern machine learning models, including neural networks and tree-based methods, there is also a considerable literature on realizing distributional regression using these nonparametric techniques and a regression-type loss that shares a similar spirit to $L_2$ and $L_1$ losses in that the computation is easy and thus the implementation is similar to running conditional mean or conditional quantile regression, both in optimization stability and computation time. Two main types of loss are under consideration: the Continuous Ranked Probability Score (CRPS) \citep{henzi2021isotonic, clement2024distributional, padilla2025risk} and the conditional Maximum Mean Discrepancy (MMD) loss \citep{ren2016conditional, huang2022evaluating, cevid2022distributional, chatterjee2025one}.

However, it is still unclear whether these losses can be realized for both sample-efficient and computation-efficient distributional regression in theory. There is some work establishing the statistical rate of convergence for the estimators using CRPS loss, including the isotonic regression \citep{henzi2021isotonic}, and kernel estimators \citep{pic2023distributional}; the latter suffers from the curse-of-dimensionality in the covariate dimension $d$. \cite{padilla2025risk} shows that a variant of CRPS loss, integrated by neural networks, can efficiently adapt to the low-dimensional structure in the conditional CDF function. However, only the in-sample error of CDF, i.e., $\frac{1}{n}\sum_{i=1}^n \int \{ \hat{F}(t|X=X_i)-F^\star(t|X=X_i)\}^2 dt$, is established, and it is unclear how to simulate or estimate the conditional distribution of $Y$ for new observed $x$ and the correspondings theoretical guarantees. Furthermore, this estimation procedure is not computation-efficient, given that one needs to run the loss for different quantiles $t$, and thus is conceptually similar to running quantile regressions with $\Omega(\mathrm{poly}(n))$ cumulative probabilities. \cite{chatterjee2025one} establishes a convergence rate for a neural network using a conditional MMD loss. However, it explicitly uses the nearest neighbour in the estimation procedure, which slows the computation and inevitably leads to a statistical rate that suffers from the curse of dimensionality in $d$. 

In this paper, we fill the above gap by showing that a neural generative distributional regression method, which uses regression-type loss, with similar computation and implementation as regression, can fully leverage the neural network's sample efficiency in modeling conditional distributions and further leads to accurate downstream conditional quantities estimation. The core idea, which dates back to the conditional generative adversarial network \citep{mirza2014conditional}, is to represent the conditional law of $Y|X$ by transforming an exogenous noise variable through a deep neural network, i.e., $Y = g(X,U)$ with the covariate $X\in \mathbb{R}^d$ and some noise $U$. This implicit formulation unifies various estimation objectives -- conditional means, quantiles, moments, or densities -- into a single modeling framework. Once trained, downstream quantities can be efficiently estimated via Monte Carlo sampling from the fitted conditional simulator. While we consider the general kernel function in Appendix~\ref{sec:general-kernel}, we focus on the specific instance where the estimator is an empirical risk minimizer of a regression-type loss -- energy distance \citep{szekely2013energy} loss. This loss is a realization of the average MMD loss \citep{huang2022evaluating} with first-order Sobolev kernel, or constrained kernel $K(x,y) = -|x-y|$, and is also an unbiased estimate of the CRPS loss. The estimator is conceptually the same as the Engression \citep{shen2025engression} when the response $Y$ is one-dimensional.

We provide a thorough and sharp non-asymptotic analysis of the proposed method. We summarize the key findings below.
\begin{itemize}
    \item We establish the non-asymptotic oracle inequalities akin to regression for the proposed estimator. Moreover, in a similar spirit to the conditional moment regression, we show that as long as the conditional quantile function admits the form of a hierarchical composition model \citep{bauer2019deep, schmidt2020nonparametric, fan2022factor}, the estimator can get an optimal, no curse-of-dimension rate of convergence in the $L_2$ distance of the induced CDF. 

    \item Our error bound offers an explicit dependence on the number of auxiliary noises $U$. We show that a constant number of noise samples per batch, in particular, two, is sufficient for an optimal rate of convergence. This distinguishes significantly from prior distributional estimation methods that require a growing number of noise/quantile values to attain consistency. 

    \item Although the conditional distribution simulator $\hat{g}$ only has error guarantees of $L_2$ error in conditional CDFs functions, which is a relatively weak distance compared with, for example, TV distance, we show the simulator can successfully result in (1) optimal error rates in downstream conditional moments estimation; and (2) a no curse-of-dimensionality error rate with a constant pre-factor in exponent in downstream conditional quantiles/prediction-band/densities/score-functions. This supports its versatility and sample efficiency in downstream estimation tasks.
\end{itemize}

\subsection{Related Works}

Our study connects to several active areas of research spanning generative modeling, distributional prediction, and statistical learning theory.

\noindent\textbf{Implicit generative modeling.}
The proposed estimator is closely related to implicit generative models implemented via deep neural networks, including Generative Adversarial Networks (GANs; \citealp{goodfellow2014generative}) and diffusion-based models \citep{ho2020denoising}. 
Both frameworks have proven highly effective for flexible density modeling but face well-documented limitations. 
GANs and their variants, such as Wasserstein GANs~\citep{arjovsky2017wasserstein}, require solving a minimax optimization problem and often suffer from mode collapse, training instability, and sensitivity to hyperparameter tuning \citep{thanh2020catastrophic}, whereas diffusion models typically entail heavy computational overhead and slow inference due to iterative denoising procedures \citep{li2024snapfusion}. 
In contrast, our method offers a lightweight and stable alternative: it preserves the expressiveness of implicit generative models while maintaining computational complexity comparable to standard regression.

\noindent\textbf{Uncertainty quantification and distributional prediction.}
Our framework also contributes to the literature on uncertainty-aware prediction, encompassing quantile regression \citep{koenker2001quantile,belloni2019conditional}, conformal prediction \citep{lei2014distribution,barber2023conformal}, and predictive inference methods \citep{duchi2024predictive,fan2023utopia}. 
These approaches typically augment point estimators with post-hoc uncertainty measures or calibration procedures. 
By contrast, our method models the entire conditional distribution directly, thereby integrating uncertainty learning intrinsically within the estimation process and avoiding the need for separate calibration or conformalization steps.

\noindent\textbf{Energy-based objectives and distributional metrics.}
The loss function employed in our estimator is based on the energy distance \citep{szekely2013energy}, which can be viewed as a specific instance of the maximum mean discrepancy (MMD; \citealp{gretton2012kernel}), itself belonging to the broader class of integral probability metrics (IPM; \citealp{muller1997integral}).
Energy-based distances provide a flexible nonparametric approach for quantifying discrepancies between probability distributions when the likelihood function is computationally intractable.
A growing line of work has studied the statistical properties of MMD in parametric estimation problems \citep{briol2019statistical,cherief2022finite,oates2022minimum}. Our work differs by focusing on the nonparametric setting.
Moreover, our analysis reveals a surprising efficiency property: using as few as two auxiliary samples (\(m_2 = 2\)) per batch suffices to achieve the same asymptotic accuracy as employing a large number of samples.


\noindent\textbf{Implicit conditional modeling.}
Finally, our approach is conceptually related to Engression \citep{shen2025engression}, which also adopts an implicit generative framework for conditional distribution estimation. However, their work primarily focuses on the extrapolation capability of the method, considering the regime where auxiliary sample size per batch $m_2$ tends to infinity, and does not establish nonasymptotic results for nonparametric settings. In contrast, our work develops rigorous statistical theory—including optimal nonasymptotic convergence rates, oracle inequalities, and guarantees for downstream tasks.

\subsection{Organization and Notations}

The rest of the paper is organized as follows. Section \ref{sec:method} discusses our proposed method in detail. Section \ref{sec:theory} offers theoretical guarantees of the generic method and various downstream tasks. Section \ref{sec:numerical} showcases the effectiveness of our proposed method through numerical experiments. All the proofs are collected in the supplemental material.

\noindent\textbf{Notation}
Constants $c,C_0,C',\cdots$ may vary from line to line.
We denote $[m]=\{1,2\cdots,m\}$. For two nonnegative sequences $f(n)$ and $g(n)$, we use $f(n)\lesssim g(n)$ or $f(n)=\mathcal{O}(g(n))$ to represent that $f(n)\le Cg(n)$ for some universal constant $C$; similarly, we use $f(n)\gtrsim g(n)$ or $f(n)=\Omega(g(n))$ to represent that $f(n)\ge Cg(n)$ for some universal constant $C$. We write $f(n)\asymp g(n)$ if both $f(n)\lesssim g(n)$ and $f(n)\gtrsim g(n)$. Moreover, we have the shorthand that $a\wedge b=\min\{a,b\}$ and $a\vee b=\max\{a,b\}$. For a random event $\mathcal{E}$, we denote $\indicator{\mathcal{E}}$ or $\mathds{1}_{\mathcal{E}}$ to denote the indicator function of $\mathcal{E}$. For random samples $X_1,X_2\cdots,X_n\in \calX$, and $f$ a function defined on $\calX$, denote $\|f\|_n=\sqrt{\frac{1}{n}\sum_{i=1}^n f(X_i)^2}$. $\|f\|_2=\sqrt{\int f^2(x)d\mu(x)}$ is the 2-norm and $\|\cdot\|_{\infty}=\sup_{x\in \calX} f(x)$ is the $\infty$-norm.

\vspace{-1em}
\section{Setup}
In this section, we introduce our setup, as well as provide relevant background knowledge.

\vspace{-1em}
\subsection{Implicit Conditional Distribution Estimation via Sampling}

Let $X \in \mathbb{R}^p$ be the covariate and $Y\in \mathbb{R}$ be the response variable. Our primary goal is to understand or predict $Y$ based on observed $X$, characterizing the conditional distribution of $Y$ given $X$. While we can directly estimate probabilistic quantities like the conditional Cumulative Distribution Function (CDF) or conditional probability density function (PDF) directly, in this paper, we adopt a sampling-based approach.  Given an i.i.d. sample $(X_1, Y_1),\ldots, (X_n, Y_n) \sim \mu_0$, we estimate the conditional distribution of $Y$ given $X$ in an implicit manner. Specifically, we aim to learn a function $\hat g$, using the data $\mathcal{D} = \{(X_i, Y_i)\}_{i=1}^n$, such  that 
\begin{align}
\label{eq:goal}
    \hat g(X=x, U) \overset{d}{\approx} Y|X=x \qquad (X,Y)\sim \mu_0,
\end{align} 
for a noise vector $U \in \mathbb{R}^q$.
We refer to the above problem as the \emph{implicit conditional distribution estimation} since it is modeled implicitly via sampling.

\subsection{Energy Statistics}

Let $X$ and $Y$ be two independent random variables satisfying $\mathbb{E}[|X|], \mathbb{E}[|Y|] < \infty$ and let $F_X$ and $F_Y$ be their cumulative distribution functions (CDF), respectively. The energy distance (in one dimension) between $X$ and $Y$ is defined as
\begin{align}
\label{eq:energy-stat}
    E^2(X,Y) = 2\mathbb{E}\left[|X-Y|\right] - \mathbb{E}[|X-X'|] - \mathbb{E}[|Y-Y'|]
\end{align} where $X'$ and $Y'$ are independent copies of $X$ and $Y$, respectively. Energy distance \citep{szekely2013energy} in general is a statistical distance between two distributions of random vectors, with the name motivated by analogy to the potential energy between objects in a gravitational space. It has been widely applied in hierarchical clustering \citep{szekely2005hierarchical}, distributional testing \citep{szekely2005new,rizzo2010disco}, change-point detection \citep{kim2009using}, scoring rule \citep{gneiting2007strictly}, to name a few.

We have the following identity unveiling the relationship between the defined quantity \eqref{eq:energy-stat} and the $L_2$ difference between their CDF, also known as Cramér's distance \citep{cramer1928composition}.  
\begin{align}
\label{eq:energy-stat-ident}
    E^2(X,Y) = 2\int_{-\infty}^\infty \left\{F_X(t) - F_Y(t)\right\}^2 dt.
\end{align}
We will later utilize this crucial identity in our proof to connect the distributional distance with its sample-analog loss.

\subsection{Stochastic Neural Networks}

We consider the following stochastic fully connected deep neural network with ReLU activation $\sigma(\cdot) = \max\{0, \cdot\}$ that injects noise in the first layer of a standard fully connected deep neural network, and call it \emph{stochastic deep ReLU network} for short. Let $L, N, \tilde d$ be three positive integers, a \emph{deep ReLU network with depth $L$, ambient width $N$, and noise dimension $\tilde d$~\footnote{When $\tilde d = 1$, the learned neural network approximates the conditional quantile, as we will demonstrate later.  Typically, $\tilde d = 1$ suffices, but choosing a larger $\tilde d$ that possibly depends on the covariate dimension $p$ gives some flexibility of approximation.
For this reason, we treat $\tilde d$ as a generic positive integer, and we will show different choices in the numerical results.}} admits the form that is defined recursively as
\begin{align}
\label{eq:nn-architecture-0}
    g(x, u)= T_{L+1} \circ \bar{\sigma}_\ell \circ T_{L} \circ \cdots \circ T_1([x,u])
\end{align} Here $T_{l}(z) = W_l z + b_l: \mathbb{R}^{d_l} \to \mathbb{R}^{d_{l+1}}$ is a linear map with weight matrix $W_l \in \mathbb{R}^{d_{l}\times d_{l-1}}$ and bias vector $b_{l} \in \mathbb{R}^{d_{l}}$, where $(d_0,d_1\ldots, d_{L-1}, d_L, d_{L+1}) = (d, N, \ldots, N, N, 1)$, and $\bar{\sigma}_l: \mathbb{R}^{d_l} \to \mathbb{R}^{d_l}$ applies the ReLU activation $\sigma(\cdot)$ to each entry of a $d_l$-dimensional vector in $l\in [L]$. From the definition, we have $d = p + \tilde d$, where we recall $p$ is the covariate dimension. See an example of $L=3$, $N=7$, $d=4$ and $\tilde d=1$ in Fig. \ref{fig:snn-0}.

\begin{figure}
\begin{center}
\begin{tikzpicture}[scale=1.5]
\draw[gray] (-0.03, -0.33) rectangle +(0.24+0.03*2, 0.3*2+0.24+0.03*2);
\draw[gray] (0.12, 0.7) node {$x$};
\foreach \x in {-1, 0, 1}{
  \draw[black] (0, \x*0.3) rectangle +(0.24, 0.24);
}

\draw[myblue] (-0.03, -0.93+0.3) rectangle +(0.24+0.03*2, 0.24+0.03*2);
\draw[myblue] (0, -2*0.3) rectangle +(0.24, 0.24);
\draw[myblue] (0.12, -0.93-0.15+0.3) node {\myblue{$u$}};

\foreach \l in {1,2,3}{
  \foreach \x in {-3,-2,-1,0,1,2,3}{
    \draw[black] (\l, \x*0.3) rectangle +(0.24, 0.24);
  }
}

\draw[black] (4, 0) rectangle +(0.24, 0.24);

\foreach \xin in {-1,0,1,-2}{ 
  \foreach \y in {-3,-2,-1,0,1,2,3}{
    \draw[->, myred] (0+0.24, \xin*0.3+0.12) -- (1, \y*0.3+0.12);
  }
}

\foreach \l in {2,3}{
  \foreach \x in {-3,-2,-1,0,1,2,3}{
    \foreach \y in {-3,-2,-1,0,1,2,3}{
      \draw[->, myred] (\l-1+0.24, \x*0.3+0.12) -- (\l, \y*0.3+0.12);
    }
  }
}

\foreach \x in {-3,-2,-1,0,1,2,3}{
  \draw[->, myred] (3+0.24, \x*0.3+0.12) -- (4, 0*0.3+0.12);
}

\draw[black] (4+0.24+0.6, 0.4) node{\small $T_M(\cdot)$};
\draw[->] (4+0.24 + 0.1, 0.12) -- (4+0.24+1.1, 0.12);

\end{tikzpicture}
\end{center}
\caption{A visualization of the stochastic neural network: depth $L=3$, width $N=7$. It takes a $p=3$ covariate vector $x$ and a $\tilde d=1$ noise scalar $u$ as inputs to the first hidden layer.}
\label{fig:snn-0}
\end{figure}
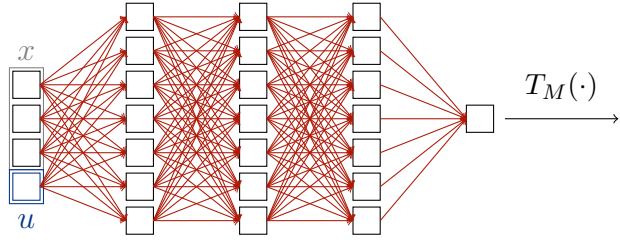

\begin{definition}[Stochastic Deep ReLU Network Class]
\label{def:snn-0}
    Define the family of stochastic deep ReLU networks taking $d$-dimensional vector as input with depth $L$, ambient width $N$, noise dimension $\tilde d$, truncated by $M$ and weights bounded by $B$ as 
        \[\mathcal{H}_{\mathtt{nn}}(d, L, N, \tilde d, M,B) = \{\tilde{g}(x, u) = \mathrm{Tc}_M(g(x, u)): g(x, u) \text{ in } \eqref{eq:nn-architecture-0} \ \text{with} \ \|W_l\|_{\max}\lor \|b_l\|_{\max}\le B\},\]
        where $\mathrm{Tc}_M: \mathbb{R} \to \mathbb{R}$ is the truncation operator defined as $\mathrm{Tc}_M(z) = \min\{|z|, M\} \cdot \mathrm{sign}(z)$.
\end{definition}



\section{Method}
\label{sec:method}
\subsection{Distributional Regression Using Energy Loss}

Let $\mathcal{D}=\{(X_i, Y_i)\}_{i=1}^n$ be an i.i.d. sample from $\mu_0$. Suppose the noise vector $U$ is sampled from the distribution $\mu_u$, the $\mu_u$-population-level objective given the data can be written as
\begin{align}
\label{eq:loss0}
    \mathsf{R}_n(g) = \frac{1}{n} \sum_{i=1}^n \Big\{2\mathbb{E}_U[|g(X_i, U) - Y_i|] - \mathbb{E}_{U,U'}[|g(X_i, U) - g(X_i, U')|] \Big\},
\end{align} 
which is the energy distance \eqref{eq:goal}, after ignoring constant term independent of $g$.
As the form $g(X, U)$ is data generative and \eqref{eq:loss0} is similar to the regression, but minimizing the discrepancy in distribution,
we term it as “generative distributional regression”.
However, the exact calculation of $\mathbb{E}_U[|g(X_i, U) - Y_i|]$ and $\mathbb{E}_{U,U'}[|g(X_i, U) - g(X_i, U')|]$ are both infeasible for complicated functions $g$ like neural networks. We consider the following fully empirical-level objective. For given $m=(m_1, m_2) \in \{1,2,\ldots\} \times \{2, 3,\ldots\}$, let $\mathcal{U} = \{U_{b, i, k}\}_{b\in [m_1], i\in [n], k\in [m_2]}$ be an i.i.d. sample of $\mu_u$ that is also independent of $\mathcal{D}$, our neural distributional regression estimator minimizes the following objective
\begin{align}
\label{eq:loss-n}
\begin{split}
    \mathsf{R}_{n,m}(g) = \frac{1}{m_1} \sum_{b=1}^{m_1} \frac{1}{n} \sum_{i=1}^n \Bigg(&\frac{2}{m_2} \sum_{k=1}^{m_2} |Y_i - g(X_i, U_{b,i,k})| \\
    &~~~~~~~~ - \frac{1}{m_2(m_2-1)}\sum_{k\neq k'} |g(X_i, U_{b, i, k}) - g(X_i, U_{b, i, k'})|\Bigg)
\end{split}
\end{align}
over the stochastic deep ReLU network class $\mathcal{H}_{\mathtt{nn}}(d, L, N, \tilde d, M, B)$ for network architecture parameters $(L, N, \tilde d, M,B)$. It is also possible to replace the Monte Carlo simulation with a low-discrepancy quasi-Monte Carlo approximation, including grid point approximation.  We do not pursue this direction.

Here we disentangle the number of auxiliary samples $\mathcal{U}$ into two dimensions $m=(m_1,m_2)$ rather than the predecessors that uses $U$-statistics \citep{cherief2022finite,briol2019statistical} that consider $m=(1, m_2)$.   First of all, it involves a quadratic computation in terms of $m_2$.  The computation in terms of $m_1$ can easily be distributed in several batches, while the computation in terms of $m_2$ can not.  As shown later, the statistical rates depend on the numerical error, which is related to $m_1 m_2$.  Therefore, it is computationally fast to consider a large $m_1$ and a small $m_2 \geq 2$.


Minimizing the objective function \eqref{eq:loss-n} can be approximated via the following batched gradient descent algorithm (Algorithm \ref{algo1}) with the number of epochs $T=m_1$ and the number of noise samples per data $K=m_2$ to be chosen. To choose a small $m_1$, we adopt Algorithm \ref{algo2} as a variant of Algorithm \ref{algo1} that is deferred to Appendix \ref{sec:algo2}. Note that Algorithm \ref{algo2} adopts a small $m_1$ by reusing the simulated noise, and Algorithm \ref{algo1} is a special case when taking $m_1=T$. 

\begin{algorithm}
\caption{Batched Gradient Descent Training for Neural Distributional Regression}
\label{algo1}
\begin{algorithmic}[1]
\State \textbf{Hyper-parameter: } number of epoches $T$, batch size $B$, number of noise samples per data $K$, stepsize $\eta$.
\State \textbf{Input:} data $\{(X_i, Y_i)\}_{i=1}^n$
\State Initialize $\theta$ with random weights
\For {$t \in \{1,\ldots, T\}$} 
    \State Shuffled sample $\breve{\mathcal{D}} =  \{(\breve{X}_{i}, \breve{Y}_{i})\}_{i=1}^n = \{(X_{\pi(i)}, Y_{\pi(i)})\}_{i=1}^n$ by random permutation $\pi$ in $[n]$.
    \For {$\ell \in \{1,\ldots, \lfloor n/B\rfloor\}$}
        \State Get batch ${\mathcal{B}} = \{(\tilde{X}_{i}, \tilde{Y}_i)\}_{i=1}^{B} = \{(\breve{X}_{i+(\ell-1)B}, \breve{Y}_{i+(\ell-1)B})\}_{i=1}^{B}$.
        \State Sample noises $\{\tilde{U}_{i,k}\}_{i\in [B], k\in [K
        ]}$ i.i.d. from $\mu_u$.
        \State Update $\theta$ by gradient descent with stepsize $\eta$:
            \begin{align*}
                \small 
                \theta\leftarrow \theta-\eta
                \nabla_\theta \left[\frac{1}{B}\sum_{i=1}^B \left\{\frac{2}{K} \sum_{k=1}^K |\tilde{Y}_i-g_\theta(\tilde{X}_i,\tilde{U}_{i,k})| - \frac{\sum_{k\neq l} |g_\theta(\tilde{X}_i, \tilde{U}_{i,k}) - g_\theta(\tilde{X}_i, \tilde{U}_{i, l})|}{K(K-1)}\right\}\right]
            \end{align*}
    \EndFor
\EndFor
\State \textbf{Output:} distributional regression estimate $g_\theta$. 
\end{algorithmic}
\end{algorithm}

Our neural distributional regression estimator approximately optimizes the objective \eqref{eq:loss-n}. Specifically, we return $\hat g$ such that
\begin{align}
\label{eq:estimator}
    \mathsf{R}_{n,m}(\hat{g}) \le \inf_{g\in \mathcal{H}_{\mathtt{nn}}(d, L, N, \tilde d, B)} \mathsf{R}_{n,m}(g) + \delta_{\mathtt{opt}}
\end{align} with some optimization error $\delta_{\mathtt{opt}}$.

\subsection{Downstream Estimation via Sampling}
\label{sec:downstream-estimation}

Given that our distributional regression estimator can learn the conditional distribution of $Y$ given $X$ using observations, many downstream estimation and inference problems can be further done via first sampling from our distributional regression estimator $\hat{g}$ in \eqref{eq:estimator} followed by straightforward calculation, which includes point estimations, prediction band construction, conditional density estimation, and conditional score function estimation. Let $U_1,\ldots U_k\overset{i.i.d.}{\sim} \mu_u$ as generated i.i.d. noises used in all the following, where $k$ is the user-defined number of noise samples in evaluation.

\noindent\textbf{Conditional Moments.} For given $x \in \mathbb{R}^p$, one can estimate the conditional expectation  $m_{w}^\star(x)=\mathbb{E}[w(Y)|X=x]$ for a given function $w$.  This can be estimated by $\mathbb{E}_U [w(\hat{g}(x, U))]$ and computed by generating a sufficiently large sample size $k$ and taking the generated sample mean:
\begin{align}
\label{eq:cm-est}
    \hat{m}_w(x) = \frac{1}{k} \sum_{l=1}^k w(\hat{g}(x, U_l)).
\end{align}
It can also be implemented by numerical integration with the midpoint method. 

\noindent\textbf{Conditional Quantile.} One may also be interested in estimating the conditional quantile of $Y$ given $X=x$, which can also be viewed as a prediction interval construction. Let $Q^\star(\alpha|x)$ be the conditional $\alpha$-quantile of $Y$, given by
\begin{align}
\label{eq:conditional-quantile}
    Q^\star(\alpha|x) = \inf_{q} \left\{\mathbb{P}_{(X,Y) \sim \mu_0}(Y \le q|X=x) \ge \alpha \right\}.
\end{align}
We use the following empirical quantile estimator based on generated samples of $\hat{g}(x, U)$ to estimate $Q^{\star}(\alpha|x)$, namely,
\begin{align}
\label{eq:cq-est}
\hat Q(\alpha|x)=\inf_{t}\Big\{t:\frac{1}{k}\sum_{l=1}^k\indicator{\hat g(x,U_l)\le t}\ge \alpha\Big\}.
\end{align}

\noindent\textbf{Prediction Interval. }
One may also be interested in constructing a prediction interval that ensures both good coverage and small average width. 
Here, we can use the generated samples to build a more straightforward and covariate-adaptive prediction band. For a pre-determined confidence level $\alpha \in (0,1)$, we first calculate the ordered statistics $\{\hat{Y}_{(l)}(x)\}_{l=1}^k$ of $\{\hat{g}(x, U_l)\}_{l=1}^k$, and let 
\begin{align}
\label{eq:pb-left-def}
\hat{l}(x)=\arg\min_{1\le l\le k-\lceil \alpha k\rceil}\left(\hat{Y}_{(l+\lceil \alpha k\rceil)}(x)-\hat{Y}_{(l)}(x)\right),
\end{align}
the shortest interval with empirical coverage $\alpha$.  Then, the prediction band is 
\begin{align}
\label{eq:pb-est}
\hat{\mathrm{PI}}_{\alpha}(x):=\left[\hat{Y}_{\hat{l}(x)}(x),\hat{Y}_{\hat{l}(x)+\lceil (1-\alpha) k\rceil}(x)\right].
\end{align}
The minimum in \eqref{eq:pb-left-def} is to reduce the length of the interval.  The prediction intervals can also be
constructed by using the traditional interval $ \left[\hat{Y}_{\lfloor \alpha k/2\rfloor}(x),\hat{Y}_{\lceil (1-\alpha/2) k\rceil}(x)\right]$.

\noindent\textbf{Conditional Density Function Estimation.} 
One can also build a conditional density estimator to estimate the conditional density $p^\star(y|x)$ on top of our learned conditional sampler $\hat{g}(x, \cdot)$. For a given smooth hyper-parameter $m \in \mathbb{N}^+$ and bandwidth $h$, let $K: \mathbb{R}\to \mathbb{R}$ be an $m$-order kernel function satisfying
\begin{align}
K(0)=1, \qquad  \int s^j K(s)ds=0 ~~ \forall j\in [m-1], \qquad \text{and} ~~ \int |s^m K(s)|ds< \infty. 
\end{align} 
Typical choices of second-order kernel are uniform,  Epanechnikov and Gaussian kernels. The kernel conditional density estimator for given $x$ is given by 
\begin{align}
\label{eq:cd-est}
\hat p(y|x)=\frac{1}{kh}\sum_{l=1}^k K\left(\frac{y - \hat g(x,U_l)}{h}\right).
\end{align}

\noindent\textbf{Conditional Score Function Estimation.} The conditional score function 
\begin{align*}
    s^\star(y|x) = \frac{d}{dy} \left[\log p^\star(y|x) \right] = \frac{d}{dy} \left[ p^\star(y|x) \right] \Big/ p^\star(y|x) 
\end{align*} is also of interest, for example, constructing a double robust estimation of the average partial effect \citep{klyne2023average}.   It also appears in data generation using diffusion models \citep{rezende2015flows,ho2020denoising}.  
Following the notation in the estimation of the density function, we estimate $s^\star(y|x)$ by
\begin{align}
\label{eq:cs-est}
\hat s(y|x)= \left\{\frac{1}{kh^2}\sum_{l=1}^k \tilde K'\left(\frac{y - \hat g(x,U_l)}{h}\right)\right\} \Big/ \max\left\{\frac{1}{kh}\sum_{l=1}^k K\left(\frac{y - \hat g(x,U_l)}{h}\right),c_d\right\}
\end{align}
where $\tilde K$ is an $(m-1)$-order kernel function, and $c_d$ is a constant used to bound the denominator from below.

All the above discussions in the downstream estimation are for fixed $x$. For evaluations on multiple $x$ values, the generated noises $\{U_l\}_{l=1}^k$ should be resampled for different $x$. 

\section{Theoretical Analysis}
\label{sec:theory}

Let $F^\star(x, y) = \mathbb{P}_{\mu_0}(Y\le y|X=x)$ be the conditional CDF of $Y$ given $X=x$. For a stochastic neural network $g$ taking covariate $x$ and noise $U\sim \mu_u$ as input, we define the induced conditional CDF $F_{g}(x, y)$ as
\begin{align}
\label{eq:induced-CDF}
    F_g(x, y) = \mathbb{E}_{U\sim \mu_u}\left[ \indicator{g(x, U) \le y}\right].
\end{align}
For two conditional CDFs $F(x, y)$ and $G(x, y)$, we define the following $L_2$-type metric
\begin{align}
\label{eq:l2-norm}
    \|F-G\|_2 = \sqrt{\int \left(\int_{-\infty}^\infty (F(x, t) - G(x, t))^2 dt \right) \mu_{0,x}(dx)}.
\end{align}
Similar to the identity \eqref{eq:energy-stat-ident}, the population-level counterpart of \eqref{eq:loss0} has the following identity
\begin{align}
\label{eq:identity-ccdf}
    \mathbb{E}\left[2|Y - g(X, U)| - |g(X, U) - g(X, U')|\right] = 2\|F_g - F^\star\|_2^2 + C_{\mu_0},
\end{align} where $(X,Y)$ are sampled from $\mu_0$, $(U, U')$ are independently drawn from $\mu_u$ that are also independent of $(X,Y)$, $C_{\mu_0}$ is a constant dependent on $\mu_0$ but independent of $g$; see a proof in Appendix \ref{proof:identity}. This means that minimizing the energy loss is equivalent to minimizing the $L_2$ norm $\|F_g-F^\star\|_2$ defined in \eqref{eq:l2-norm} in population, and one can also establish the error bound on $\|F_{\hat{g}} - F^\star\|_2$ for the empirical-level minimizer $\hat{g}$ in \eqref{eq:estimator}. 


\subsection{Main Non-Asymptotic Result}
We introduce some standard conditions.
\begin{condition}[Regularity]
\label{cond:bounded} There exists some constant $c_0>0$ such that
\begin{itemize}
\item[(a)] \underline{(Data Generating Process)} $\{(X_i, Y_i)\}_{i=1}^n$ are i.i.d. drawn from $\mu_0$.

\item[(b)] \underline{(Boundedness)} $X$ is supported on $[-c_0,c_0]^p$ and $|Y|\le c_0$ $\mu_0$-a.s. 

\item[(c)] \underline{(Regularity for Neural Networks)} $\min\{L, n\}\ge 3$, $N\ge \tilde d + d + 2$, and $M = c_0$. In addition, $\log B,\log (NL)\lesssim\log n$.
\end{itemize}
\end{condition}


Condition \ref{cond:bounded} (a)--(b) are standard in nonparametric statistics to facilitate technical proofs. It is also easy to replace the boundedness condition (b) by sub-Gaussianity of $(X,Y)$ at the cost of introducing $\mathrm{polylog}(n)$ factors in the error bound. Condition \ref{cond:bounded} (c) is imposed to simplify the presentation of the error bounds. Here $M$ can be larger than $c_0$: one can choose the truncation parameter at $\mathrm{polylog}(n)$ rate and it will not affect the result up to poly-log factors. 

Letting $t>0$ be the logarithm of failure probability, we define the following error quantities:
\begin{align}
\delta_{\mathtt{a}} = \inf_{g\in \mathcal{H}_{\mathtt{nn}}(d, L, N, \tilde d, B)} \|F_g - F^\star\|_2^2 \qquad \text{and} \qquad \delta_{\mathtt{s}} = \frac{(NL)^2 \log (n) + t}{n}.
\label{eq:nn-error}
\end{align}
The first quantity measures the approximation error of the stochastic neural networks to the ground truth conditional distribution $F^\star$ with respect to the metric defined in \eqref{eq:l2-norm}, and the second quantity is the standard stochastic error. 

We have the following oracle-type inequality regardless of the structure of $F$. It shows that the error bound is the sum of the approximation error $\delta_{\mathtt{a}}$, the stochastic error $\delta_{\mathtt{s}}$, and the product of $\sqrt{\delta_{\mathtt{s}}}$ and the parametric rate in terms of $m_1m_2$. The dependency on the first two terms is similar to running a standard regression. It also characterizes the dependency on the number of auxiliary noise samples $m=(m_1,m_2)$.

\begin{theorem}
\label{thm:oracle-n-infty}
Assume Condition \ref{cond:bounded} holds. For implicitly estimated conditional distribution function $F_{\hat{g}}$, there exists some constant $C$ dependent only on $c_0$ such that for any $t>0$ and $n\ge 3$, 
\begin{align}
\label{eq:est-error-ndre}
    \|F_{\hat{g}} - F^\star\|^2_2/{C} \le \delta_{\mathtt{a}} + \delta_{\mathtt{s}} + \sqrt{\delta_{\mathtt{s}}} \sqrt{\frac{1}{m_1 m_2}}+  \delta_{\mathtt{opt}} := \delta_{\mathtt{NDRE}}
\end{align} with probability at least $1-Ce^{-t}$.
\end{theorem}

\begin{remark}
This theorem shows that the stochastic error originated from $U_{b,i,k}$ can be negligible when $m_1m_2\gtrsim \frac{1}{\delta_s}$, and we get the usual oracle-type inequality in nonparametric estimation. The computation burden is low when we take $m_2\ge 2$ to be small, as $m_1$ can be inherently regarded as the number of epochs $T$ in running SGD; see the definition of $T$ in Algorithm \ref{algo1}. Notably, the optimal convergence can be attained by choosing a large $m_1$ even when $m_2=2$. 
\end{remark}

\begin{remark}
The results hold for more kernel MMD with mild smoothness besides the first-order Sobolev kernel. We leave the details in Appendix~\ref{sec:general-kernel}.
\end{remark}

\begin{remark}
We also conduct an experiment to empirically corroborate the theoretical findings. 
Figure~\ref{fig:error-decay} depicts the empirical relationship between the CDF error $\|F_{\hat g}-F^*\|_2$ and the product $m_1 m_2$. 
The data are generated from\vspace{-1em}
\[
Y = \sin(3X) + \bigl(1 - 0.7\cos(4X)\bigr)U, 
\qquad U \sim \mathcal{U}[-1,1],
\]
with the neural network architecture specified as $(L, N, \tilde d) = (3, 100, 1)$. 
The CDF error is computed via numerical integration using the midpoint rule over a 2000-point uniform grid. 
Each curve, shown in different colors, corresponds to a distinct choice of $m_2$, while the horizontal axis represents the product $m_1 m_2$. The shaded bands indicate the standard errors across repeated trials.
As observed in Figure~\ref{fig:error-decay}, the error decays approximately polynomially with $m_1 m_2$ at first and subsequently levels off.

\end{remark}

\begin{figure}[h]
	\centering
	\includegraphics[width=0.5\linewidth]{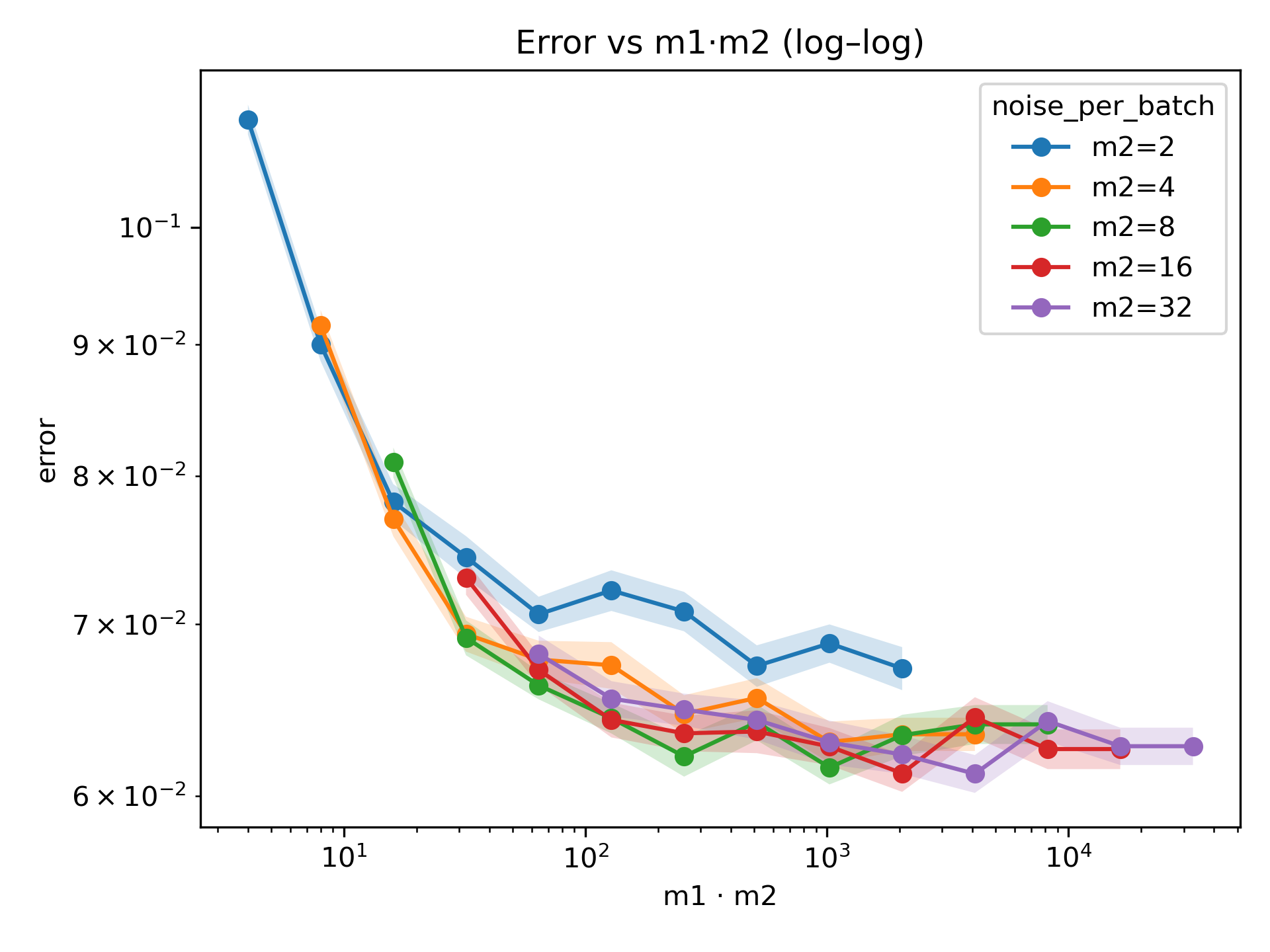}
	\caption{Empirical validation on a log--log plot with sample size $n = 2000$, averaged over 100 trials. 
		The x-axis represents the product $m_1 m_2$, and the y-axis shows the CDF error $\|F_{\hat g} - F^*\|_2$. 
		Each colored curve depicts the mean error across 100 trials for a given $m_2$, with shaded bands indicating the corresponding standard errors. Specifically, blue, orange, green, red, purple curves correspond to $m_2=2,4,8,16,32$, respectively.}
	\label{fig:error-decay}
\end{figure}

\subsection{Adapting to Low-dimensional Structures}
In this section, we establish the explicit $L_2$ estimation error on $\|F_{\hat{g}} - F^\star\|_2$ when the conditional quantile function $Q^\star(\alpha|x)$ defined in \eqref{eq:conditional-quantile} possesses low-dimensional structures. To emphasize that $Q^\star(\alpha|x)$ is an $(p+1)$-dimensional function that takes $x \in \mathbb{R}^p$ and $\alpha \in \mathbb{R}$ as input, we also write it as $Q^\star(x, \alpha)$ in this section. 

\begin{definition}[$(\beta,C)$-smooth Function]
Let $\beta=r+s$ for some nonnegative integer $r$ and $0<s\le 1$, and $C>0$. A $d$-variate function is $(\beta,C)$-smooth if for every $\balpha\in \NN^d$ such that $\sum_{i=1}^d \alpha_i=r$, the partial derivative $\partial^r f/\partial x_1^{\alpha_1}\cdots x_d^{\alpha_d}$ exists and satisfies 
\begin{align*}
\Big|\frac{\partial^r f}{\partial x_1^{\alpha_1}\cdots x_d^{\alpha_d}}(\bx)-\frac{\partial^r f}{\partial x_1^{\alpha_1}\cdots x_d^{\alpha_d}}(\bz)\Big|\le C\|\bx-\bz\|^s_2, \text{\quad for  every\ } \bx,\bz
\end{align*}
\end{definition}

Neural networks are well-known for their ability to be adaptive to the low-dimensional compositional structures in nonparametric regression \citep{kohler2021rate,fan2022factor} without the supervision of the precise composition structure. We introduce the idea of the hierarchical composition model, which is the composition of $t$-variate $(\beta, C)$ smooth functions with $(\beta, t) \in \mathcal{P}$.

\begin{definition}[Hierarchical composition model]
\label{hcm}
    We define function class of hierarchical composition model $\mathcal{H}(d, l, \mathcal{P})$ \citep{kohler2021rate} with $l, d \in \mathbb{N}^+$, $C\in \mathbb{R}^+$, and $\mathcal{P}$, a subset of $[1,\infty) \times \mathbb{N}^+$, in a recursive way as follows~\footnote{We omit the dependency on $C$ for simplicity.}. Let $\mathcal{H}(d, 0,\mathcal{P})=\{h(x)=x_j, j\in [d]\}$, and for each $l\ge 1$, \vspace{-1em}
    \begin{align*}
    \mathcal{H}(d, l,\mathcal{P}) = \big\{&h: \mathbb{R}^d \to \mathbb{R}: h(x) = g(f_1(x),...,f_t(x))\text{, where} \\
    &~~~~~ g \text{ is } (\beta,C) \text{ smooth} \text{ with } (\beta, t)\in \mathcal{P} \text{ and } f_i \in \mathcal{H}(d, l-1,\mathcal{P})\big\}.
\end{align*} 
For a hierarchical composition model, we define the following dimension-adjusted degree of smoothness \citep{fan2022factor} that measures the hardest component in the composition:\vspace{-1em}
\[\gamma^\star_{\calH(d,l,\calP)}=\min_{(\beta,t)\in \calP}\frac{\beta}{t},\]
when clear from the context, we abbreviate it as $\gamma^\star$.
\end{definition}

\begin{condition}[Function Complexity]
\label{asp:conditional-quantile-HCM}
$Q^\star(x,\alpha)\in \calH(p+1,l,\calP)$\footnote{Here we view $Q^\star$ as a joint function of $(x,\alpha)$.}. Moreover, there exists some constant $c_1>0$ such that $|F^\star(x,y_1)-F^\star(x,y_2)|\le c_1 |y_1-y_2|$ for every $x \in [0,1]^p$ and $y_1,y_2 \in [-1, 1]$.
\end{condition}

Under Condition~\ref{asp:conditional-quantile-HCM}, we have the following corollary. Notably, no assumption is imposed on the distribution of covariates $x$.
\begin{corollary}
\label{cor:final-rate}
Let $\tilde d \ge 1$. Let $\mu_u$ be a distribution with bounded support whose c.d.f. {$F_0$} is $(\gamma^{\star},C)$ smooth. Suppose Conditions \ref{cond:bounded} and \ref{asp:conditional-quantile-HCM} hold. Also take $NL\asymp n^{\frac{1}{4\gamma^{\star}+2}}(\log n)^{\frac{4\gamma^{\star}-1}{2\gamma^{\star}+1}}$, $m_1 m_2\gtrsim n^{\frac{2\gamma^{\star}}{2\gamma^{\star}+1}}$. Then there exists some universal constant $C$ such that with probability at least $1-Ce^{-t}$,
\begin{align*}
{\|F_{\hat{g}} - F^\star\|^2_2}/{C}\le n^{-\frac{2\gamma^{\star}}{2\gamma^{\star}+1}}(\log n)^{\frac{12\gamma^{\star}}{2\gamma^{\star}+1}}+\delta_{\mathtt{opt}}+\frac{t}{n}.
\end{align*}
\end{corollary}

As a specific case, consider the location-scale model as follows
\vspace{-1em}
\begin{align*}
    Y = f(X) + s(X) \cdot \varepsilon ~~\text{with}~~ \varepsilon \sim F_\varepsilon,
\end{align*} here $F_\varepsilon$ is the CDF for the noise $\varepsilon$. It is easy to see that the conditional CDF is a composition of $f(\cdot), s(\cdot)$ and $F_\varepsilon$, that is, $F^\star(x,y)=F_\varepsilon(\frac{y-f(x)}{s(x)})$. Define $Q_\varepsilon(\cdot)$ to be the quantile function of $\varepsilon$. We have $Q^\star(x,\alpha)=f(x)+s(x)Q_\varepsilon(\alpha)$.

\begin{assumption}
\label{asp:HCM-location-scale}
Suppose $f(x), s(x)\in \calH(p,l_1,\calP_1)$ with $\gamma^\star_{\calH(p,l_1,\calP_1)}=\gamma^\star_1$. And $Q_\varepsilon$ is $(\beta_1,C_1)$-smooth.
\end{assumption}

\begin{corollary}
Under the setting of Corollary \ref{cor:final-rate}, if Assumption~\ref{asp:HCM-location-scale} further holds, the conclusion of Corollary \ref{cor:final-rate} holds with $\gamma^{\star}=\min\{\gamma^{\star}_1,\beta_1\}$.
\end{corollary}

\begin{proof}
Define $h(f_1,f_2,f_3)=f_1+f_2f_3$. It is straightforward to show that $h$ is $(\beta,1)$-smooth for any $\beta\ge 4$. Note that $Q^{\star}(x,\alpha)=h(f(x),s(x),Q_\epsilon(\alpha))$, hence for $\calP_2=\calP_1\cup\{(\beta,C_1),(\beta,1)\}$, $Q^{\star}\in \calH(p+1,l_1+1,\calP_2)$. Pick $\frac{\beta}{3}>\min\{\gamma^{\star}_1,\beta_1\}$, we have $\gamma^{\star}(\calH(p+1,l_1+1,\calP_2))=\min\{\gamma^{\star}_1,\beta_1\}$.  
It then follows immediately from Corollary \ref{cor:final-rate}.
\end{proof}

\subsection{Applications to Downstream Estimation}

\subsubsection{Conditional Moments}

Estimating conditional moments—such as the conditional mean or variance—is a central problem in statistics. In particular, estimating the conditional mean reduces to a standard regression task. Under suitable smoothness assumptions, the minimax rate~\citep{stone1982optimal} can be attained by various nonparametric methods, including kernel estimators~\citep{nadaraya1964estimating} and local polynomial regression~\citep{fan1996local}. 
Our approach also achieves the optimal rate, provided that the Monte Carlo sample size $k$ is sufficiently large.

\begin{condition}
\label{cond:conditional-moments}
The function $w$ is differentiable and satisfies $\sup_{y\in [-c_0, c_0]} w(y) \le c_1$ and $\int |w'(y)|^2 dy \le c_1$ for some constant $c_1$. 
\end{condition}

\begin{proposition}
\label{prop:conditional-moments}
Under the setting of Theorem \ref{thm:oracle-n-infty}, assume further that Condition \ref{cond:conditional-moments} holds. There exists some constant $C$ dependent on $(c_0, c_1)$ such that under the event in \eqref{eq:est-error-ndre}, the conditional moment estimator constructed in \eqref{eq:cm-est} satisfies, for any $k\ge 1$, 
\begin{align*}
    \mathbb{E}_{X\sim \mu_0}\left[|\hat{m}_w(X) - m^\star_w(X)|^2\Big|\hat g\right] \le C\left(\frac{1}{k} + \delta_{\mathtt{NDRE}}\right).
\end{align*} Suppose we further have $n'$ i.i.d. observations $X'_1, \ldots, X'_{n'}$ drawn from $\mu_0$.  Then, for any $n'\ge 3$ and $t'>0$, the following event happens with probability at least $1-e^{-t'}$.
\begin{align}
\label{eq:pred-error}
    \frac{1}{n'} \sum_{i=1}^{n'} \left(\hat{m}_w(X'_i) - m^\star_w(X'_i)\right)^2 \le C\left(\delta_{\mathtt{NDRE}} +\frac{1}{k}+ \frac{t'}{n'}\right)
\end{align} 
\end{proposition}

The proof is relegated to Appendix \ref{sec:proof-cond-moments}.
Proposition \ref{prop:conditional-moments} indicates that when the number of auxiliary samples satisfies $k \ge \delta_{\mathtt{NDRE}}^{-1} $, one can expect to estimate $m^\star_w(X)$ well on average when the event that $\hat{g}$ is estimated well, i.e., \eqref{eq:est-error-ndre}, occurs. Furthermore, if $n'$, the number of test data, is also large in that $n'\ge (\delta_{\mathtt{NDRE}})^{-1} \log(n)$, then the prediction error in the test data, which is defined in the L.H.S. of \eqref{eq:pred-error}, is of order $O(\delta_{\mathtt{NDRE}})$ with probability at least $1-n^{-10}$.
It is easy to see that Condition \ref{cond:conditional-moments} is satisfied with moment functions $w(y)=y^m$ for $m>0$.

\subsubsection{Conditional Quantile}
Conditional quantile estimation is another important problem~\citep{koenker1978regression}. The quantile estimator based on local polynomial~\citep{chaudhuri1991nonparametric} and neural networks~\citep{shen2021deep} both achieve minimax optimal rate. 
Below, we impose an assumption that the conditional density is bounded away from zero, similar to~\citet{belloni2011l1} and \citet{he1994convergence}. Note that our result is uniform over the quantile level.

\begin{condition}
\label{cond:conditional-quantile}
The conditional density of $Y=y$ given $X=x$, denoted as $p^\star(y|x)$, exists everywhere and $p^\star(y|x)\ge c_d$ for every $x,y$.
\end{condition}

\begin{proposition}
\label{prop:condition-quantile}
Under the setting of Theorem \ref{thm:oracle-n-infty}, further shown that Condition \ref{cond:conditional-quantile} holds. There exists some constant $C$ dependent on $(c_0, c_d)$ such that under the event in \eqref{eq:est-error-ndre}, for a fixed $0<\alpha<1$, the conditional quantile estimator constructed in \eqref{eq:cq-est} satisfies, for any $k\ge 1$, 
\begin{align*}
\EE\sup_{0<\alpha<1}\Big[(\hat Q(X,\alpha)-Q^{\star}(X,\alpha))^2\Big|\hat g\Big]\le C\left[(\delta_{\mathtt{NDRE}})^{2/3}+\frac{\log k}{k}\right].
\end{align*}
Suppose we further have $n'$ i.i.d. observations $X'_1, \ldots, X'_{n'}$ drawn from $\mu_0$, then for any $n'\ge 3$ and $t'>0$, the following event holds with probability at least $1-e^{-t'}$.
\begin{align}
\label{equ:result-conditional-quantile}
\frac{1}{n'} \sum_{i=1}^{n'} \sup_{0<\alpha<1}(\hat Q(X'_i,\alpha)-Q^{\star}(X'_i,\alpha))^2 \le C\left[(\delta_{\mathtt{NDRE}})^{2/3}+\frac{\log k}{k}+\frac{t'}{n'}\right]
\end{align} 
\end{proposition}

The proof is provided in Appendix \ref{sec:proof-conditional-quantile}. The bound in~\eqref{equ:result-conditional-quantile} comprises three components. The first term, $(\delta_{\mathtt{NDRE}})^{2/3}$, corresponds to the statistical rate incurred when translating an $L_2$-type CDF estimation error into an $L_\infty$-type quantile estimation error. This rate can be improved if the conditional density satisfies higher-order smoothness assumptions, under which an additional smoothing layer can be applied to exploit the extra regularity. The second term arises from Monte Carlo sampling variability, which becomes negligible when $k\gg(\delta_{\mathtt{NDRE}})^{-2/3}$. The third term accounts for the contribution from the estimation tail.

\subsubsection{Prediction Band}
Prediction bands are a central object in conformal prediction, where the goal is to guarantee valid coverage while keeping the band width as small as possible. Conformalized Quantile Regression~\citep{romano2019conformalized}) has been particularly successful in this regard, as it combines the distribution-free validity of conformal prediction with the local adaptivity of quantile regression, leading to sharper intervals in heteroskedastic settings. More recently, LinCDE~\citep{gao2022lincde} views prediction band construction as a natural byproduct of conditional density estimation. While attractive for its unified treatment of density and interval estimation, theoretical results on the efficiency or optimal width remain underdeveloped.

\begin{condition}
\label{cond:prediction-band}
The conditional density of $Y=y$ given $X=x$, denoted as $p(y|x)$, exists everywhere and $p(y|x)\le C_d$ for every $x,y$.
\end{condition}
\begin{proposition}
\label{prop:prediction-band}
Under the setting of Theorem \ref{thm:oracle-n-infty}, assume further that Condition \ref{cond:prediction-band} holds. There exists some constant $C$ dependent on $(c_0, C_d)$ such that under the event in \eqref{eq:est-error-ndre},  the prediction band constructed in \eqref{eq:pb-est} satisfies, for any $k\ge 1$, 
\begin{align}
\label{eq:prop-prediction-band-1}
\EE\Big[\sup_{0<\alpha<1}\big|\PP(Y\in \hat{\mathrm{PI}}_{\alpha}(X)|X,\hat g)-\alpha\big|^2\Big|\hat g\Big]\le C\left[(\delta_{\mathtt{NDRE}})^{2/3}+\frac{\log k}{k}\right]
\end{align}
Suppose we further have $n'$ i.i.d. observations $(X'_1,Y'_1), \ldots, (X'_{n'},Y'_{n'})$ drawn from $\mu_0$.  Then for any $n'\ge 3$ and $t'>0$, the following event holds with probability at least $1-e^{-t'}$.
\begin{align}
\label{eq:prop-prediction-band-2}
    \frac{1}{n'} \sum_{i=1}^{n'} \sup_{0<\alpha<1}\big|\PP(Y'_i\in \hat{\mathrm{PI}}_{\alpha}(X'_i)|\hat g)-\alpha\big|^2 \le C\left[(\delta_{\mathtt{NDRE}})^{2/3}+\frac{\log k}{k}+\frac{t'}{n'}\right]
\end{align}

Furthermore, if Condition \ref{cond:conditional-quantile} holds, we have 
\begin{align}
\label{eq:prop-prediction-band-3}
\EE\sup_{0< \alpha<1}\left[|\hat{\mathrm{PI}}_{\alpha}(X)|-\inf_{(l,u):\PP(Y\in [l,u]|X)\ge \alpha}(u-l)\right]\le C\left[(\delta_{\mathtt{NDRE}})^{2/3}+\frac{\log k}{k}\right].
\end{align}

\end{proposition}

The proof can be found in Appendix \ref{sec:proof-prediction-band}. The right-hand side of~\eqref{eq:prop-prediction-band-2} takes the same form as~\eqref{equ:result-conditional-quantile}, since prediction intervals are conceptually analogous to quantiles. In~\eqref{eq:prop-prediction-band-3}, we further quantify the optimality of the proposed prediction band. As both the sample size and the number of Monte Carlo simulations increase, the proposed band approaches the performance of the oracle counterpart.

\subsubsection{Conditional Density Estimation}
Conditional density estimation is a core problem in statistics and machine learning, with applications in prediction, anomaly detection, and generative modeling. Classical methods include local likelihood and bandwidth selection techniques \citep{fan1996estimation, hall2004cross, hyndman1996estimating}. 

\begin{condition}
\label{cond:conditional-density-estimation}
For every fixed $x$, $F^{\star}(x,y)\in \mathbb{C}^{m}$ is $m$-times differentiable, and $\partial^{m}F^{\star}(x,y)/\partial y^{m}$ is $l_f$-Lipschitz for some constant $l_f$.
\end{condition}
\begin{proposition}
\label{prop:conditional-density-estimation}
Under the setting of Theorem \ref{thm:oracle-n-infty}, assume further that Condition \ref{cond:conditional-density-estimation} holds. There exists some constant $C$ dependent on $(c_0, l_f)$ such that under the event in \eqref{eq:est-error-ndre}, the conditional density estimator constructed in \eqref{eq:cd-est} satisfies
\begin{align*}
\EE\Big[\int(\hat p(y|X)-p^{\star}(y|X))^2dy\Big|\hat g\Big]\le  C\max\{k^{-\frac{1}{2m+1}},(\delta_{\mathtt{NDRE}})^{\frac{1}{2m+3}}\}^{2m}.
\end{align*}
Suppose we further have $n'$ i.i.d. observations $X'_1, \ldots, X'_{n'}$ drawn from $\mu_0$, then for any $n'\ge 3$ and $t'>0$, the following event happens with probability at least $1-e^{-t'}$.
\begin{align}
\label{equ:result-conditional-density}
    \frac{1}{n'} \sum_{i=1}^{n'} \int(\hat p(y|X'_i)-p^{\star}(y|X'_i))^2dy\le C\left(\max\{k^{-\frac{1}{2m+1}},(\delta_{\mathtt{NDRE}})^{\frac{1}{2m+3}}\}^{2m}+\frac{t'}{n'}\right)
\end{align} 
\end{proposition}

The proof is provided in Appendix~\ref{sec:proof-conditional-density-estimation}.
The right-hand side of~\eqref{equ:result-conditional-density} consists of three components.
The first arises from Monte Carlo sampling variability, the second reflects the intrinsic statistical rate for estimating the conditional distribution, and the third captures the contribution from the tail behavior.
The degradation in rate is again attributable to transforming an $L_2$-type CDF estimation error into a kernel density estimation error.
As the number of Monte Carlo samples $k$ increases and the smoothness parameter $m$ becomes larger, the overall rate approaches $\delta_{\mathtt{NDRE}}$.

\subsubsection{Conditional Score Function Estimation}
The score function is a fundamental building block of efficient estimation and inference \citep{van2000asymptotic}. Estimating conditional scores provides access to gradients of log-densities, which are crucial for tasks such as semiparametric inference~\citep{chernozhukov2018double} and score-based generative modeling~\citep{song2020score}.

\begin{condition}
\label{cond:conditional-score-function-estimation}
For every $x,y$, $\frac{d}{dy}[p^{\star}(y|x)]$ exists and $|\frac{d}{dy}[p^{\star}(y|x)]|\le C_q$.~\footnote{This condition can be implied by Condition~\ref{cond:conditional-density-estimation} with $m\ge 2$; it is restated here for clarity.}
\end{condition}

\begin{proposition}
\label{prop:conditional-score-estimation}
Under the setting of Theorem \ref{thm:oracle-n-infty}, assume further that Condition \ref{cond:conditional-quantile}, \ref{cond:conditional-density-estimation} and \ref{cond:conditional-score-function-estimation} holds. There exists some constant $C$ dependent on $(c_0, l_f,c_d)$ such that under the event in \eqref{eq:est-error-ndre}, the conditional score function estimator constructed in \eqref{eq:cs-est} satisfies, for any $k\ge 1$,
\begin{align*}
\EE\Big[\int(\hat s(y|X)-s^{\star}(y|X))^2dy\Big|\hat g\Big]\le  C\max\{k^{-\frac{1}{2m+1}},(\delta_{\mathtt{NDRE}})^{\frac{1}{2m+3}}\}^{2(m-1)}
\end{align*}
Suppose we further have $n'$ i.i.d. observations $X'_1, \ldots, X'_{n'}$ drawn from $\mu_0$, then for any $n'\ge 3$ and $t'>0$, the following event happens with probability at least $1-2e^{-t'}$.
\begin{align}
\label{equ:result-conditional-score}
    \frac{1}{n'} \sum_{i=1}^{n'} \int(\hat s(y|X'_i)-s^{\star}(y|X'_i))^2dy\le C\left(\max\{k^{-\frac{1}{2m+1}},(\delta_{\mathtt{NDRE}})^{\frac{1}{2m+3}}\}^{2(m-1)}+\frac{t'}{n'}\right)
\end{align} 
\end{proposition}

The proof is provided in Appendix~\ref{sec:proof-conditional-score-estimation}.
In~\eqref{equ:result-conditional-score}, the only difference from~\eqref{equ:result-conditional-density} lies in the exponent of the maximum term, which is 
$2(m-1)$ instead of $2m$.
This adjustment reflects the fact that the score function corresponds to a first-order derivative in the nonparametric setting, thereby reducing the effective smoothness by one order.
\section{Numerical Studies}
\label{sec:numerical}

In this section, we present both simulation and real-data experiments to validate the theoretical results and demonstrate the empirical performance of the proposed method.

\subsection{Simulation Models}

We begin by examining several synthetic data-generating mechanisms designed to evaluate different aspects of the proposed estimator. 
Throughout, the covariate \( X \) is drawn from the uniform distribution on \([-1,1]\), that is, \( X \sim \mathcal{U}[-1,1] \). We investigate five models representing distinct conditional structures for 
\( Y \mid X \):
Model~1a is a univariate location–scale model, and Model~1b provides its multivariate extensions.
Models~2a and~2b specify explicit and implicit forms of the conditional distribution, respectively.
Finally, Model~3 is a Gaussian mixture model designed to test the estimator’s ability to capture multimodality.

\noindent\textbf{\underline{Model 1a (Univariate location–scale model).}}
We begin with a heteroskedastic location–scale model of the form \vspace{-1em}
\[
Y = \mu(X) + \sigma(X)U, \quad 
\mu(X) = \sin(3X), \quad 
\sigma(X) = 1 - 0.7\cos(4X),
\]
where \( U \sim \mathcal{U}[-1,1] \). This model features nonlinear mean and variance structures.

\noindent\textbf{\underline{Model 1b (Multivariate location–scale model).}}
We next extend Model~1a to a five-dimensional covariate \( X = (x_1,\ldots,x_5) \). 
Two variants are examined:
\begin{enumerate}
    \item \emph{Additive model:}\vspace{-1em}
    \[\mu_{\mathsf{add}}(X) = \sin(3x_1) + 0.5\sin(3x_2), \quad 
    \sigma_{\mathsf{add}}(X) = 1 - 0.3\cos(x_1) - 0.2\cos(2x_3).
    \]
    \item \emph{Interactive model:}
$\qquad
    \mu_{\mathsf{int}}(X) = x_1x_2, \quad 
    \sigma_{\mathsf{int}}(X) = 1 - \frac{1}{2}x_3x_4.
$
\end{enumerate}
These settings allow us to examine both additive and interaction effects in higher dimensions.

\noindent\textbf{\underline{Model 2a (Explicit conditional density).}}
We then consider a model with an analytically tractable conditional probability density function (PDF):
\[
p_Y(t \mid X) = \left( \frac{1}{2} - \frac{v(X)t}{16} \right)\mathbb{I}\{|t| \le 1\}, \quad v(X) = 1 + X + 4X^2.
\]
The corresponding conditional cumulative distribution function (CDF), quantile function, and conditional mean are
\vspace{-1em}
\begin{align*}
F_Y(t \mid X) &= (t+1)\Big(\tfrac{1}{2} + \tfrac{(1-t)v(X)}{32}\Big),\\ \quad
Q(\alpha \mid X) &= \frac{8 - \sqrt{64 + v(X)\big(16 + v(X) - 32\alpha\big)}}{v(X)}, \quad
\mathbb{E}[Y \mid X] = -\frac{v(X)}{24}.
\end{align*}

\noindent\textbf{\underline{Model 2b (Implicitly defined conditional distribution).}}
Here we specify \( Y \) through a transformation of a latent variable \( U \):\vspace{-1em}
\[
Y = g(X,U) = \exp\!\left(-\frac{v(X)U + U^2}{4}\right), \quad U \sim \mathcal{U}[0,1],
\]
with \( v(X) = 1 + X + 4X^2 \). 
The corresponding conditional CDF admits the closed-form expression 
\[
F_Y(t \mid X) = \frac{-v(X) + \sqrt{v(X)^2 - 16\log t}}{2}, \quad t \in [0, e^{-(v(X)+1)/4}],
\]
and the conditional mean is given by \vspace{-1em}
\[
\mathbb{E}[Y \mid X] = 2\sqrt{\pi} \Big[ \Phi\!\Big(\frac{v(X)}{2\sqrt{2}} + \frac{1}{\sqrt{2}}\Big) - \Phi\!\Big(\frac{v(X)}{2\sqrt{2}}\Big) \Big] e^{v(X)^2/16},
\]
where \( \Phi \) denotes the standard normal CDF.

\noindent\textbf{\underline{Model 3 (Gaussian mixture model).}}
Finally, we include a conditional Gaussian mixture model to capture multimodality:\vspace{-1em}
\[
Y \sim \frac{1}{2}\mathcal{N}\!\left(\tfrac{1}{2} + \cos(4\pi X), \tfrac{1}{2}\right) 
+ \frac{1}{2}\mathcal{N}\!\left(\tfrac{1}{2} + \cos^2(4\pi X), \tfrac{1}{2}\right).
\]
This design allows us to assess the estimator’s robustness under non-Gaussian and mixture structures.

\vspace{-2em}
\subsection{Hyperparameter Grid Search}
We begin by examining the influence of the key hyperparameters \( (m_1, m_2, \tilde d) \) on estimation accuracy. 
Here, \( m_1 \) and \( m_2 \) denote the numbers of samples used for training and evaluation, respectively, while \( \tilde d \) corresponds to the noise dimension parameter in the model architecture. 
Performance is assessed using the $L_2$ error of the estimated conditional CDF.

\noindent\textbf{Metric computation.}
To evaluate the CDF estimation error, we generate 2000 independent samples of \( X \) for computing the expectation with respect to $X$, 2000 evenly spaced grids for computing integration with respect to $Y$,  and compute the $L_2$ error of the estimated CDF via midpoint method for numerical integration. 

\noindent\textbf{Results.}
We conduct experiments under three representative models: Model~1b (Multivariate Additive Location–Scale), Model~2a, and Model~3. 
For all experiments, the sample size is $10000$, and the number of maximum training epochs is fixed at 2000 with early stopping. The batch size and learning rate are set to be 5000 and 0.01, respectively. Each configuration is independently replicated 50 times to account for stochastic variation. 
The detailed results are reported in Tables~\ref{table:model1b}–\ref{table:model3}.

\noindent\textbf{Discussion.}
Several consistent patterns emerge from the results. First, increasing \(m_1\) and \(m_2\) generally improves estimation accuracy. Notably, \(m_2\) has the greater effect, although its quadratic computational cost warrants practical consideration. Second, for relatively simple, low-dimensional problems, a smaller noise dimension (\(\tilde d = 1\)) is sufficient, whereas in higher-dimensional settings, increasing the noise dimension to \(\tilde d = 3\) yields clear benefits.

\begin{table}[ht]
\centering
\footnotesize

\begin{minipage}{0.48\textwidth}
\centering
\caption{Model 1b (Mean $\pm$ Std of CDF Error)}
\label{table:model1b}
\begin{tabular}{c|c|cc}
\toprule
$m_1$ & $m_2$ & $\tilde d=1$ & $\tilde d=3$ \\
\midrule
20   & 2  & $0.0843 \pm 0.0100$ & $0.0827 \pm 0.0077$ \\
20   & 10 & $0.0726 \pm 0.0081$ & $0.0684 \pm 0.0063$ \\
20   & 50 & $\mathbf{0.0675 \pm 0.0055}$ & $\mathbf{0.0629 \pm 0.0048}$ \\
\midrule
200  & 2  & $0.0806 \pm 0.0081$ & $0.0791 \pm 0.0073$ \\
200  & 10 & $0.0727 \pm 0.0083$ & $0.0702 \pm 0.0070$ \\
200  & 50 & $\mathbf{0.0683 \pm 0.0055}$ & $\mathbf{0.0653 \pm 0.0057}$ \\
\midrule
2000 & 2  & $0.0815 \pm 0.0101$ & $0.0796 \pm 0.0080$ \\
2000 & 10 & $0.0735 \pm 0.0085$ & $0.0693 \pm 0.0060$ \\
2000 & 50 & $\mathbf{0.0661 \pm 0.0049}$ & $\mathbf{0.0639 \pm 0.0051}$ \\
\bottomrule
\end{tabular}
\end{minipage}
\hfill
\begin{minipage}{0.48\textwidth}
\centering
\caption{Model 2a (Mean $\pm$ Std of CDF Error)}
\label{table:model2a}
\begin{tabular}{c|c|cc}
\toprule
$m_1$ & $m_2$ & $\tilde d=1$ & $\tilde d=3$ \\
\midrule
20   & 2  & $0.0267 \pm 0.0053$ & $0.0303 \pm 0.0060$ \\
20   & 10 & $0.0212 \pm 0.0031$ & $0.0226 \pm 0.0041$ \\
20   & 50 & $\mathbf{0.0207 \pm 0.0025}$ & $\mathbf{0.0204 \pm 0.0031}$ \\
\midrule
200  & 2  & $0.0250 \pm 0.0042$ & $0.0265 \pm 0.0052$ \\
200  & 10 & $0.0217 \pm 0.0035$ & $0.0221 \pm 0.0032$ \\
200  & 50 & $\mathbf{0.0206 \pm 0.0026}$ & $\mathbf{0.0208 \pm 0.0033}$ \\
\midrule
2000 & 2  & $0.0260 \pm 0.0062$ & $0.0277 \pm 0.0068$ \\
2000 & 10 & $0.0229 \pm 0.0039$ & $0.0216 \pm 0.0034$ \\
2000 & 50 & $\mathbf{0.0219 \pm 0.0030}$ & $\mathbf{0.0202 \pm 0.0032}$ \\
\bottomrule
\end{tabular}
\end{minipage}

\end{table}

\begin{table}[ht]
\centering
\footnotesize
\caption{Model 3 (Mean $\pm$ Std of CDF Error)}
\label{table:model3}
\begin{tabular}{c|c|cc}
\toprule
$m_1$ & $m_2$ & $\tilde d=1$ & $\tilde d=3$ \\
\midrule
20   & 2  & $0.0771 \pm 0.0100$ & $0.0998 \pm 0.0096$ \\
20   & 10 & $0.0545 \pm 0.0060$ & $0.0592 \pm 0.0089$ \\
20   & 50 & $\mathbf{0.0502 \pm 0.0051}$ & $\mathbf{0.0516 \pm 0.0071}$ \\
\midrule
200  & 2  & $0.0664 \pm 0.0116$ & $0.0863 \pm 0.0109$ \\
200  & 10 & $0.0535 \pm 0.0066$ & $0.0565 \pm 0.0112$ \\
200  & 50 & $\mathbf{0.0501 \pm 0.0054}$ & $\mathbf{0.0473 \pm 0.0032}$ \\
\midrule
2000 & 2  & $0.0684 \pm 0.0087$ & $0.0935 \pm 0.0121$ \\
2000 & 10 & $0.0539 \pm 0.0079$ & $0.0543 \pm 0.0074$ \\
2000 & 50 & $\mathbf{0.0492 \pm 0.0049}$ & $\mathbf{0.0485 \pm 0.0049}$ \\
\bottomrule
\end{tabular}
\end{table}

\subsection{Comparison with Standard Regression}
We next compare the proposed Neural Distributional Regression (NDR) method with conventional regression approaches, including least squares regression and quantile regression. 
To ensure a fair comparison, we employ identical neural networks and training configurations across all methods. 
Specifically, the networks are implemented with depth 3 and width 100, trained for 500 epochs with a batch size of 5000, learning rate of 0.01, and weight decay set to zero, thereby aligning the hyperparameter settings across different estimators.

The quantitative results are summarized in Table~\ref{table:comparison-standard-regression}. We replicate the experiments for $50$ times. Each entry is presented in the form of $\text{mean}\pm \text{std}$.
For relatively simple settings, NDR achieves performance comparable to that of standard regression methods, despite addressing a more general and challenging task of learning the entire conditional distribution rather than a single functional moment. 
For more complex scenarios, such as Model~1b, NDR demonstrates superior performance, highlighting its advantage in capturing nonlinear and heteroskedastic structures that standard regression methods fail to fully model.

\begin{table}[ht]
\centering
\scriptsize
\caption{Comparison of error metrics using standard regression and NDR. Mean errors are measured with $L_2$ loss; quantile errors are measured with $L_1$ loss.}
\label{table:comparison-standard-regression}
\begin{tabular}{c|ccc|ccc}
\toprule
\multirow{2}{*}{Model} 
& \multicolumn{3}{c|}{Regression} 
& \multicolumn{3}{c}{NDR} \\
& Mean ($L_2$) & $Q_{0.5}\ (L_1)$ & $Q_{0.25}\ (L_1)$ 
& Mean ($L_2$) & $Q_{0.5} (L_1)$ & $Q_{0.25}\ (L_1)$ \\
\midrule
Model 1a 
& 0.0323 $\pm$ 0.0083 
& 0.0363 $\pm$ 0.0095 
& 0.0319 $\pm$ 0.0075 
& 0.0411 $\pm$ 0.0108 
& 0.0476 $\pm$ 0.0078 
& 0.0424 $\pm$ 0.0105 \\
Model 1b 
& 0.0731 $\pm$ 0.0058 
& 0.0799 $\pm$ 0.0073 
& 0.0721 $\pm$ 0.0064 
& 0.0699 $\pm$ 0.0073 
& 0.0668 $\pm$ 0.0062 
& 0.0638 $\pm$ 0.0061 \\
Model 2a 
& 0.0169 $\pm$ 0.0090 
& 0.0241 $\pm$ 0.0142 
& 0.0167 $\pm$ 0.0086 
& 0.0242 $\pm$ 0.0039 
& 0.0308 $\pm$ 0.0037 
& 0.0246 $\pm$ 0.0031 \\
Model 2b 
& 0.0056 $\pm$ 0.0015 
& 0.0077 $\pm$ 0.0018 
& 0.2759 $\pm$ 0.0057 
& 0.0083 $\pm$ 0.0021 
& 0.0102 $\pm$ 0.0018 
& 0.2761 $\pm$ 0.0049 \\
Model 3 
& 0.0510 $\pm$ 0.0063 
& 0.0549 $\pm$ 0.0075 
& 0.0532 $\pm$ 0.0390 
& 0.0816 $\pm$ 0.0137 
& 0.0644 $\pm$ 0.0081 
& 0.0569 $\pm$ 0.0079 \\
\bottomrule
\end{tabular}
\end{table}

\subsection{Real Data Experiments}
We further evaluate the proposed NDR method on two benchmark real-world datasets. 
For both datasets, we preprocess the data by removing outliers using the interquartile range (IQR) rule and standardizing all features to have zero mean and unit variance.

\noindent\textbf{California Housing.}
The first dataset is the well-known \emph{California Housing} dataset, originally compiled from the 1990 U.S. Census by \citet{pace1997sparse}. 
It contains a total of 20{,}640 observations, each representing a California district, with eight predictor variables such as median income, average number of rooms, average occupancy, and geographic coordinates. 
The response variable is the median house value in each district (measured in hundreds of thousands of U.S. dollars). 
We use all eight predictors as covariates.

\noindent\textbf{Protein Tertiary Structure (CASP).}
The second dataset is the \emph{Physicochemical Properties of Protein Tertiary Structure} dataset, available from the UCI Machine Learning Repository.\footnote{\url{https://archive.ics.uci.edu/dataset/265/physicochemical+properties+of+protein+tertiary+structure}} 
It consists of 45{,}730 samples, each characterized by nine continuous physicochemical descriptors (e.g., hydrophobicity, polarity, and secondary structure score). 
The target variable is the root mean square deviation (RMSD) between the predicted and actual protein structures, which is a continuous measure of structural similarity.

\noindent\textbf{Experimental setup.}
We compare NDR with two state-of-the-art conditional distribution estimation methods implemented in \texttt{R}: \emph{LinCDE} \citep{gao2022lincde} and \emph{Distribution Boosting} \citep{friedman2020contrast}. 
For each dataset, we perform 100 random permutations and split the data into training (70\%), validation (15\%), and test (15\%) sets.
All three methods are trained and evaluated on identical splits to ensure comparability. 
Performance is assessed based on (i) the average width and empirical coverage of the 95\%-targeted prediction interval, and (ii) the average negative log-likelihood (NLL) loss computed on the test set.

\noindent\textbf{Results.}
The results, summarized in Table~\ref{table:real-data}, demonstrate that NDR consistently outperforms both LinCDE and Distribution Boosting across the two datasets. 
In particular, NDR achieves narrower prediction intervals while maintaining accurate coverage, reflecting its ability to learn sharper conditional distributions. 
The competing methods tend to produce overly conservative prediction bands, resulting in excessive interval width. 
Moreover, despite its improved accuracy, NDR attains comparable computational efficiency to the baseline methods, highlighting its scalability to moderately large datasets.

\begin{table}[ht]
\centering
\small
\caption{Comparison of NDR, LinCDE, and Distribution Boosting (DB) on Housing and Protein datasets (100 seeds). Reported are the mean $\pm$ standard deviation.}
\label{table:real-data}
\begin{tabular}{lcccc}
\toprule
Dataset & Method & NLL & Coverage ($\alpha=0.95$) & Width ($\alpha=0.95$) \\
\midrule
\multirow{3}{*}{Housing} 
  & NDR & $\mathbf{0.569 \pm 0.051}$ & 0.947 $\pm$ 0.010 & 1.739 $\pm$ 0.108 \\
  & LinCDE       & 0.596 $\pm$ 0.014 & 0.977 $\pm$ 0.003 & 2.379 $\pm$ 0.019 \\
  & DB (conTree) & 0.820 $\pm$ 0.059 & 0.970 $\pm$ 0.004 & 2.246 $\pm$ 0.029 \\
\midrule
\multirow{3}{*}{Protein} 
  & NDR & $\mathbf{0.604 \pm 0.030}$ & 0.950 $\pm$ 0.005 & 2.205 $\pm$ 0.114 \\
  & LinCDE       & 0.616 $\pm$ 0.010 & 0.975 $\pm$ 0.002 & 2.799 $\pm$ 0.006 \\
  & DB (conTree) & 0.752 $\pm$ 0.027 & 0.966 $\pm$ 0.002 & 2.739 $\pm$ 0.016 \\
\bottomrule
\end{tabular}
\end{table}

\setstretch{1} 
\bibliographystyle{apalike2}
\bibliography{main.bib}
\newpage

\setstretch{1.5} 
\appendixpage
\appendix

\section{Proof of Theorem \ref{thm:oracle-n-infty}}

We write the population-level counterpart of \eqref{eq:loss-n} as
\begin{align*}
    \mathsf{R}_{\infty, \infty}(g) = \mathbb{E} \left[ 2 |g(X, U) - Y| - |g(X, U) - g(X, U')|\right].
\end{align*}
For any $g, \tilde{g} \in \mathcal{H}_{\mathtt{nn}}(d+1, N, L, B, M)$, define the following two key quantities.
\begin{align*}
    \Delta_n(g, \tilde{g}) &= \mathsf{R}_{n,\infty}(g) - \mathsf{R}_{n, \infty}(\tilde{g}) - \left(\mathsf{R}_{\infty,\infty}(g) - \mathsf{R}_{\infty, \infty}(\tilde{g})\right)\\
    \chi_n(g)&=\mathsf{R}_{n,m}(g) - \mathsf{R}_{n, \infty}({g}).
\end{align*}

These two quantities separate the two sources of randomness coming from original samples $\{(x_i,y_i)\}_{i=1}^n$ and the generated random variables $\{U_{b,i,k}\}$.
The following two propositions are the key to our main results, where the first uses localization to derive an instance-dependent error bound, and the second unravels the effects of randomly generated random variables.
\begin{proposition}
\label{prop:instant-dependent-error-bound}
    Under Condition \ref{cond:bounded}, the following event
    \begin{align*}
        |\Delta_n(g, \tilde{g})| \le C\left(\delta_{n, t}^2 + \delta_{n, t} \|F_g - F_{\tilde{g}}\|_2\right)\ ,\forall g,\tilde g\qquad \text{with} \qquad \delta_{n, t} = NL \sqrt{\frac{\log n}{n}} + \sqrt{\frac{t}{n}}
    \end{align*} occurs with probability at least $1-2e^{-t}$.
\end{proposition}

\begin{proposition}
\label{prop:chainig-bernstein}
Under Condition \ref{cond:bounded}, the following event 
\begin{align*}
|\chi_n(g)|\le C\left(\frac{
\delta_{n,t}}{\sqrt{m_1m_2}}+\delta^2_{n,t}\right)
\end{align*}
occurs with probability at least $1-2e^{-t}$.

\end{proposition}

We are now ready to prove Theorem \ref{thm:oracle-n-infty}. It follows from the definition of $\mathsf{R}_{\infty, \infty}$ that
\begin{align*}
    \mathsf{R}_{\infty, \infty}(g) - \mathsf{R}_{\infty, \infty}(\tilde{g}) 
    &= \int (F_g(x, t) - F_{\tilde{g}}(x, t))^2 dt\mu_x(dx) \\
    &~~~~~~~~ + 2 \int (F_g(x, t) - F_{\tilde{g}}(x, t))(F_{\tilde{g}}(x, t) - F^\star(x, t)) dt\mu_x(dx) \\
    &\ge 0.5 \int (F_g(x, t) - F_{\tilde{g}}(x, t))^2 dt\mu_x(dx) - 2 \int (F^\star(x, t) - F_{\tilde{g}}(x, t))^2 dt\mu_x(dx) \\
    &= 0.5 \|F_g - F_{\tilde{g}}\|_2^2 - 2 \|F^\star - F_{\tilde{g}}\|_2^2.
\end{align*}
At the same time, it follows from the optimality condition and Proposition \ref{prop:instant-dependent-error-bound} and \ref{prop:chainig-bernstein} that for any $\tilde{g}\in\mathcal{G}$,
\begin{align*}
    \mathsf{R}_{\infty, \infty}(\hat{g}) - \mathsf{R}_{\infty, \infty}(\tilde{g}) &= \mathsf{R}_{\infty, \infty}(\hat{g}) - \mathsf{R}_{n,\infty}(\hat{g}) + \mathsf{R}_{n,\infty}(\hat{g}) - \mathsf{R}_{n,\infty}(\tilde{g}) + \mathsf{R}_{n,\infty}(\tilde{g}) - \mathsf{R}_{\infty, \infty}(\tilde{g}) \\
    &\le |\Delta_n(\hat{g}, \tilde{g})| + \mathsf{R}_{n,\infty}(\hat{g}) - \mathsf{R}_{n,\infty}(\tilde{g}) \\
    &\le |\Delta_n(\hat{g}, \tilde{g})| + \mathsf{R}_{n,m}(\hat{g}) - \mathsf{R}_{n,m}(\tilde{g}) +|\chi_n(\hat g)|+|\chi_n(\tilde g)|\\
    &\le 3C \delta_{n, t}^2 + C^2 \delta_{n, t}^2 + \frac{1}{4} \|F_{\hat{g}} - F_{\tilde{g}}\|_2^2 + \delta_{\mathtt{opt}}+2C\frac{
\delta_{n,t}}{\sqrt{m_1m_2}}.
\end{align*} 
where in the last inequality we used $\delta_{n,t}\|F_{\hat g}-F_{\tilde g}\|_2\le \frac{1}{4C}\|F_{\hat g}-F_{\tilde g}\|_2^2+C\delta_{n,t}^2$. Combining the two inequalities together, we can obtain
\begin{align*}
    0.25 \|F_{\hat{g}} - F_{\tilde{g}}\|_2^2 \le (C^2 + 3C) \delta_{n, t}^2 + 2C\frac{
\delta_{n,t}}{\sqrt{m_1m_2}}+\delta_{\mathtt{opt}} + 2\|F^\star - F_{\tilde{g}}\|_2^2.
\end{align*} We can then conclude the proof using triangle inequality and taking $\tilde{g}$ by minimizing the approximation error.

\subsection{Proof of Proposition \ref{prop:instant-dependent-error-bound}}

Denote $\mathcal{G} = \{g: [0,1]^{d} \to \mathbb{R}\} = \mathcal{H}_{\mathtt{nn}}(d, N, L, \tilde d, M, B)$, and $\log \mathcal{N}(\delta, \mathcal{G}, \|\cdot \|_{\infty})$ to be the logarithmic covering number of $\mathcal{G}$ with respect to $\|\cdot \|_{\infty, [0,1]^{d}}$ norm. 
We first introduce a lemma characterizing the log-covering number of the function class. Recall $M=c_0$ in Condition \ref{cond:bounded}.
\begin{lemma}
\label{lemma:covering-number-G}
There exists some constant $C_1$ such that for $\delta<1$,
\begin{align*}
\log \mathcal{N}(\delta,\mathcal{G},\|\cdot\|_\infty)\le C_1 N^2L^2\log(\frac{(BN+2)L}{\delta}).
\end{align*}    
\end{lemma}

\begin{proof}
We first claim a continuity-type result, similar to Lemma 8 in \cite{fan2022factor}.
\begin{claim}
\label{claim:continuity-of-parameters}
Let $\theta(g)=\{(W_l,b_l)_{l=1}^{L+1}\}$ as the set of parameters for $g\in \mathcal{G}$, $\theta(\breve g)=\{(\breve W_l,\breve b_l)_{l=1}^{L+1}\}$ as the set of parameters for $\breve g\in \mathcal{G}$. Define 
\begin{align*}
\|\theta(g)-\theta(\breve g)\|_{\infty}=\max_{1\le l\le L+1}\left(\|b_l-\breve b_l\|_{\max}\vee\|W_l-\breve W_l\|_{\max}\right)
\end{align*}
Then we have the following holds
\begin{align*}
\|g(x,u)-\breve g(x,u)\|_{\infty}\le 2N(L+1)(BN+2)^{L+1}\|\theta(g)-\theta(\breve g)\|_{\infty}
\end{align*}
\end{claim}
\begin{proof}[Proof of Claim \ref{claim:continuity-of-parameters}]
To streamline the proof, define $T^{\star}_l(x,u_l,\cdots,u_L)=(T_l(x,u_l),u_{l+1},\cdots,u_L)$ for $1\le l\le L$, and $T^{\star}_L(x,u_L)=T_L(x,u_L)$. Also define $\bar\sigma^{\star}(x,u_{l+1},\cdots,u_L)=(\bar \sigma(x),u_{l+1},\cdots,u_L)$. we have \begin{align*}
g(x,u)=T^{\star}_{L+1}\circ\bar \sigma^{\star}_L\circ T^{\star}_{L}\circ\cdots\circ \bar \sigma^{\star}_1\circ T^{\star}_{1}(x,u). 
\end{align*}
Using a similar definition, it holds that 
\begin{align*}
\breve g(x,u)=\breve T^{\star}_{L+1}\circ\bar \sigma^{\star}_L\circ \breve T^{\star}_{L}\circ\cdots\circ \bar \sigma^{\star}_1\circ \breve T^{\star}_{1}(x,u). 
\end{align*}
Now define 
\begin{align*}
g^{(l)}_-&=T^{\star}_{L+1}\circ\bar \sigma^{\star}_L\circ T^{\star}_{L}\circ\cdots\circ \bar \sigma^{\star}_l\circ T^{\star}_{l}\\
g^{(l)}_+&=\bar \sigma^{\star}_{l-1}\circ T^{\star}_{l-1}\circ\cdots\circ \bar \sigma^{\star}_1\circ T^{\star}_{1}.
\end{align*}
We have $g=g^{(l)}_-\circ g^{(l)}_+$ and the following decomposition
\begin{align}
\label{equ:covering-number-decomposition}
|g(x,u)-\breve g(x,u)|\le |(T^{\star}_{L+1}-\breve T^{\star}_{L+1})\circ g^{(L+1)}_+(x,u)|+\sum_{l=1}^L|\breve g^{(l+1)}_-\circ \bar \sigma^{\star}_l\circ(T_l^{\star}-\breve T_l^{\star})\circ g^{(l)}_+(x,u)|.
\end{align}
Next, we bound $\|g^{(l)}_+\|_\infty$. Recursively, it is straightforward to show that
\begin{align*}
\|g^{(l)}_+\|_\infty&\le \| T^{\star}_{l-1}\circ\cdots\circ \bar \sigma^{\star}_1\circ T^{\star}_{1}\|_\infty\\&\le \|g^{(l-1)}_+\|_\infty+\|T_{l-1}I(g^{(l-1)}_+)\|_\infty\\&\le \|g_+^{(l-1)}\|_\infty+B+BN\|I(g_+^{(l-1)})\|_\infty\\&\le \|g_+^{(l-1)}\|_\infty+B+BN\|g_+^{(l-1)}\|_\infty\\&=B+(BN+1)\|g_+^{(l-1)}\|_\infty
\end{align*}
where we denote $I(g_+^{(l-1)})$ as the sub-indexing function that picks the first two blocks of variables of $g_+^{(l-1)}$. By induction, we have the bound
\begin{align*}
\|g_+^{(l)}\|_{\infty}\le (BN+2)^l
\end{align*}
Further, we consider $\|g_-^{(l)}(x_1,u_1)-g_-^{(l)}(x_2,u_2)\|_{\infty}$. 
We use induction again. If $\|g_-^{(l+1)}(x_1,u_1)-g_-^{(l+1)}(x_2,u_2)\|_{\infty}\le L_{l+1}\|(x_1,u_1)-(x_2,u_2)\|_\infty$,
\begin{align*}
|g^{(l)}_-(x_1,u_1)-g^{(l)}_-(x_2,u_2)|&=\Big|g^{(l+1)}_-\circ\bar \sigma^{\star}_l\circ T^{\star}_{l}(x_1,u_1)-g^{(l+1)}_-\circ\bar \sigma^{\star}_l\circ T^{\star}_{l}(x_2,u_2)\Big|\\&\le L_{l+1}\|\bar \sigma^{\star}_l\circ T^{\star}_{l}(x_1,u_1)-\bar \sigma^{\star}_l\circ T^{\star}_{l}(x_2,u_2)\|_\infty\\&
\le L_{l+1}\|T^{\star}_{l}(x_1,u_1)-T^{\star}_{l}(x_2,u_2)\|_\infty\\&
\le L_{l+1}\|T_{l}(x_1,u_1)-T_{l}(x_2,u_2)\|_\infty+\|(x_1,u_1)-(x_2,u_2)\|_\infty\\&\le L_{l+1}(BN+1)\|(x_1,u_1)-(x_2,u_2)\|_\infty.
\end{align*}
Hence we can take $L_l=L_{l+1}(BN+1)$ and $L_l=(BN+1)^{L-l+2}$, as $L_{L+2}=1$.

Now we go back to \eqref{equ:covering-number-decomposition}, it follows that 
\begin{align}
|(T^{\star}_{L+1}-\breve T^{\star}_{L+1})\circ g^{(L+1)}_+(x,u)|&\le \|W_{L+1}-\breve W_{L+1}\|_{\infty}\|g^{(L+1)}_+(x,u)\|_{\infty}+\|b_{L+1}-\breve b_{L+1}\|_{\infty}\nonumber\\
&\le \|\theta(g)-\theta(\breve g)\|_{\infty}(BN+2)^{L+1}+\|\theta(g)-\theta(\breve g)\|_{\infty}\nonumber\\&\le 2(BN+2)^{L+1}\|\theta(g)-\theta(\breve g)\|_{\infty}\label{equ:covering-number-intermediate-1}
\end{align}
and
\begin{align}
|\breve g^{(l+1)}_-\circ \bar \sigma^{\star}_l\circ(T_l^{\star}-\breve T_l^{\star})\circ g^{(l)}_+(x,u)|&\le (BN+1)^{L-l+1}|\bar \sigma^{\star}_l\circ(T_l^{\star}-\breve T_l^{\star})\circ g^{(l)}_+(x,u)|\nonumber\\&\le (BN+1)^{L-l+1}|(T_l^{\star}-\breve T_l^{\star})\circ g^{(l)}_+(x,u)|\nonumber\\&\le (BN+1)^{L-l+1}\left(d_{l+1}\|W_l-\breve W_l\|_{\infty}\|g^{(l)}_+\|_\infty+\|b_l-\breve b_l\|_{\infty}\right)\nonumber\\&\le 
(BN+1)^{L-l+1}\left(N\|\theta(g)-\theta(\breve g)\|_{\infty}(BN+2)^l+\|\theta(g)-\theta(\breve g)\|_{\infty}\right)\nonumber\\&\le 2N(BN+2)^{L+1}\|\theta(g)-\theta(\breve g)\|_{\infty}.
\label{equ:covering-number-intermediate-2}
\end{align}
Combining \eqref{equ:covering-number-intermediate-1} and \eqref{equ:covering-number-intermediate-2}, it holds that
\begin{align*}
|g(x,u)-\breve g(x,u)|\le 2N(L+1)(BN+2)^{L+1}\|\theta(g)-\theta(\breve g)\|_{\infty}.
\end{align*}

\end{proof}

The remaining proof is similar in vein to Lemma 7 in \cite{fan2022factor}.
First, we construct $\delta$-set.
Let $\mathcal{G}_{\delta}=\Big\{g\in \mathcal{G}:\theta(g)=\{(W_l,b_l)_{l=1}^{L+1}\},[W_l]_{i,j},[b_l]_{i}\in \{-B,-B+\epsilon,\cdots,-B+\epsilon\lceil\frac{2B}{\epsilon}\rceil\}\Big\}$, where $\epsilon:=\frac{\delta}{2N(L+1)(BN+2)^{L+1}}$. By construction, for any $g\in\mathcal{G}$, there exists a $\tilde g\in \mathcal{G}_{\delta}$ such that $\|\theta(g)-\theta(\tilde g)\|_\infty\le \epsilon$, and by Claim \ref{claim:continuity-of-parameters} we have 
\begin{align*}
\|g-\tilde g\|_{\infty}\le 2N(L+1)(BN+2)^{L+1}\epsilon\le \delta
\end{align*}
Hence $\mathcal{G}_\delta$ is indeed a $\delta$-set. It suffices to calculate the cardinality $|\mathcal{G}_\delta|$.
And it is immediate that 
$|\mathcal{G}_\delta|\le (\frac{2B}{\epsilon}+1)^{(N^2+N)(L+1)}$, which implies that
\begin{align*}
\log \mathcal{N}(\delta,\mathcal{G},\|\cdot\|_\infty)&\le \log |\mathcal{G}_\delta|\\&\le N(N+1)(L+1)\log \Big(\frac{4BN(L+1)(BN+2)^{L+1}}{\delta}+1\Big)\\&\lesssim N^2L^2\log(\frac{(BN+2)L}{\delta}).
\end{align*}

The proof is thus completed.
\end{proof}

The following lemma converts the population-level energy distance loss to an integral.
\begin{lemma}
\label{lemma:energy-loss-to-integral}
For given $x, y$, if $\mathbb{E}[|g(x, U)|] < \infty$, then
\begin{align*}
    2\mathbb{E}_U\left[|y - g(x, U)|\right] - \mathbb{E}_{U,U'}[|g(x, U) - g(x, U')|] = 2\int \left( F_g(x, t) - \indicator{y\le t} \right)^2 dt
\end{align*}
\end{lemma}
\begin{proof}[Proof of Lemma~\ref{lemma:energy-loss-to-integral}] The proof is similar to Section \ref{proof:identity} and is thus omitted.
\end{proof}

It follows from Lemma \ref{lemma:energy-loss-to-integral} that 
\begin{align*}
    \frac{1}{2} \Delta_n(g, \tilde{g}) &= \frac{1}{n} \sum_{i=1}^n \int (F_g(X_i, t) - \indicator{Y_i \le t} )^2 - (F_{\tilde{g}}(X_i, t) - \indicator{Y_i \le t} )^2 dt \\
    &~~~~~~ - \mathbb{E} \left[ \int (F_g(X, t) - \indicator{Y \le t} )^2 - (F_{\tilde{g}}(X, t) - \indicator{Y \le t} )^2 \right] \\
    &= \left(\mathbb{P}_n - \mathbb{P} \right) \int (F_g(X, t) - F_{\tilde{g}}(X, t))^2 dt \\
    &~~~~~~~ + 2 \left(\mathbb{P}_n - \mathbb{P} \right)\int (F_g(X, t) - F_{\tilde{g}}(X, t))(F_{\tilde{g}}(X, t) - \indicator{Y_i \le t}) dt\\
    &:= \mathsf{T}_1(g, \tilde{g}) + \mathsf{T}_2(g, \tilde{g}).
\end{align*} 

We will establish instance-dependent error bounds on $\mathsf{T}_1$ and $\mathsf{T}_2$ respectively. 

\noindent\textbf{Instance-Dependent Error Bound on $\mathsf{T}_1$}

Now we define $\mathcal{V}(\mathcal{G})$ as 
\begin{align*}
    \mathcal{V}(\mathcal{G}) = \left\{v_{g, \tilde{g}}(x) = {\int \left(F_g(x, t) - F_{\tilde{g}}(x, t)\right)^2 dt}: g, \tilde{g} \in \mathcal{G}\right\}
\end{align*} and let $\log \mathcal{N}(\epsilon, \mathcal{V}(\mathcal{G}), \|\cdot \|_{\infty})$ be the logarithmic covering number of $\mathcal{V}(\mathcal{G})$ with respect to $\|\cdot\|_{\infty, [0,1]^d}$ norm. The following lemma characterizes the relationship between $\log \mathcal{N}(\epsilon, \mathcal{V}(\mathcal{G}), \|\cdot \|_{\infty})$ and $\log \mathcal{N}(\epsilon, \mathcal{G}, \|\cdot \|_{\infty})$. 

\begin{lemma}
\label{lemma:log-cover-number-conversion}
We have
\begin{align}
\label{equ:covering-relation-V(G)-G}
\log \mathcal{N}(\epsilon, \mathcal{V}(\mathcal{G}), \|\cdot \|_{\infty}) \le 2 \log \mathcal{N}\left((\epsilon/17)^2/c_0, \mathcal{G}, \|\cdot \|_{\infty}\right)
\end{align}
\end{lemma}
\begin{proof}[Proof of Lemma~\ref{lemma:log-cover-number-conversion}] See Section \ref{sec:proof-log-cover-number-coversion}.
\end{proof}

For any $v\in \mathcal{V}(\mathcal{G})$, we let
\begin{align*}
    \|v\|_\infty = \sup_{x\in [0,1]^d} |v(x)|,\qquad
    \|v\|_2 = \sqrt{\int v^2(x) \mu_x(dx)},\qquad \|v\|_n = \sqrt{\frac{1}{n} \sum_{i=1}^n v^2(X_i)}
\end{align*} It is easy to see that $\|v\|_\infty \le 2c_0$. 

Observe that now we have
\begin{align*}
    \mathsf{T}_1(g, \tilde{g}) = \frac{1}{n} \sum_{i=1}^n v_{g, {\tilde{g}}}(X_i) - \mathbb{E} [v_{g, {\tilde{g}}}(X)]
\end{align*} with bounded $\|v\|_\infty \le 2c_0$. Then we can show the following claim, which directly gives the bound on $\mathsf{T}_1$.

\begin{claim}
\label{claim:standard-localization}
\begin{align*}
    \forall v_{g, \tilde{g}} \in \mathcal{V}(\mathcal{G}), \qquad \left|\frac{1}{n} \sum_{i=1}^n v_{g, {\tilde{g}}}(X_i) - \mathbb{E} [v_{g, {\tilde{g}}}(X)] \right| &\le C \left(\delta_{n, t}^2 + \delta_{n, t} \|v_{g, \tilde{g}}\|_2 \right) 
\end{align*}
with probability at least $1-e^{-t}$. 
\end{claim}

In fact, from Lemma \ref{lemma:covering-number-G}, \eqref{equ:covering-relation-V(G)-G} and $\log B,\log(NL)\lesssim \log n$ in Condition \ref{cond:bounded}, we have that for some constant $C_2$ and every $\epsilon\le 1$ that  
\[\log \mathcal{N}(\epsilon, \mathcal{V}(\mathcal{G}), \|\cdot \|_{\infty}) \le C_2N^2L^2\log(\frac{2n}{\epsilon}).\]

Now we apply Theorem 19.3 in \cite{gyorfi2002distribution} on the function class $\mathcal{H}=\{h=v^2:v\in \mathcal{V}(\mathcal{G})\}$. \footnote{One can also directly apply Lemma 17 in \cite{chai2025deep}.} First note that $F$ is supported on $[-1,1]$, and therefore by the Condition \ref{cond:bounded}, for every $h=v^2\in \mathcal{H}$, we have $h(x)\le 4c_0^2$ and $\mathbb{E}[h(x)^2]\le 4c_0^2\mathbb{E}[h(x)]$.

Following the same notations as Theorem 19.3 in \cite{gyorfi2002distribution}, we choose $\epsilon=\frac{1}{2}$ and $\alpha\ge K_1\frac{N^2L^2\log n}{n}$ for some large constant $K_1$, and for any fixed $x_1^n$, we can bound the log-covering number in the empirical 2-norm $\|\cdot\|_n$ at $x_1^n$ of $\mathcal{H}$ as follows,

\begin{align*}
\log \mathcal{N}\left(u,\left\{h\in \mathcal{H},\frac{1}{n}\sum_{i=1}^n h^2(x_i)\le 16\delta\right\},\|\cdot\|_n\right)&\le \log \mathcal{N}\left(u,\mathcal{H},\|\cdot\|_\infty\right)\\&\le \log \mathcal{N}\left(\frac{u}{2c_0},\mathcal{V}(\mathcal{G}),\|\cdot\|_\infty\right)\\
&\le C_2N^2L^2\log(\frac{2n}{u})
\end{align*}where in the second inequality we used uniform boundedness of $\mathcal{V}(\mathcal{G})$ which implies any $\epsilon$-net of $\mathcal{V}(\mathcal{G})$ is also an $2c_0\epsilon$-net of  $\mathcal{H}$.

The condition in Theorem 19.3 in \cite{gyorfi2002distribution} is then guaranteed as for any $\delta\ge \alpha/8$, it holds that
\begin{align*}
\int_{\delta/512}^{\sqrt{\delta}}\sqrt{\log\mathcal{N}\left(u,\left\{h\in \mathcal{H},\frac{1}{n}\sum_{i=1}^n h^2(x_i)\le 16\delta\right\},\|\cdot\|_n\right)}du&\le \sqrt{\delta N^2L^2 \log(\frac{2n\cdot 512}{\delta})}\\
&\le \frac{\sqrt{n}\delta}{96\cdot4\cdot 8\sqrt{2}}.
\end{align*}

Therefore, we have by (19.12) in \cite{gyorfi2002distribution}, for some constants $C_3,C_4$,
\[\mathbb{P}\left[\sup_{h\in \mathcal{H}}\frac{\mathbb{E}[h(X)]-\frac{1}{n}\sum_{i=1}^n h(x_i)}{\alpha+\mathbb{E}[h(X)]}>\frac{1}{2}\right]\le C_3\exp{\left(-\frac{n\alpha}{C_4}\right)}.\]
With a change of variables $\alpha=t^2$, we arrive at the conclusion that with probability at least $1-\exp(-cnt^2)$,
\begin{equation}
\label{equ:L2-Ln-relation}\forall v\in \mathcal{V}(\mathcal{G}),\quad |\|v\|_2^2-\|v\|_n^2|\le \frac{1}{2}(\|v\|_2^2+t^2)\end{equation}
as long as $t\ge C\sqrt{\frac{N^2L^2\log n}{n}}$.

Now by the standard symmetrization technique, we have 
\begin{align*}
\mathbb{E}\sup_{v_{g,\tilde g}\in \mathcal{B}(r,\mathcal{V}(\mathcal{G}))}\left|\frac{1}{n} \sum_{i=1}^n v_{g, {\tilde{g}}}(X_i) - \mathbb{E} [v_{g, {\tilde{g}}}(X)] \right|\le \mathbb{E}_X\mathbb{E}_\epsilon\sup_{v_{g,\tilde g}\in \mathcal{B}(r,\mathcal{V}(\mathcal{G}))}\left|\frac{1}{n}\sum_{i=1}^n \epsilon_i v_{g,\tilde g}(X_i)\right|
\end{align*}
where $\epsilon_i$ are i.i.d. Rademacher random variables.

Now we can conduct the chaining technique. Indeed, denote $\mathcal{B}(r,\mathcal{V}(\mathcal{G}))=\{v\in \mathcal{V}(\mathcal{G}):\|v\|_2\le r\}$ and $\mathcal{B}_n(r,\mathcal{V}(\mathcal{G}),x_1^n)=\{v\in \mathcal{V}(\mathcal{G}):\|v\|_n\le r\}$. We have by 
\eqref{equ:L2-Ln-relation} that
with probability at least $1-\exp(-cnt^2)$, $\mathcal{B}(r,\mathcal{V}(\mathcal{G}))\subseteq \mathcal{B}_n(2r+t,\mathcal{V}(\mathcal{G}),x_1^n)$, as long as $t\ge C\sqrt{\frac{N^2L^2\log n}{n}}$. Denote this event as $\mathcal{A}_1$.

It follows that under $\mathcal{A}_1$,
\begin{align*}
\mathbb{E}_X\mathbb{E}_\epsilon\sup_{v_{g,\tilde g}\in \mathcal{B}(r,\mathcal{V}(\mathcal{G}))}\left|\frac{1}{n}\sum_{i=1}^n \epsilon_i v_{g,\tilde g}(X_i)\right|&\lesssim \frac{1}{\sqrt{n}}\mathbb{E}\int_0^2 \sqrt{\log \mathcal{N}_n (\epsilon,\mathbb{B}(r,\mathcal{V}(\mathcal{G})),x_1^n)d\epsilon}\\&\lesssim \frac{1}{\sqrt{n}}\mathbb{E}\int_0^{2r+t} \sqrt{\log \mathcal{N}_n (\epsilon,\mathcal{B}_n(2r+t,\mathcal{V}(\mathcal{G}),x_1^n),x_1^n)}d\epsilon\\
&\lesssim \frac{1}{\sqrt{n}}\mathbb{E}\int_0^{2r+t} \sqrt{\log \mathcal{N} (\epsilon,\mathcal{V}(\mathcal{G}),\|\cdot\|_\infty}d\epsilon\\
&\lesssim \frac{1}{\sqrt{n}}\int_0^{2r+t} \sqrt{N^2L^2 \log(\frac{4n}{\epsilon})}d\epsilon\\&\le \frac{1}{\sqrt{n}} \sqrt{N^2L^2 \int_0^{2r+t}\log(\frac{4n}{\epsilon})}d\epsilon
\\
&\lesssim (2r+t)\sqrt{\frac{N^2L^2\log n}{n}}
\end{align*}
where the last inequality follows from $\int_0^x \log(\frac{1}{\epsilon})d\epsilon\asymp x\log(1/x)$ when $0\le x\ll 1$ and $t\gtrsim \sqrt{\frac{N^2L^2\log n}{n}}$.

This further leads to for any $t>0$, with probability at least $1-\exp(-cnt^2)$, 
\[\mathbb{E}_X\mathbb{E}_\epsilon\sup_{v_{g,\tilde g}\in \mathcal{B}(r,\mathcal{V}(\mathcal{G}))}\left|\frac{1}{n}\sum_{i=1}^n \epsilon_i v_{g,\tilde g}(X_i)\right|\lesssim (2r+t)\sqrt{\frac{N^2L^2\log n}{n}}+\frac{N^2L^2\log n}{n}.\]

Next we apply Talagrand's concentration Inequality (see e.g. Theorem 3.27 in \cite{wainwright2019high}). Let $Z=\sup_{v\in \mathcal{B}(r,\mathcal{V}(\mathcal{G}))}\left|\frac{1}{n} \sum_{i=1}^n v(X_i) - \mathbb{E} [v(X)] \right|$. We have that with probability at least $1-\exp(-nt)$,
\begin{align*}
Z-\mathbb{E}[Z]\le \sqrt{2}\sqrt{(\sigma^2+2b\mathbb{E}[Z])t}+\frac{1}{3}bt,
\end{align*}
where $b=\sup_{v\in\mathcal{B}(r,\mathcal{V}(\mathcal{G}))}\|v\|_\infty\le 2$ and $\sigma=\sup_{v\in\mathcal{B}(r,\mathcal{V}(\mathcal{G}))}\sqrt{\mathbb{E}\left|v-\mathbb{E}[v]\right|^2}\le \sup_{v\in\mathcal{B}(r,\mathcal{V}(\mathcal{G}))}\sqrt{\mathbb{E}[v^2]}=r$. Hence, there exist some constants $c_1,c_2$ such that with probability at least $1-\exp(-nt)$,
\[Z\le 2\mathbb{E}[Z]+c_1r\sqrt{t}+c_2t.\]
This combined with the previous bound on $\mathbb{E}_X\mathbb{E}_\epsilon\left|\frac{1}{n}\sum_{i=1}^n \epsilon_i v_{g,\tilde g}(X_i)\right|$ and the conditioning on $\mathcal{A}_1$, implies that with probability at least $1-\exp(-cnt^2)$,
\begin{align*}
\sup_{v\in \mathcal{B}(r,\mathcal{V}(\mathcal{G}))}\left|\frac{1}{n} \sum_{i=1}^n v(X_i) - \mathbb{E} [v(X)] \right|\le C'\left((2r+t)\sqrt{\frac{N^2L^2\log n}{n}}+\frac{N^2L^2\log n}{n}+rt+t^2\right)
\end{align*}
for some constants $c,C'$.

By a change of variable, we get
with probability at least $1-\exp(-t)$,
\begin{align*}
\sup_{v\in \mathcal{B}(r,\mathcal{V}(\mathcal{G}))}\left|\frac{1}{n} \sum_{i=1}^n v(X_i) - \mathbb{E} [v(X)] \right|\le C\left(\Big(r+\sqrt{\frac{t}{n}}\Big)\delta_n+\delta_n^2+r\sqrt{\frac{t}{n}}+\frac{t}{n}\right)
\end{align*}where we abbreviate $\delta_n=\sqrt{\frac{N^2L^2\log n}{n}}$.

The final step is the standard peeling technique. More specifically, we denote $\mathcal{S}_l=\{v\in\mathcal{V}(\mathcal{G}):\alpha_{l-1}\delta_n\le\|v\|_2\le \alpha_l\delta_n\}$ for $0\le l\le L:=\lceil\log(2/\delta_n)\rceil$, where $\alpha_{-1}=0$ and $\alpha_l=2^l$ for $l>0$.

For every $l\ge 0$, it holds with probability at least $1-\exp(-t-l)$ that,
\begin{align*}
\sup_{v\in \mathcal{S}_l}\left|\frac{1}{n} \sum_{i=1}^n v(X_i) - \mathbb{E} [v(X)] \right|\le C\left(\Big(\alpha_l\delta_n+\sqrt{\frac{t+l}{n}}\Big)\delta_n+\delta_n^2+\alpha_l\delta_n\sqrt{\frac{t+l}{n}}+\frac{t+l}{n}\right).
\end{align*}
Denote this event as $\mathcal{B}_l$. Note that for $1\le l\le L$, $\|v\|_2\ge \alpha_{l-1}\delta_n\ge \alpha_{l}\delta_n/2$, hence for $1\le l\le L$,
under $\mathcal{B}_l$, it holds that for every $v\in\mathcal{S}_l$,
\begin{align*}
\left|\frac{1}{n} \sum_{i=1}^n v(X_i) - \mathbb{E} [v(X)] \right|&\le C\left(\Big(2\|v\|_2+\sqrt{\frac{t+l}{n}}\Big)\delta_n+\delta_n^2+2\|v\|_2\sqrt{\frac{t+l}{n}}+\frac{t+l}{n}\right)\\
&\lesssim \|v\|_2\delta_n+\sqrt{\frac{t+L}{n}}\delta_n+\delta_n^2+\|v\|_2\sqrt{\frac{t+L}{n}}+\frac{t+L}{n}\\
&\lesssim \|v\|_2\delta_n+\left(\sqrt{\frac{t}{n}}+\delta_n\right)\delta_n+\delta_n^2+\|v\|_2\left(\sqrt{\frac{t}{n}}+\delta_n\right)+\frac{t}{n}+\delta_n^2\\
&\lesssim \|v\|_2\left(\sqrt{\frac{t}{n}}+\delta_n\right)+\left(\sqrt{\frac{t}{n}}+\delta_n\right)^2
\end{align*}
where in the third inequality we use $L=\lceil\log(2/\delta_n)\rceil\le \log n$ and $\sqrt{\frac{L}{n}}\le \sqrt{\frac{\log n}{n}}\le \delta_n$.

For $l=0$, under $\mathcal{B}_0$, we have for any $v\in\mathcal{S}_0$, 
\begin{align*}
\left|\frac{1}{n} \sum_{i=1}^n v(X_i) - \mathbb{E} [v(X)] \right|&\le C\left(\Big(\delta_n+\sqrt{\frac{t+l}{n}}\Big)\delta_n+\delta_n^2+\delta_n\sqrt{\frac{t+l}{n}}+\frac{t+l}{n}\right)\\
&\lesssim \delta_n^2+\left(\sqrt{\frac{t}{n}}+\delta_n\right)\delta_n+\delta_n^2+\delta_n\left(\sqrt{\frac{t}{n}}+\delta_n\right)+\frac{t}{n}+\delta_n^2\\
&\lesssim \left(\sqrt{\frac{t}{n}}+\delta_n\right)^2\le \|v\|_2\left(\sqrt{\frac{t}{n}}+\delta_n\right)+\left(\sqrt{\frac{t}{n}}+\delta_n\right)^2.
\end{align*}

Note that $\mathcal{V}(\mathcal{G})=\cup_{l=0}^L \mathcal{S}_l$.
Recall that $\delta_{n, t} = NL \sqrt{\frac{\log n}{n}} + \sqrt{\frac{t}{n}}$. Therefore, under $\cap_{l=0}^L\mathcal{B}_l$, we have that
\[\left|\frac{1}{n} \sum_{i=1}^n v(X_i) - \mathbb{E} [v(X)] \right|\lesssim \delta_{n,t}^2+\delta_{n,t}\|v\|_2,\qquad \forall v\in \mathcal{V}(\mathcal{G})\]
It boils down to calculate $\mathbb{P}(\cap_{l=0}^L\mathcal{B}_l)=1-\mathbb{P}(\cup_{l=0}^L\mathcal{B}_l^c)$. We have by union bound, 
\begin{align*}
\mathbb{P}(\cup_{l=0}^L\mathcal{B}_l^c)&\le \sum_{l=0}^L \mathbb{P}(\mathcal{B}_l^c)\\
&\le \sum_{l=0}^L \exp(-t-l)\le 2\exp(-t).
\end{align*}

Therefore, Claim \ref{claim:standard-localization} is proved.
At the same time, we have
\begin{align*}
    \|v_{g, \tilde{g}}\|_2^2 = \int \left(\int \left(F_g(x, t) - F_{\tilde{g}}(x, t)\right)^2 dt\right)^2 \mu_x(dx) \le 2 \int \left(F_g(x, t) - F_{\tilde{g}}(x, t)\right)^2 dt \mu_x(dx).
\end{align*} Hence, we can conclude that
\begin{align}
\label{eq:instance-dependent-b1}
    \forall g, \tilde{g} \in \mathcal{G}, \qquad \left|\mathsf{T}_1(g, \tilde{g}) \right| \le C (\delta_{n, t}^2 + \delta_{n, t} \|F_g - F_{\tilde{g}}\|_2).
\end{align}

\noindent\textbf{Instance-Dependent Error Bound on $\mathsf{T}_2$}

We also define $\mathcal{U}(\mathcal{G})$ as
\begin{align*}
    \mathcal{U}(\mathcal{G}) = \left\{u_{g,\tilde{g}}(x, y) = \int (F_g(x, t) - F_{\tilde{g}}(x, t)) (F_{\tilde{g}}(x, t) - \indicator{y \le t}) dt: g,\tilde{g} \in \mathcal{G}\right\}
\end{align*} and let $\log \mathcal{N}(\epsilon, \mathcal{U}(\mathcal{G}), \|\cdot\|_\infty)$ be the logarithmic covering number of $\mathcal{U}(\mathcal{G})$ with respect to $\|\cdot\|_{\infty, [0,1]^d\times \mathbb{R}}$ norm. The following lemma characterizes the relationship between $\log \mathcal{N}(\epsilon, \mathcal{U}(\mathcal{G}), \|\cdot \|_{\infty})$ and $\log \mathcal{N}(\epsilon, \mathcal{G}, \|\cdot \|_{\infty})$.

\begin{lemma}
\label{lemma:log-cover-number-conversion-2}
We have
\begin{align*}
\log \mathcal{N}(\epsilon, \mathcal{U}(\mathcal{G}), \|\cdot \|_{\infty}) \le 2 \log \mathcal{N}\left((\epsilon/20)^2/c_0, \mathcal{G}, \|\cdot \|_{\infty}\right)
\end{align*}
\end{lemma}
\begin{proof}[Proof of Lemma~\ref{lemma:log-cover-number-conversion-2}] The proof is similar to that in \ref{sec:proof-log-cover-number-coversion}. See Section \ref{sec:proof-log-cover-number-coversion-2}.
\end{proof}
For any $u\in \mathcal{U}(\mathcal{G})$, we let
\begin{align*}
    \|u\|_\infty = \sup_{x\in [0,1]^d, y\in \mathbb{R}} |u(x, y)|, \qquad \|u\|_2 = \sqrt{\int u^2(x, y) \mu(dx, dy)}, \qquad \|u\|_n = \sqrt{\frac{1}{n} \sum_{i=1}^n u^2(X_i, Y_i)}
\end{align*} It is also easy to see that $\|u\|_\infty \le 2c_0$ by Cauchy-Schwarz inequality. Now we can write
\begin{align*}
    \mathsf{T}_2(g, \tilde{g}) = \frac{1}{n} \sum_{i=1}^n u_{g,\tilde{g}}(X_i, Y_i) - \mathbb{E}[u_{g,\tilde{g}}(X,Y)].
\end{align*} Following the same procedure as $\mathsf{T}_1(g,\tilde g)$, we can show that
\begin{align*}
    \forall u_{g, \tilde{g}} \in \mathcal{U}(\mathcal{G}), \qquad \left|\frac{1}{n} \sum_{i=1}^n u_{g, {\tilde{g}}}(X_i) - \mathbb{E} [u_{g, {\tilde{g}}}(X)] \right| &\le C \left(\delta_{n, t}^2 + \delta_{n, t} \|u_{g, \tilde{g}}\|_2 \right)
\end{align*} with probability at least $1-e^{-t}$. Meanwhile, it follows from the Cauchy-Schwarz inequality that 
\begin{align*}
    \|u_{g,\tilde{g}}\|_2^2 &= \int \left(\int \left(F_g(x, t) - F_{\tilde{g}}(x, t)\right) \left(F_{\tilde{g}}(x, t) - \indicator{y\le t}\right)  dt\right)^2 \mu(dx, dy) \\
    &\le \int \left(\int \left(F_g(x, t) - F_{\tilde{g}}(x, t)\right)^2 dt \right)\left(\int \left(F_{\tilde{g}}(x, t) - \indicator{y\le t}\right)^2 dt\right)\mu(dx, dy) \\
    &\le 2c_0 \|F_g - F_{\tilde{g}}\|_2^2.
\end{align*} Therefore, we can conclude that
\begin{align}
\label{eq:instance-dependent-b2}
    \forall g, \tilde{g} \in \mathcal{G}, \qquad \left|\mathsf{T}_2(g, \tilde{g}) \right| \le C (\delta_{n, t}^2 + \delta_{n, t} \|F_g - F_{\tilde{g}}\|_2).
\end{align}

\noindent\textbf{Conclusion.} Combining \eqref{eq:instance-dependent-b1} and \eqref{eq:instance-dependent-b2} concludes the proof.

\subsection{Proof of Proposition~\ref{prop:chainig-bernstein}}
\label{sec:proof-chaining-bernstein}
Recall the definition of $\mathsf R_{n,m}(g)$ as $\mathsf R_{n,m}(g)=\frac{1}{n}\sum_{i=1}^n \tilde Z_i(g)$, where we define \begin{align*}
\tilde Z_i(g)=\frac{1}{m_1} \sum_{b=1}^{m_1} \left(\frac{1}{m_2} \sum_{k=1}^{m_2} 2|Y_i - g(X_i, U_{b,i,k})| - \frac{1}{m_2(m_2-1)}\sum_{k\neq k'} |g(X_i, U_{b, i, k}) - g(X_i, U_{b, i, k'})|\right).\end{align*}

Then we have $\mathsf R_{n,\infty}(g)=\frac{1}{n}\sum_{i=1}^n \EE(\tilde Z_i(g)|X_i,Y_i)$, which leads to 
\[\chi(g)=\frac{1}{n}\sum_{i=1}^n \left(\tilde Z_i(g)-\EE(\tilde Z_i(g)|X_i,Y_i)\right)=\frac{1}{n}\sum_{i=1}^n Z_i(g)\]where we let $Z_i(g):=\tilde Z_i(g)-\EE(\tilde Z_i(g)|X_i,Y_i)$.

It suffices to prove $\sup_{g\in\mathcal{G}}\chi(g)\le C(\sqrt{\frac{\delta_{n,t}}{m_1m_2}}+\delta^2_{n,t})$. To achieve the goal, we again resort to the chaining method. Construct a sequence of
$\epsilon$-nets with decreasing scale in $\|\cdot\|_\infty$-norm.
To be more specific, let $\mathcal{S}_k$ be the $2^{-k}$-net of $\mathcal{G}$ in $\|\cdot\|_\infty$-norm. By Lemma \ref{lemma:covering-number-G}, we have that the cardinality of $\mathcal{S}_k$ is bounded by $N_k\le C_1(N^2L^2\log n+k\log 2)$. Also denote $\pi_k(g)=\arg\inf_{f\in \mathcal{S}_k}\|g-f\|_\infty$ as its closest neighbor in $\mathcal{S}_k$.

Note that for an integer $M$ to be specified later, we have 
\[\chi(g)=\chi(\pi_M(g))+\sum_{i=M}^\infty \left[\chi(\pi_{i+1}(g))-\chi(\pi_i(g))\right]\] and thus
\begin{equation}
\label{equ:chaining-decomposition}\sup_{g\in \mathcal{G}}\chi(g)=\sup_{g\in \mathcal{G}}\chi(\pi_M(g))+\sum_{i=M}^\infty \sup_{g\in \mathcal{G}}\left[\chi(\pi_{i+1}(g))-\chi(\pi_i(g))\right]\end{equation}
where we used the separability of $\mathcal{G}$.

We start off by bounding $\sup_{g\in \mathcal{G}}\chi(\pi_M(g))$ using the Bernstein's Inequality. 
To that aim, we need the following claim.
\begin{claim}
For a fixed $g\in \mathcal{G}$ and fixed $i$, with probability at least $4\exp(-\frac{m_1m_2t^2}{50})$, 
it holds that 
\[|Z_i(g)|\le t.\]
Furthermore, we have $\mathsf{Var}(Z_i(g))\le \frac{56}{m_1m_2}$.
\end{claim}

\begin{proof}
We fix $X_i$ and $Y_i$ as constants throughout this proof.

Denote $\mu_1=\mathbb{E}[|Y_i-g(X_i,U_{b,i,k})|\Big|X_i,Y_i]$ and $\mu_2=-\mathbb{E}[|g(X_i,U_{b,i,k})-g(X_i,U_{b,i,k'})|\Big|X_i,Y_i]$ for any $k\neq k'$,
we have 
\[Z_i(g)=\frac{2}{m_1m_2}\sum_{b=1}^{m_1}\sum_{k=1}^{m_2}(|Y_i-g(X_i,U_{b,i,k})|-\mu_1)-\frac{1}{m_1m_2(m_2-1)}\sum_{b=1}^{m_1}\sum_{k\neq k'}(|g(X_i,U_{b,i,k})-g(X_i,U_{b,i,k'})|-\mu_2).\]

The first term in the above is easy as $|Y_i-g(X_i,U_{b,i,k})|$ as independent for $b\in[m_1],k\in[m_2]$, and $0\le |Y_i-g(X_i,U_{b,i,k})|\le 2$. By Hoeffding's inequality, we have 
\begin{align}
\label{equ:bound-Z-first-term}\mathbb{P}\left(\left|\frac{2}{m_1m_2}\sum_{b=1}^{m_1}\sum_{k=1}^{m_2}(|Y_i-g(X_i,U_{b,i,k})|-\mu_1)\right|\ge t\right)\le 2\exp(-\frac{m_1m_2t^2}{8}).\end{align}
For the second term, we use the standard splitting for proving U-statistic-type concentration inequality.

Let $l=\lfloor\frac{m_2}{2}\rfloor$, and set
\[W_i(u_1,u_2,\cdots,u_n)=\frac{1}{l}\sum_{j=1}^l (|g(X_i,u_{2j-1})-g(X_i,u_{2j})|-\mu_2),\] i.e., we break the samples into $k$ consecutive non-overlapping pairs. Let $\mathcal{S}^{m_2}$ as the permutation group in $[m_2]$, we can verify 
\begin{align*}
\frac{1}{m_1m_2(m_2-1)}\sum_{b=1}^{m_1}\sum_{k\neq k'}(|g(X_i,U_{b,i,k})-g(X_i,U_{b,i,k'})|-\mu_2)&=\frac{1}{m_2!}\sum_{\sigma\in\mathcal{S}^{m_2}}\frac{1}{m_1}\sum_{b=1}^{m_1}W_i(U_{b,i,\sigma(1)},\cdots,U_{b,i,\sigma(m_2)})\\&=\frac{1}{m_2!}\sum_{\sigma\in\mathcal{S}^{m_2}}T_\sigma
\end{align*}
where for $\sigma\in\mathcal{S}^{m_2}$ denote $T_{\sigma}=\frac{1}{m_1}\sum_{b=1}^{m_1}W_i(U_{b,i,\sigma(1)},\cdots,U_{b,i,\sigma(m_2)})$. We can see that there are $lm_1$ i.i.d. terms in the expression of each $T_{\sigma}$. By Hoeffding's Inequality, we get $T_\sigma$ is sub-Gaussian with variance proxy $\frac{1}{l^2m_1^2}$. 

We can bound the moment generating function of $\frac{1}{m_2!}\sum_{\sigma\in\mathcal{S}^{m_2}}T_\sigma$ as 
\begin{align*}
\mathbb{E}\exp\left(\frac{\lambda}{m_2!}\sum_{\sigma\in\mathcal{S}^{m_2}}T_\sigma\right)&\le \frac{1}{m_2!}\sum_{\sigma\in\mathcal{S}^{m_2}}\mathbb{E}\exp(\lambda T_\sigma)\\&\le \exp\left(\frac{\lambda^2}{2l^2m_1^2}\right)
\end{align*}
where the first inequality is by Jensen's Inequality. Therefore, $\frac{1}{m_2!}\sum_{\sigma\in\mathcal{S}^{m_2}}T_\sigma$ is also sub-Gaussian with variance proxy $\frac{1}{l^2m_1^2}\le \frac{9}{m_1^2m_2^2}$ as $m_2\ge 2$, i.e.,
\begin{align}
\label{equ:bound-Z-second-term}\mathbb{P}\left(\left|\frac{1}{m_1m_2(m_2-1)}\sum_{b=1}^{m_1}\sum_{k\neq k'}(|g(X_i,U_{b,i,k})-g(X_i,U_{b,i,k'})|-\mu_2)\right|\ge t\right)\le 2\exp(-\frac{m_1m_2t^2}{18}).\end{align}
By the union bound of \eqref{equ:bound-Z-first-term} and \eqref{equ:bound-Z-second-term}, we arrive at 
\[\mathbb{P}\left(|Z_i(g)|\ge  5t\right)\le 4\exp(-\frac{m_1m_2t^2}{2}).\]
Replacing $t$ with $5t$ finishes the first half of the proof. 

The bound on $\mathsf{Z_i}(g)$ follows by combining the following two equations,
\[\mathsf{Var}\left(\frac{2}{m_1m_2}\sum_{b=1}^{m_1}\sum_{k=1}^{m_2}(|Y_i-g(X_i,U_{b,i,k})|-\mu_1)\right)\le \frac{16}{m_1m_2}\]
and 
\begin{align*}
\mathsf{Var}\left(\frac{1}{m_1m_2(m_2-1)}\sum_{b=1}^{m_1}\sum_{k\neq k'}(|g(X_i,U_{b,i,k})-g(X_i,U_{b,i,k'})|-\mu_2)\right)&=\mathsf{Var}\left(\frac{1}{m_2!}\sum_{\sigma\in\mathcal{S}^{m_2}}T_\sigma\right)\\
&\le \max_{\sigma\in\mathcal{S}^{m_2}}\mathsf{Var}(T_\sigma)\le \frac{4}{lm_1}\le \frac{12}{m_1m_2}.
\end{align*}
\end{proof}

Recall that
$\chi(g)=\frac{1}{n}\sum_{i=1}^n Z_i(g)$, and for any $g\in\mathcal{G}$, $|Z_i(g)|\le 4$, $\mathsf{Var}(Z_i(g))\le \frac{56}{m_1m_2}$.

Applying Bernstein's Inequality, we have for any $t>0$,
\[\mathbb{P}\left(|\chi(g)|\ge t\right)\le \mathbb{P}\left(|\sum_{i=1}^n Z_i(g)|\ge tn\right)\le 2\exp\left(-\frac{n^2t^2}{2n\cdot \frac{32}{m_1m_2}+\frac{2}{3}\cdot 4nt}\right)\]
which implies that with probability at least $1-2\delta$,
\[|\chi(g)|\lesssim \sqrt{\frac{\log(1/\delta)}{nm_1m_2}}+\frac{\log(1/\delta)}{n}.\]

Taking union bound over $g\in\mathcal{S}_M$, we have with probability at least $1-2\delta$,
\begin{align}
|\sup_{g\in\mathcal{G}}\chi(\pi_M(g))|&\lesssim \sqrt{\frac{\log(N_M/\delta)}{nm_1m_2}}+\frac{\log(N_M/\delta)}{n}\nonumber\\
&\lesssim  \sqrt{\frac{\log(1/\delta)+N^2L^2\log n+MN^2L}{nm_1m_2}}+\frac{\log(1/\delta)+N^2L^2\log n+MN^2L}{n}\label{equ:chaining-part-1}
\end{align}
where in the second inequality we use $\log N_k\le C(N^2L^2\log n+kN^2L\log 2)$.

We proceed to bound $\sum_{k=M}^\infty \sup_{g\in \mathcal{G}}\left[\chi(\pi_{k+1}(g))-\chi(\pi_k(g))\right]$.
It holds by the triangle inequality that,
\begin{align*}
|\tilde Z_i(g_1)-\tilde Z_i(g_2)|&\le\frac{1}{m_1} \sum_{b=1}^{m_1} \Big(\frac{1}{m_2} \sum_{k=1}^{m_2} 2|g_1(X_i, U_{b,i,k})- g_2(X_i, U_{b,i,k})| \\&\qquad+ \frac{1}{m_2(m_2-1)}\sum_{k\neq k'} |g_1(X_i, U_{b, i, k}) - g_1(X_i, U_{b, i, k'})-g_2(X_i, U_{b, i, k}) + g_2(X_i, U_{b, i, k'})|\Big)\\
&\le 4\|g_1-g_2\|_\infty.
\end{align*}

By definition of $\pi_k(g)$ and $\pi_{k+1}(g)$, we have $\|\pi_k(g)-\pi_{k+1}(g)\|_\infty\le \|\pi_k(g)-g\|_\infty+\|g-\pi_{k+1}(g)\|_\infty\le 3\cdot 2^{-k-1}$.
The above display leads to $|\tilde Z_i(g_1)-\tilde Z_i(g_2)|\le 3\cdot 2^{1-k}$. Similarly, we have $|\mathbb{E}\tilde Z_i(g_1)-\mathbb{E}\tilde Z_i(g_2)|\le 3\cdot 2^{1-k}$, and hence $|Z_i(g_1)- Z_i(g_2)|\le 3\cdot 2^{2-k}$.

Applying Hoeffding's Inequality, we have that 
\begin{align*}
\mathbb{P}\left(|\chi(\pi_{k+1}(g))-\chi(\pi_k(g))|\ge t\right)&=\mathbb{P}\left(\Big|\sum_{i=1}^n[\tilde Z_i(\pi_{k+1}(g))-\tilde Z_i(\pi_k(g))]\Big|\ge nt\right)\\&
\le 2\exp\left(-\frac{2(nt)^2}{n(3\cdot 2^{2-k})^2}\right)
\end{align*}
which implies with probability at least $1-2\delta$, it holds that
\[|\chi(\pi_{k+1}(g))-\chi(\pi_k(g))|\le 3\cdot 2^{2-k}\sqrt{\frac{\log(1/\delta)}{2n}}.\]

Note that $(\pi_{k+1}(g),\pi_k(g))$ can only take at most $N_k\times N_{k+1}$ choices. A union bound then implies that 
with probability at least $1-2\delta$, it holds that 
\begin{align*}
\sup_{g\in\mathcal{G}}|\chi(\pi_{k+1}(g))-\chi(\pi_k(g))|&\le 3\cdot 2^{2-k}\sqrt{\frac{\log(N_kN_{k+1}/\delta)}{2n}}\\
&\lesssim 2^{-k}\sqrt{\frac{\log(1/\delta)+N^2L^2\log n+kN^2L}{n}}.
\end{align*}
For the sake of further addition, for $k\ge M$ replace $\delta$ by $\delta 2^{M-k}$, and we have with probability at least $1-\delta 2^{M-k+1}$,
\begin{align}
\sup_{g\in\mathcal{G}}|\chi(\pi_{k+1}(g))-\chi(\pi_k(g))|
&\lesssim 2^{-k}\sqrt{\frac{\log(1/\delta)+(k-M)\log 2+N^2L^2\log n+kN^2L}{n}}\nonumber\\&\lesssim 
2^{-k}\sqrt{\frac{\log(1/\delta)+N^2L^2\log n}{n}}+\frac{2^{-k}\sqrt{kN^2L}}{\sqrt{n}}. \label{equ:chaining-part-2}
\end{align}

Now we are ready to apply a further union bound over \eqref{equ:chaining-part-1} and \eqref{equ:chaining-part-2} for $k=M,M+1,\cdots$.
And in view of \eqref{equ:chaining-decomposition}, we have under the joint event that we denote as $\mathcal{A}$, 
\begin{align*}
|\sup_{g\in\mathcal{G}}\chi(g)|&\lesssim  \sqrt{\frac{\log(1/\delta)+N^2L^2\log n+MN^2L}{nm_1m_2}}+\frac{\log(1/\delta)+N^2L^2\log n+MN^2L}{n}\\&\quad\qquad+\sum_{k=M}^\infty \left(2^{-k}\sqrt{\frac{\log(1/\delta)+N^2L^2\log n}{n}}+\frac{2^{-k}\sqrt{kN^2L}}{\sqrt{n}}\right)\\&\lesssim  \sqrt{\frac{\log(1/\delta)+N^2L^2\log n+MN^2L}{nm_1m_2}}+\frac{\log(1/\delta)+N^2L^2\log n+MN^2L}{n}\\&\quad\qquad+2^{-M}\sqrt{\frac{\log(1/\delta)+N^2L^2\log n+MN^2L}{n}}.
\end{align*}

Taking $M=\lceil\frac{\log n}{2\log 2}\rceil$, we have $2^{-M}\le \frac{1}{\sqrt{n}}$
which leads to
\begin{align*}
2^{-M}\sqrt{\frac{\log(1/\delta)+N^2L^2\log n+MN^2L}{n}}&\le  \frac{\sqrt{\log(1/\delta)+N^2L^2\log n+MN^2L}}{n}\\
&\le \frac{\log(1/\delta)+N^2L^2\log n+MN^2L}{n}
\end{align*}
and hence
\begin{align*}
|\sup_{g\in\mathcal{G}}\chi(g)|&\lesssim \sqrt{\frac{\log(1/\delta)+N^2L^2\log n}{nm_1m_2}}+\frac{\log(1/\delta)+N^2L^2\log n}{n}
\end{align*}
where we used $MN^2L\lesssim N^2L^2\log n$.

Now we are left with calculating the probability of event $\mathcal{A}$. By \eqref{equ:chaining-part-1} and \eqref{equ:chaining-part-2}, we have
\begin{align*}
\mathbb{P}[\mathcal{A}^c]\le 2\delta+\sum_{k=M}^\infty \delta2^{M-k+1}=6\delta.
\end{align*}
Letting $\delta=\frac{e^{-t}}{3}$, we have with probability at least $1-2e^{-t}$,
\begin{align*}
|\sup_{g\in\mathcal{G}}\chi(g)| \lesssim \sqrt{\frac{t+N^2L^2\log n}{nm_1m_2}}+\frac{t+N^2L^2\log n}{n}\lesssim\frac{
\delta_{n,t}}{\sqrt{m_1m_2}}+\delta^2_{n,t},
\end{align*}
which concludes the proof.
\subsection{Proof of Lemma~\ref{lemma:log-cover-number-conversion}}
\label{sec:proof-log-cover-number-coversion}

We need the following technical lemma.

\begin{lemma}
\label{lemma:g-to-G}
    If $\mathbb{E}[|g(U)|] + \mathbb{E}[|h(U)|] < \infty$ and $\|g - h\|_\infty < \epsilon$, then
    \begin{align*}
        \int (F_g(t) - F_h(t))^2 dt \le 4\epsilon
    \end{align*} where $F_g(t) = \mathbb{E}_U[\indicator{g(U)\le t}]$.
\end{lemma}
\begin{proof}[Proof of Lemma~\ref{lemma:g-to-G}]
Observe that
\begin{align*}
    F_g(t) = \mathbb{P}(g(U) \le t) = \mathbb{P}\left(h(U) + g(U) - h(U) \le t\right) \in [\mathbb{P}(h(U) \le t-\epsilon), \mathbb{P}(h(U) \le t+\epsilon)],
\end{align*} where the last inequality follows from $\|g-h\|_\infty \le \epsilon$ and the monotonicity of CDF function. Then
\begin{align*}
    F_h(t-\epsilon) \le F_g(t) \le F_h(t+\epsilon), 
\end{align*} which further implies
\begin{align*}
    \int |F_g(t) - F_h(t)|^2 dt &\le \int \max\{|F_h(t-\epsilon) - F_h(t)|^2, |F_h(t+\epsilon) - F_h(t)|^2\} dt \\
    &\le \int |F_h(t-\epsilon) - F_h(t)|^2 + |F_h(t+\epsilon) - F_h(t)|^2 dt.
\end{align*} We claim that 
\begin{align}
\label{eq:ident1}
\int (F(t) - F(t - \epsilon))^2 dt \le 2\epsilon\qquad \text{  for any } \epsilon>0 \text{ and CDF }F \text{ satisfying }\mathbb{E}_{X\sim F}[|X|] < \infty.
\end{align} This completes the proof. 
\noindent \emph{Proof of \eqref{eq:ident1}.} Let $X, X'\sim F$ be independent and $Y=X'-\epsilon$, $Y'=X-\epsilon$, it follows from the identity of energy distance that
\begin{align*}
    \int (F(t) - F(t - \epsilon))^2 dt &= 2\mathbb{E}[|X-Y|] - \mathbb{E}[|X-X'|] - \mathbb{E}[|Y-Y'|] \\
    &= 2\mathbb{E}[|X-X'+\epsilon|] - \mathbb{E}[|X-X'|] - \mathbb{E}[|X-X'|] \le 2\epsilon.
\end{align*}\end{proof}

Now we are ready to prove Lemma \ref{lemma:log-cover-number-conversion}.

\begin{proof}[Proof of Lemma~\ref{lemma:log-cover-number-conversion}]
    By the definition of covering number, for any $\epsilon^\star>0$, there exists $N^\star \le \mathcal{N}(\epsilon^\star, \mathcal{G}, \|\cdot\|_\infty)$ and $g_1,\ldots, g_{N^\star}$ such that for any $g \in \mathcal{G}$, there exists some $\pi(g) \in \{g_1,\ldots, g_{N^\star}\}$ such that
    \begin{align}
    \label{eq:cover-prop}
        \sup_{(x,u)\in [0,1]^{d+1}} |g(x, u) - \pi(g)(x, u)| \le \epsilon^\star.
    \end{align} We construct a covering of $\mathcal{V}(\mathcal{G})$ as 
    \begin{align}
    \label{eq:cover1}
        \{ v_{{g_k}, {g_j}}(x) \}_{k, j \in [N^\star]}
    \end{align} Then for any $v_{g,\tilde{g}}(x) \in \mathcal{V}(\mathcal{G})$ with $g, \tilde{g} \in \mathcal{G}$, and any $x$, we will show that
    \begin{align}
    \label{eq:cover1-error}
        |v_{g,\tilde{g}}(x) - v_{\pi(g), \pi(\tilde{g})}(x)| \le 17\sqrt{c_0}(\epsilon^\star)^{1/2}.
    \end{align} Without loss of generality, we assume that $(\epsilon^\star)^{1/2} \le \sqrt{c_0}/8$, otherwise the inequality is trivial because $v(x) \in [0, 2c_0]$. 

    To be specific, for any $g, h$, we denote $\langle g, h\rangle(x) = \int g(x, t) h(x, t) dt$. We will use the decomposition that 
    \begin{align*}
         v_{g, \tilde{g}}(x) - v_{\pi(g), \pi(\tilde{g})}(x) &= \int (F_g - F_{\pi(g)} + F_{\pi(g)} - F_{\pi(\tilde{g})} + F_{\pi(\tilde{g})} - F_{\tilde{g}})^2(x, t) dt - \int (F_{\pi(g)} - F_{\pi(\tilde{g})})^2(x, t) dt \\
         &= \underbrace{\langle F_g - F_{\pi(g)}, F_g - F_{\pi(g)}\rangle(x) + \langle F_{\tilde{g}} - F_{\pi(\tilde{g})}, F_{\tilde{g}} - F_{\pi(\tilde{g})}\rangle(x)}_{\mathsf{T}_1(x)} \\
         &~~~~~~~~ \underbrace{- 2 \langle F_g - F_{\pi(g)}, F_{\tilde{g}} - F_{\pi(\tilde{g})}\rangle(x)}_{\mathsf{T}_2(x)} \\
         &~~~~~~~~ + \underbrace{2 \langle F_g - F_{\pi(g)}, F_{\pi(g)} - F_{\pi(\tilde{g})}\rangle(x) - 2\langle F_{\tilde g} - F_{\pi(\tilde{g})}, F_{\pi(g)} - F_{\pi(\tilde{g})}\rangle(x)}_{\mathsf{T}_3(x)}
    \end{align*} with
    \begin{align*}
        |\mathsf{T}_1(x)| \lor |\mathsf{T}_2(x)| \le 8\epsilon^\star \qquad \text{and} \qquad |\mathsf{T}_3(x)| \le 8\sqrt{\epsilon^\star} v_{\pi(g), \pi(\tilde{g})}(x),
    \end{align*} where the bounds on $\mathsf{T}_1$ -- $\mathsf{T}_3$ follows from applying covering property \eqref{eq:cover-prop} in Lemma \ref{lemma:g-to-G} and Cauchy-Schwarz inequality. 

For example, we have 
\begin{align*}
|T_3(x)|&\le 
2\|F_g-F_{\pi(g)}\|(x)\cdot \|F_{\pi(g)}-F_{\pi(\tilde g)}\|(x)+2\|F_{\tilde g}-F_{\pi(\tilde g)}\|(x)\cdot \|F_{\pi(g)}-F_{\pi(\tilde g)}\|(x)\\&\le 
2\sqrt{4\epsilon^{\star}v_{\pi(g), \pi(\tilde{g})}(x)}+2\sqrt{4\epsilon^{\star}v_{\pi(g), \pi(\tilde{g})}(x)}=8\sqrt{2c_0\epsilon^{\star}}.
\end{align*}
    
This completes the proof of the claim \eqref{eq:cover1-error} by noting that
    \begin{align*}
        \left|v_{g, \tilde{g}}(x) - v_{\pi(g), \pi(\tilde{g})}(x) \right|\le 8(\epsilon^{\star})^{1/2} \times \frac{\sqrt{c_0}}{8} +8 \sqrt{2c_0\epsilon^\star} \le 17 \sqrt{c_0\epsilon^\star}.
    \end{align*}
    
    Now we have already shown that the \eqref{eq:cover1} is a $17(c_0\epsilon^\star)^{1/2}$-cover of $\mathcal{V}(\mathcal{G})$ by \eqref{eq:cover1-error}. Letting $\epsilon = 17(c_0\epsilon^\star)^{1/2}$, we obtain
    \begin{align*}
        \log \mathcal{N}(\epsilon, \mathcal{V}(\mathcal{G}), \|\cdot \|_{\infty}) \le 2 \log \mathcal{N}(\epsilon^\star, \mathcal{G}, \|\cdot \|_{\infty}) = 2 \log \mathcal{N}\left((\epsilon/17)^2/c_0, \mathcal{G}, \|\cdot \|_{\infty}\right).
    \end{align*}
\end{proof}

\subsection{Proof of Lemma~\ref{lemma:log-cover-number-conversion-2}}
\label{sec:proof-log-cover-number-coversion-2}

The proof of Lemma \ref{lemma:log-cover-number-conversion-2} follows the same vein.

\begin{proof}[Proof of Lemma~\ref{lemma:log-cover-number-conversion-2}]
Same as Lemma \ref{lemma:log-cover-number-conversion}, for any $\epsilon^{\star}>0$, there exists $N^{\star}\le \mathcal{N}(\epsilon^{\star},\mathcal{G},\|\cdot\|_{\infty})$ and $g_1,\cdots,g_{N^{\star}}$ such that for any $g\in \mathcal{G}$, there exists some $\pi(g)\in \{g_1,\cdots,g_{N^{\star}}\}$ such that 
\[\sup_{(x,u)\in [0,1]^{d+1}}|g(x,u)-\pi(g)(x,u)|\le \epsilon^{\star}.\]
Construct a covering of $\mathcal{U}(\mathcal{G})$ as 
\[\{u_{g_k,g_j}(x,y)\}_{k,j\in [N^{\star}]}\]
For any $u_{g,\tilde g}(x,y)\in \mathcal{U}(\mathcal{G})$ with $g,\tilde g\in \mathcal{G}$ and any $(x,y)$, we can show that 
\begin{align}
\label{eq:cover2-error}
|u_{g,\tilde g}(x,y)-u_{\pi(g),\pi(\tilde g)}(x,y)|\le 20 (c_0\epsilon^{\star})^{1/2}.
\end{align}
To see this, we have the following decomposition that 
\begin{align*}
    &u_{g,\tilde g}(x,y)-u_{\pi(g),\pi(\tilde g)}(x,y)\\=&\int (F_g(x, t) - F_{\tilde{g}}(x, t)) (F_{\tilde{g}}(x, t) - \indicator{y \le t}) dt- \int (F_{\pi(g)}(x, t) - F_{\pi(\tilde{g})}(x, t)) (F_{\pi(\tilde{g})}(x, t) - \indicator{y \le t}) dt\\
    =&\underbrace{\langle F_{\pi(g)} - F_{\pi(\tilde{g})}, F_{\tilde{g}}-F_{\pi(\tilde{g})}\rangle (x)}_{\mathsf{T}_1(x)}
    +\underbrace{\langle F_g-F_{\pi(g)},F_{\tilde g}\rangle(x)+\langle F_{\pi(\tilde g)}-F_{\tilde g},F_{\tilde g}\rangle(x)}_{\mathsf{T}_2(x)}\\
    \quad &+\underbrace{\langle F_{\pi(g)}-F_g,\indicator{\cdot\ge y} \rangle(x)-\langle F_{\pi(\tilde g)}-F_{\tilde g},\indicator{\cdot\ge y}\rangle(x)}_{\mathsf{T}_3(x)}.
\end{align*}

By Cauchy-Schwartz Inequality and the covering definition, we can bound each term as 
\[\mathsf{T}_1(x)\le \sqrt{2c_0\times 4\epsilon},\qquad \mathsf{T}_2(x)\le  2\sqrt{2c_0\times 4\epsilon}\qquad \mathsf{T}_3(x)\le  2\sqrt{2c_0\times 4\epsilon}\]

This completes the proof of \eqref{eq:cover2-error}.

Similar to the proof of Lemma \ref{lemma:log-cover-number-conversion}, let $\epsilon=20(c_0\epsilon^{\star})^{1/2}$, we obtain that
\begin{align*}
        \log \mathcal{N}(\epsilon, \mathcal{U}(\mathcal{G}), \|\cdot \|_{\infty}) \le 2 \log \mathcal{N}(\epsilon^\star, \mathcal{G}, \|\cdot \|_{\infty}) = 2 \log \mathcal{N}\left((\epsilon/20)^2/c_0, \mathcal{G}, \|\cdot \|_{\infty}\right).
    \end{align*}
\end{proof}

\section{Proof of Energy Identity}
\label{proof:identity}
\begin{proof}
Let $P,Q$ be two general distributions.
For random variables $X\sim P,\ Y\sim Q$, we rewrite the energy statistics as $E(X,Y)=d_E(P, Q)$.
We first claim that
\begin{align}
    d_E(P, Q) = 2\int \{P(t) - Q(t)\}^2 dx
\label{eq:e-stats-1}
\end{align} where for simplicity of notation, we also use $P$ and $Q$ to denote the CDFs of $P$ and $Q$.

To prove that,
we have from a simple calculation that
\begin{align*}
    \mathbb{E}_{X\sim P, Y\sim Q}[|X-Y|] &= \int_{-\infty}^\infty \left(\int_{-\infty}^x (x-y) dQ(y) + \int_{x}^\infty (y-x) dQ(y)\right) dP(x) \\
    &= \int \left(2xQ(x) - 2\int_{-\infty}^x ydQ(y) + \mathbb{E}[Y] \right) dP(x).
\end{align*}
Then, we have
\begin{align*}
&\mathbb{E}_{X\sim P, Y\sim Q}[|X-Y|] - \mathbb{E}_{X, X'\sim P}[|X-X'|] \\
&= \mathbb{E}[Y] - \mathbb{E}[X] + \int_{-\infty}^\infty \left(2x(Q(x) - P(x)) - 2\int_{-\infty}^x y(dQ(y) - dP(y))\right) dP(x) \\
&= \mathbb{E}[Y] - \mathbb{E}[X] + 2\left(xQ(x) - xP(x) - \int_{-\infty}^x y(dQ(y)-dP(y)) \right) \Big|_{x=-\infty}^{x=\infty} \\ &~~~~~~~~~~~~ - \int_{-\infty}^\infty 2(Q(x) - P(x))P(x) dx \\
&= \int_{-\infty}^\infty 2P(x)(P(x) - Q(x)) dx.
\end{align*}
Similarly, we also have
\begin{align*}
\mathbb{E}_{X\sim P, Y\sim Q}[|X-Y|] - \mathbb{E}_{Y, Y'\sim Q}[|Y-Y'|] = \int_{-\infty}^\infty 2Q(x)(Q(x) - P(x)) dx.
\end{align*}
Hence, adding up, we complete the proof of \eqref{eq:e-stats-1}.

For any fixed $x$, take $P$ and $Q$ to be the conditional distribution of $Y$ and $g(X,U)$ given $X=x$ respectively, we have that 
\begin{align*}
2\EE_{Y|x}|Y-g(x,U)|-\EE_{Y|x}|Y-Y'|-\EE|g(x,U)-g(x,U')|=d_E^2(P,Q)=2\int \{F_g(x,y)-F^{\star}(x,y)\}^2dy.
\end{align*}

Now Taking expectation with respect to $X=x$, we have that 
\begin{align*}
2\EE|Y-g(x,U)|-\EE|Y-Y'|-\EE|g(x,U)-g(x,U')|=2\int \{F_g(x,y)-F^{\star}(x,y)\}^2dyd\mu_0(x)=2\|F_g-F^{\star}\|_2^2.
\end{align*}
Note that in the last display, $-\EE|Y-Y'|$ is independent of $g$ and only depends on $\mu_0$, hence letting $C_{\mu_0}=\EE|Y-Y'|$
completes the proof.

\end{proof}

\section{Proofs of Rates and Downstream Tasks}

\subsection{Proof of Corollary \ref{cor:final-rate}}

Note that the distribution of $F_0(U)$ is $\mathcal{U}[0,1]$. Let $g^{\star}(x,u)=Q^{\star}(x,F_0(u))$, we have the $\alpha$-conditional quantile of $g^{\star}(x,U)$ is $Q^{\star}(x,\alpha)$. By the definition of HCM and smoothness of $F_0$, we have that $g^{\star}\in \calH(p+1,l+1,\tilde P)$ with $\gamma^{\star}(\calH(p+1,l+1,\tilde P))=\gamma^{\star}.$

By the neural network approximation results, e.g., Theorem 3.4 in \cite{fan2022noise}, there exists some $g\in\mathcal{H}_{\mathtt{nn}}(d, L, N, \tilde d, B)$ such that $\|g-g^{\star}\|_\infty\le (NL/\log(NL))^{-2\gamma^{\star}}$. Now denote the $\alpha$-conditional quantile of $g(x,U)$ as $Q(x,\alpha)$. We have $\PP(g(x,U)\le Q^{\star}(x,\alpha)+\|g-g^{\star}\|_\infty)\ge \PP(g^{\star}(x,U)\le Q^{\star}(x,\alpha))\ge \alpha$, hence $Q(x,\alpha)\le Q^{\star}(x,\alpha)+\|g-g^{\star}\|_\infty$. Similarly, we have $Q(x,\alpha)\ge Q^{\star}(x,\alpha)-\|g-g^{\star}\|_\infty$. It follows that $\|Q-Q^{\star}\|_\infty\le \|g-g^{\star}\|_\infty\le (NL/\log(NL))^{-2\gamma^{\star}}$. 

Now we are ready to bound $\|F_g-F^{\star}\|_\infty$. To be concrete, for any $x,y$, we have $|F_g(x,y)-F^{\star}(x,y)|=|F^{\star}(x,Q^{\star}(x,F_g(x,u)))-F^{\star}(x,u)|\le c_1|Q^{\star}(x,F_g(x,u))-u|=c_1|Q^{\star}(x,F_g(x,u))-Q(x,F_g(x,u))|\le c_1(NL/\log(NL))^{-2\gamma^{\star}}$. 

Hence $\delta_a\le \|F_g-G^{\star}\|_\infty^2\le c_1(NL/\log(NL))^{-4\gamma^{\star}}$. Pick $NL\asymp n^{\frac{1}{4\gamma^{\star}+2}}(\log n)^{\frac{4\gamma^{\star}-1}{2\gamma^{\star}+1}}$, and further note $m_1 m_2\gtrsim n^{\frac{2\gamma^{\star}}{2\gamma^{\star}+1}}$, we arrive from Theorem \ref{thm:oracle-n-infty}, 
\begin{align*}
\frac{\|F_{\hat{g}} - F^\star\|^2_2}{C}\le n^{-\frac{2\gamma^{\star}}{2\gamma^{\star}+1}}(\log n)^{\frac{12\gamma^{\star}}{2\gamma^{\star}+1}}+\delta_{\mathtt{opt}}+\frac{t}{n}.
\end{align*}
\qed

\subsection{Proof of Proposition \ref{prop:conditional-moments}}
\label{sec:proof-cond-moments}
To facilitate the proof, define $\breve m_{w}(x)=\EE\left[w(\hat g(x,U))\right]$. By $|w|\le c_1$ in Condition \ref{cond:conditional-moments}, we have $\mathsf{Var}[w(\hat g(x,U))]\le c_1^2$ for every $x$.
Hence 
\begin{align*}
\EE\Big[\|\hat m_{w}-\breve m_{w}\|^2\Big| \hat g\Big]=\EE\left[\frac{\mathsf{Var}[w(\hat g(X,U))|X,\hat g]}{k}\Big|\hat g\right]\le \frac{c_1^2}{k}.
\end{align*}

On the other hand, it holds that
\begin{align}
\label{equ:distribution-to-conditional-moments}
\EE\Big[\|\breve m_{w}-m^{\star}_{w}\|^2\Big|\hat g\Big]&=\EE\left(\int w(y)dF^{\star}(X,y)-\int w(y)dF_{\hat g}(X,y)\right)^2\nonumber\\&=
\EE\left(\int (F^{\star}(X,y)-F_{\hat g}(X,y))w'(y)dy\right)^2\nonumber\\&\le c_1\EE\left(\int |F^{\star}(X,y)-F_{\hat g}(X,y)|^2dy\right)\le c_1C \delta_{\mathtt{NDRE}}
\end{align}
where the first inequality follows from the Cauchy-Schwarz Inequality. The proof of the first part follows from $\EE\Big[\|\hat m_{w}-m^{\star}_{w}\|^2\Big| \hat g\Big]\le 2\EE\Big[\|\hat m_{w}-\breve m_{w}\|^2\Big| \hat g\Big]+2\EE\Big[\| m^{\star}_{w}-\breve m_{w}\|^2\Big| \hat g\Big]$.

For the second part,
as illustrated in Section \ref{sec:downstream-estimation}, assign different noises to different test samples. Hence for test sample $X_i'$, denote the corresponding noises as $\{U_{i,l}\}_{l=1}^k$, we have $\hat m_w(X'_i)=\frac{1}{k}\sum_{l=1}^k w(\hat g(X'_i,U_{i,l}))$.
The first part says
\begin{align*}
\EE\Big[(\hat m_w(X'_i)-m^{\star}_w(X'_i))^2\Big| \hat g\Big]\le C(\frac{1}{k}+\delta_{\mathtt{NDRE}}).
\end{align*}
We can calculate the variance as 
\begin{align*}
\Var\Big[(\hat m_w(X'_i)-m^{\star}_w(X'_i))^2\Big| \hat g\Big]\lesssim  \EE\Big[(\hat m_w(X'_i)-m^{\star}_w(X'_i))^2\Big| \hat g\Big]\le C(\frac{1}{k}+\delta_{\mathtt{NDRE}}).
\end{align*}
Therefore, applying Bernstein's Inequality, it follows that with probability at least $1-e^{-t'}$,
\begin{align*}
\Big|\frac{1}{n'} \sum_{i=1}^{n'} \left(\hat{m}_w(X'_i) - m^\star_w(X'_i)\right)^2 -\EE\Big[(\hat m_w(X'_i)-m^{\star}_w(X'_i))^2\Big| \hat g\Big]\Big|
&\lesssim \sqrt{\frac{t'}{n'}}\sqrt{C(\frac{1}{k}+\delta_{\mathtt{NDRE}})}+\frac{t'}{n'}\\
&\lesssim \frac{1}{k}+\delta_{\mathtt{NDRE}}+\frac{t'}{n'}. 
\end{align*}
Hence, the proof is completed.

\qed


\subsection{Proof of Proposition \ref{prop:condition-quantile}}
\label{sec:proof-conditional-quantile}
\begin{proof}
We first state a result bounding the empirical process, see Corollary 4.15 in \cite{wainwright2019high}.

\begin{lemma}
\label{lemma:glivenko-cantelli}
Let $F(t)$ be the CDF of a random variable $X$ and let $\hat F_k$ be the empirical CDF based on $k$ i.i.d. samples $X_i\sim \PP$. Then 
\begin{align*}
\PP\left[\|\hat F_k-F\|_\infty\ge 8\sqrt{\frac{\log(k+1)}{k}}+\delta\right]\le e^{-\frac{k\delta^2}{2}}.  
\end{align*}
\end{lemma}

Fix $x$, recall that the conditional CDF of $Y|X=x$ is $F^{\star}(x,y)$. Denote the empirical conditional CDF of $\{\hat g(x,U_l)\}_{l=1}^k$ as $\hat F_k(x,y)$, we have by Lemma \ref{lemma:glivenko-cantelli}, with probability at least $1-e^{-\frac{k\delta^2}{2}}$, 
\begin{align*}
\sup_y|\hat F_k(x,y)-F_{\hat g}(x,y)|
\le 8\sqrt{\frac{\log(k+1)}{k}}+\delta. 
\end{align*}
Denote $\delta_x=\int_{-1}^1 (F^{\star}(x,y)-F_{\hat g}(x,y))^2dy$, 
it follows from Cauchy-Schwarz Inequality that for any $0<a<b<1$,
\begin{align}
\label{equ:x-pointwise-CDF-discrepancy-bound}
\int_{a}^b|\hat F_k(x,y)-F^{\star}(x,y)|^2dy
\le 2\delta_x+2(b-a)\left(8\sqrt{\frac{\log(k+1)}{k}}+\delta\right)^2.
\end{align}

Now note that the proposed estimator can be rewritten as $\hat Q(x,\alpha)=\inf\limits_t\{t:\hat F_k(x,t)\ge \alpha\}$, we then consider the following two cases.

\begin{itemize}
\item If $\hat Q(x,\alpha)< Q^{\star}(x,\alpha)$, then 
\begin{align*}
2\delta_x+2\Big|Q^{\star}(x,\alpha)-\hat Q(x,\alpha)\Big|\left(8\sqrt{\frac{\log(k+1)}{k}}+\delta\right)^2&\ge \int_{\hat Q(x,\alpha)}^{Q^{\star}(x,\alpha)}|\hat F_k(x,y)-F^{\star}(x,y)|^2dy\\&\ge \int_{\hat Q(x,\alpha)}^{Q^{\star}(x,\alpha)}|\alpha-F^{\star}(x,y)|^2dy\\&\ge \int_{\hat Q(x,\alpha)}^{Q^{\star}(x,\alpha)}\Big|\alpha-\Big(\alpha-c_d(Q^{\star}(x,\alpha)-y))\Big)\Big|^2dy\\&=\frac{c_d^2}{3}\Big|\hat Q(x,\alpha)-Q^{\star}(x,\alpha)\Big|^3
\end{align*}
where the first equality follows from \eqref{equ:x-pointwise-CDF-discrepancy-bound}, the second inequality uses the definition of $\hat Q(x,\alpha)$ and that $\hat F_k(x,y)$ is  nondecreasing in $y$. The third inequality follows from the Condition \ref{cond:conditional-quantile} which leads to $F^{\star}(x,y)\le \alpha-c_d(Q^{\star}(x,\alpha)-y)$ for $y\le Q^{\star}(x,\alpha)$.

\item If $\hat Q(x,\alpha)> Q^{\star}(x,\alpha)$, then 
\begin{align*}
2\delta_x+2\Big|Q^{\star}(x,\alpha)-\hat Q(x,\alpha)\Big|\left(8\sqrt{\frac{\log(k+1)}{k}}+\delta\right)^2&\ge \int_{ Q^{\star}(x,\alpha)}^{\hat Q(x,\alpha)}|\hat F_k(x,y)-F^{\star}(x,y)|^2dy\\&\ge \int_{ Q^{\star}(x,\alpha)}^{\hat Q(x,\alpha)}|F^{\star}(x,y)-\alpha|^2dy\\&\ge \int_{ Q^{\star}(x,\alpha)}^{\hat Q(x,\alpha)}\Big|\Big(\alpha+c_d(y-Q^{\star}(x,\alpha)))\Big)-\alpha\Big|^2dy\\&=\frac{c_d^2}{3}\Big|\hat Q(x,\alpha)-Q^{\star}(x,\alpha)\Big|^3
\end{align*}
where similarly, the first equality follows from \eqref{equ:x-pointwise-CDF-discrepancy-bound}, the second inequality uses the definition of $\hat Q(x,\alpha)$ and that $\hat F_k(x,y)$ is  nondecreasing in $y$. The third inequality follows from the Condition \ref{cond:conditional-quantile} which leads to $F^{\star}(x,y)\ge \alpha+c_d(y-Q^{\star}(x,\alpha))$ for $y\ge Q^{\star}(x,\alpha)$.
\end{itemize}
Hence, in both cases, we have that with probability at least $1-e^{-\frac{k\delta^2}{2}}$, for every $0<\alpha<1$,
\begin{align*}
2\delta_x+2\Big|Q^{\star}(x,\alpha)-\hat Q(x,\alpha)\Big|\left(8\sqrt{\frac{\log(k+1)}{k}}+\delta\right)^2\ge\frac{c_d^2}{3}\Big|\hat Q(x,\alpha)-Q^{\star}(x,\alpha)\Big|^3.
\end{align*}
Take $\delta=2\sqrt{\frac{\log k}{k}}$, we have with probability at least $1-\frac{1}{k^2}$, the following event 
\begin{align}
\label{eq:event-A-cond-quantile}
\mathcal{A}(x):=\left\{\sup_{0<\alpha<1}\Big|\hat Q(x,\alpha)-Q^{\star}(x,\alpha)\Big|^2\le \max\big\{(\frac{12\delta_x}{c_d^2})^{2/3},\frac{1200\log (k+1)}{kc_d^2}\big\}.\right\}\end{align} holds.
It follows that 
\begin{align*}
\EE\Big[\sup_{0<\alpha<1}|\hat Q(x,\alpha)-Q^{\star}(x,\alpha)|^2\Big|\hat g\Big]&=\EE\Big[\mathds{1}_{\mathcal{A}(x)}\sup_{0<\alpha<1}|\hat Q(x,\alpha)-Q^{\star}(x,\alpha)|^2\Big|\hat g\Big]\\
&~~~~~~~~~~~~+\EE\Big[\mathds{1}_{\mathcal{A}(x)^c}\sup_{0<\alpha<1}|\hat Q(x,\alpha)-Q^{\star}(x,\alpha)|^2\Big|\hat g\Big]\\&\le (\frac{12\delta_x}{c_d^2})^{2/3}+\frac{1200\log (k+1)}{kc_d^2}+\frac{4c_0^2}{k^2}.
\end{align*}
Taking expectation with respect to the randomness of $X$ yields that
\begin{align*}
\EE\Big[\sup_{0<\alpha<1}|\hat Q(X,\alpha)-Q^{\star}(X,\alpha)|^2\Big|\hat g\Big]&\le (\frac{12\EE\delta_X}{c_d^2})^{2/3}+\frac{1200\log (k+1)}{k}+\frac{4c_0^2}{k^2}\\&\lesssim (\delta_{\mathtt{NDRE}})^{2/3}+\frac{\log k}{k} 
\end{align*}
where we use Jensen's Inequality and the definition of $\delta_{\mathtt{NDRE}}$. This completes the first part. The proof of the second part is similar to Proposition \ref{prop:conditional-moments}.

\end{proof}

\subsection{Proof of Proposition \ref{prop:prediction-band}}
\label{sec:proof-prediction-band}
\begin{proof}
We start with the proof of \eqref{eq:prop-prediction-band-1}.
\noindent\textbf{Proof of \eqref{eq:prop-prediction-band-1}}
Fix $x$ first.
Define $\hat F_k(x,y),\delta_x$ same as the proof of Proposition \ref{prop:condition-quantile} and recall \eqref{equ:x-pointwise-CDF-discrepancy-bound}. 
We similarly have two cases.
\begin{itemize}
    \item If $\hat F_k(x,\tilde y)<F^{\star}(x,\tilde y)$, then 
\begin{align*}
2\delta_x+2\Big|\frac{F^{\star}(x,\tilde y)-\hat F_k(x,\tilde y)}{C_d}\Big|\left(8\sqrt{\frac{\log(k+1)}{k}}+\delta\right)^2&\ge \int_{\tilde y-\frac{F^{\star}(x,\tilde y)-\hat F_k(x,\tilde y)}{C_d}}^{\tilde y}|\hat F_k(x,y)-F^{\star}(x,y)|^2dy\\&\ge \int_{\tilde y-\frac{F^{\star}(x,\tilde y)-\hat F_k(x,\tilde y)}{C_d}}^{\tilde y}|C_d(y-\tilde y)|^2dy\\&=\frac{1}{3C_d}|F^{\star}(x,\tilde y)-\hat F_k(x,\tilde y)|^3
\end{align*}
where the first inequality is by \eqref{equ:x-pointwise-CDF-discrepancy-bound}, the second inequality follows from Condition \ref{cond:prediction-band}.
\item If $\hat F_k(x,\tilde y)>F^{\star}(x,\tilde y)$, then similarly,
\begin{align*}
2\delta_x+2\Big|\frac{F^{\star}(x,\tilde y)-\hat F_k(x,\tilde y)}{C_d}\Big|\left(8\sqrt{\frac{\log(k+1)}{k}}+\delta\right)^2&\ge \int_{\tilde y}^{\tilde y+\frac{\hat F_k(x,\tilde y)-F^{\star}(x,\tilde y)}{C_d}}|\hat F_k(x,y)-F^{\star}(x,y)|^2dy\\&\ge \int_{\tilde y}^{\tilde y+\frac{\hat F_k(x,\tilde y)-F^{\star}(x,\tilde y)}{C_d}}|C_d(y-\tilde y)|^2dy\\&=\frac{1}{3C_d}|F^{\star}(x,\tilde y)-\hat F_k(x,\tilde y)|^3
\end{align*}
\end{itemize}
Therefore, in either case we have with probability at least $1-e^{-\frac{k\delta^2}{2}},$
\begin{align*}
2\delta_x+2\Big|\frac{F^{\star}(x,\tilde y)-\hat F_k(x,\tilde y)}{C_d}\Big|\left(8\sqrt{\frac{\log(k+1)}{k}}+\delta\right)^2\ge \frac{1}{3C_d}|F^{\star}(x,\tilde y)-\hat F_k(x,\tilde y)|^3,
\end{align*}
and hence for any $\tilde y$,
\begin{align*}
|F^{\star}(x,\tilde y)-\hat F_k(x,\tilde y)|^2\lesssim (\delta_x)^{2/3}+\left(8\sqrt{\frac{\log(k+1)}{k}}+\delta\right)^2.
\end{align*}

Taking $\tilde y$ equal to $Y_{\hat l(x)}(x)$ and $Y_{\hat l(x)+\lceil\alpha k\rceil}(x)$, where recall the definition of $\hat l(x)$ in \eqref{eq:pb-left-def} we have
\begin{align*}
&\quad(\delta_x)^{2/3}+\left(8\sqrt{\frac{\log(k+1)}{k}}+\delta\right)^2\\&\gtrsim |\hat F_k(x,Y_{\hat l(x)}(x))-F^{\star}(x,Y_{\hat l(x)}(x))|^2\\&~~~~~~~~~~~~+|\hat F_k(x,Y_{\hat l(x)+\lceil\alpha k\rceil}(x))-F^{\star}(x,Y_{\hat l(x)+\lceil\alpha k\rceil}(x))|^2\\&\ge \frac{1}{4}\Big|\hat F_k(x,Y_{\hat l(x)}(x))-F^{\star}(x,Y_{\hat l(x)}(x))-\hat F_k(x,Y_{\hat l(x)+\lceil\alpha k\rceil}(x))+F^{\star}(x,Y_{\hat l(x)+\lceil\alpha k\rceil}(x))\Big|^2
\end{align*}

Note that by Condition \ref{cond:prediction-band}, the probability of two equal $Y_{l,x}$ is zero. Without loss of generality, assume $Y_{1,x}<\cdots<Y_{k,x}$, we have that $\hat F_k(x,Y_{\hat l(x)}(x))=\frac{\hat l(x)}{k}$ and $\hat F_k(x,Y_{\hat l(x)+\lceil\alpha k\rceil}(x))=\frac{\hat l(x)+\lceil\alpha k\rceil}{k}$, and hence
\begin{align}
\label{equ:prediction-band-intermediate}
\Big|-F^{\star}(x,Y_{\hat l(x)}(x))+F^{\star}(x,Y_{\hat l(x)+\lceil \alpha k\rceil}(x))-\frac{\lceil \alpha k\rceil}{k}\Big|^2\lesssim (\delta_x)^{2/3}+\left(8\sqrt{\frac{\log(k+1)}{k}}+\delta\right)^2
\end{align}

Recall the definition in \eqref{eq:pb-est} and \eqref{equ:prediction-band-intermediate}. Take $\delta=2\sqrt{\frac{\log k}{k}}$, there exists some constant $C$, such that with probability at least $1-\frac{1}{k^2}$, we have
\begin{align}
\label{equ:prediction-band-intermediate-2}
\Big|\PP(Y\in \hat{\mathrm{PI}}_{\alpha}(x)|x,\hat g)-\alpha\Big|^2&=\Big|-F^{\star}(x,Y_{\hat l(x)}(x))+F^{\star}(x,Y_{\hat l(x)+\lceil \alpha k\rceil}(x))-\alpha\Big|^2\nonumber\\&\le C(\delta_x)^{2/3}+C\left(8\sqrt{\frac{\log(k+1)}{k}}+\delta\right)^2.
\end{align}

Denote \begin{align*}\mathcal{B}(x)=\Big\{\Big|\PP(Y\in \hat{\mathrm{PI}}_{\alpha}(x)|x,\hat g)-\alpha\Big|^2\le C(\delta_x)^{2/3}+C\left(8\sqrt{\frac{\log(k+1)}{k}}+2\sqrt{\frac{\log k}{k}}\right)^2\Big\},\end{align*}
we have $\PP(\mathcal{B}(x))\ge 1-\frac{1}{k^2}$. Taking expectation on two sides of \eqref{equ:prediction-band-intermediate-2} yields 
\begin{align*}
\EE\Big[\big|\PP(Y\in \hat{\mathrm{PI}}_{\alpha}(x)|x,\hat g)-\alpha\big|^2\Big|\hat g\Big]&\le \EE\Big[|\PP(Y\in \hat{\mathrm{PI}}_{\alpha}(x)|x,\hat g)-\alpha|^2\mathds{1}_{\mathcal{B}}\Big]+\EE\Big[|\PP(Y\in \hat{\mathrm{PI}}_{\alpha}(x)|x,\hat g)-\alpha|^2\mathds{1}_{\mathcal{B}^c}\Big]\\&\lesssim (\delta_x)^{2/3}+\frac{\log(k+1)}{k}+\frac{1}{k^2}
\end{align*}
Further, taking expectation with respect to the randomness of $X$, we have
\begin{align*}
\EE\Big[\big|\PP(Y\in \hat{\mathrm{PI}}_{\alpha}(X)|X,\hat g)-\alpha\big|^2\Big|\hat g\Big]&\lesssim (\EE\delta_X)^{2/3}+\frac{\log(k+1)}{k}+\frac{1}{k^2}\\&\lesssim  (\delta_{\mathtt{NDRE}})^{2/3}+\frac{\log k}{k}.
\end{align*}
Hence, we finish the proof of \eqref{eq:prop-prediction-band-1}.
The proof of \eqref{eq:prop-prediction-band-2} is similar to Proposition \ref{prop:conditional-moments}.
Next, we aim to prove \eqref{eq:prop-prediction-band-3}.
\noindent\textbf{Proof of \eqref{eq:prop-prediction-band-3}}
For any fixed $x$, define \[(l^*_x,u^*_x)=\arg\inf_{(l,u):\PP(Y\in [l,u]|X=x)\ge \alpha}(u-l),\]where we suppress the dependence on $\alpha$ for simplicity of notation.

Due to Condition \ref{cond:prediction-band}, 
we have with probability at least $1-\frac{1}{k^2}$, the event $\mathcal{A}(x)$ defined in \eqref{eq:event-A-cond-quantile} holds.
Define $\hat l_x=\hat Q(x,F^*(x,l^*_x))$ and $\hat u_x=\hat Q(x,F^*(x,u^*_x))$. Notice that $F^*(x,u^*_x)-F^*(x,l^*_x)\ge \alpha$. By the definition of $\hat{\mathrm{PI}}_{\alpha}(x)$ in \eqref{eq:pb-est}, it holds that $\hat u_x-\hat l_x\ge |\hat{\mathrm{PI}}_{\alpha}(x)|$. Hence 
\begin{align*}
&|\hat{\mathrm{PI}}_{\alpha}(x)|-\inf_{(l,u):\PP(Y\in [l,u]|X=x)\ge \alpha}(u-l)\\&\le \hat u_x-\hat l_x- u^*_x-l^*_x\\&\le |\hat u_x-u^*_x|+\hat l_x-l^*_x|\\&\le |\hat Q(x,F^*(x,l^*_x))-Q^*(x,F^*(x,l^*_x))|+|\hat Q(x,F^*(x,u^*_x))-Q^*(x,F^*(x,u^*_x))|.
\end{align*}

For some constant $C$ we have under the event $\mathcal{A}(x)$ that
\begin{align*}
&\sup_{0<\alpha<1}\left[|\hat{\mathrm{PI}}_{\alpha}(X)|-\inf_{(l,u):\PP(Y\in [l,u]|X=x)\ge \alpha}(u-l)\right]\\&\le 2\sup_{0<\alpha<1}|\hat Q(x,\alpha)-Q^*(x,\alpha)|\\&\le C(\delta_x)^{1/3}+C\sqrt{\frac{\log (k+1)}{k}} 
\end{align*}
Also note that $|\hat{\mathrm{PI}}_{\alpha}(x)|\le 2c_0$ for any $x$. It follows that
\begin{align*}
&\EE\sup_{0< \alpha<1}\left[|\hat{\mathrm{PI}}_{\alpha}(x)|-\inf_{(l,u):\PP(Y\in [l,u]|X=x)\ge \alpha}(u-l)\right]\\&\le C(\delta_x)^{1/3}+C\sqrt{\frac{\log (k+1)}{k}}+\frac{2c_0}{k^2}.
\end{align*}
Taking expectation with respect to the randomness of $X$ yields
\begin{align*}
\EE\sup_{0< \alpha<1}\left[|\hat{\mathrm{PI}}_{\alpha}(X)|-\inf_{(l,u):\PP(Y\in [l,u]|X)\ge \alpha}(u-l)\right]\lesssim (\delta_{\mathtt{NDRE}})^{1/3}+\sqrt{\frac{\log k}{k}}.
\end{align*}

\end{proof}

\subsection{Proof of Proposition \ref{prop:conditional-density-estimation}}
\label{sec:proof-conditional-density-estimation}
\begin{proof}
Fix an $x$, we have the following decomposition of the kernel estimator,
\begin{align*}
\hat p(y|x)-p^{\star}(y|x)=\hat p(y|x)-\EE[\hat p(Y|x)\big|
\hat g]+\EE[\hat p(Y|x)\big|
\hat g]-p^{\star}(y|x).
\end{align*}
For the bias term, we have 
\begin{align*}
\left|\EE[\hat p(Y|x)\big|
\hat g]-p^{\star}(y|x)\right|&=\left|\frac{1}{h}\int K\left(\frac{t-y}{h}\right)dF_{\hat g}(x,t)-p^{\star}(y|x)\right|\\
&=\left|\frac{1}{h}\int K\left(\frac{t-y}{h}\right)d(F_{\hat g}(x,t)-F^{\star}(x,t))+ \int K\left(s\right)(p^{\star}(x,y+hs)-p^{\star}(x,y))ds\right|
\\&\le\left|\frac{1}{h^2}\int K'\left(\frac{t-y}{h}\right)(F_{\hat g}(x,t)-F^{\star}(x,t))dt\right|\\&~~~~~~~~~~~~+ \left|\int K\left(s\right)\frac{(hs)^{m-1}}{(m-1)!} \frac{\partial^{m}F^{\star}(x,\tilde y)}{\partial \tilde y^{m}}\Big|_{\tilde y=\xi(y,s)}ds\right|
\\&\le \frac{1}{h^2}\sqrt{\int K'^2\left(\frac{t-y}{h}\right)dt\int (F_{\hat g}(x,t)-F^{\star}(x,t))^2dt}+C_x h^m\\&\le \frac{C_K}{h^{3/2}}\sqrt{\int (F_{\hat g}(x,t)-F^{\star}(x,t))^2dt}+C_x h^{m}
\end{align*}
where in the third equality, we use integration by parts for the first term, Taylor expansion, and $\int K(s)s^ids=0$ for $i\le m-1$ for the second term. And we denote $C_K=\sqrt{\int K'^2(t)dt}$, and $C_x=\frac{1}{h}\int K\left(s\right)\frac{s^{m-1}}{(m-1)!} \frac{\partial^{m}F^{\star}(x,\tilde y)}{\partial \tilde y^{m}}\Big|_{\tilde y=\xi(y,s)}ds$.
Then by the Lipschitzness of $\partial^{m}F^{\star}(x,y)/\partial y^{m}$.
\begin{align*}
C_x&=\frac{1}{h}\int K\left(s\right)\frac{s^{m-1}}{(m-1)!}\left(\frac{\partial^{m}F^{\star}(x,\tilde y)}{\partial \tilde y^{m}}\Big|_{\tilde y=\xi(y,s)}-\frac{\partial^{m}F^{\star}(x,\tilde y)}{\partial \tilde y^{m}}\Big|_{\tilde y=y}\right)ds\\&\le \frac{1}{h}\int K\left(s\right)\frac{s^{m-1}}{(m-1)!} l_f (hs)ds =\int K\left(s\right)\frac{s^{m}}{(m-1)!} l_f ds 
\end{align*}
Hence $C_x$ is bounded by a constant $C_f$.

For the variance term, it holds that
\begin{align*}
\mathsf{Var}[\hat p(Y|x)\big|
\hat g]&\le \frac{1}{kh^2}\int K^2\left(\frac{t-y}{h}\right)dF_{\hat g}(x,t)\\
&=\frac{1}{kh^2}\int K^2\left(\frac{t-y}{h}\right)dF^{\star}(x,t)+\frac{1}{kh^2}\int \frac{2}{h}K\left(\frac{t-y}{h}\right)K'\left(\frac{t-y}{h}\right)(F_{\hat g}(x,t)-F^{\star}(x,t))dt\\
&\lesssim \frac{1}{kh}+\frac{1}{kh^2}\sqrt{\frac{\int (F_{\hat g}(x,t)-F^{\star}(x,t))^2dt}{h}}
\end{align*}
where we use the integral by parts in the equality and the Cauchy-Schwarz Inequality in the second inequality.
Hence we have
\begin{align*}
\EE\Big[\int(\hat p(y|X)-p^{\star}(y|X))^2dy\Big|\hat g\Big]&\le 2\EE\Big[\int(\hat p(y|X)-\EE[\hat p(Y|X)\big|
\hat g,X])^2dy\Big|\hat g\Big]\\&~~~~~~~~~~~~+2\EE\Big[\int(\EE[\hat p(Y|X)\big|
\hat g,X]-p^{\star}(y|X))^2dy\Big|\hat g\Big]\\&\lesssim \frac{1}{h^3}\EE\Big[\int (F_{\hat g}(X,t)-F^{\star}(X,t))^2dt\Big|\hat g\Big]+h^{2m}+\EE\left[\mathsf{Var}[\hat p(Y|X)\big|
\hat g,X]\Big| \hat g\right]\\&\le \frac{1}{h^3}\delta_{\mathtt{NDRE}}+h^{2m}+\frac{1}{kh}+\frac{\sqrt{\delta_{\mathtt{NDRE}}}}{kh^{5/2}}
\end{align*}
Take $h\asymp \max\{k^{-\frac{1}{2m+1}},(\delta_{\mathtt{NDRE}})^{\frac{1}{2m+3}}\}$, we have that $\frac{\sqrt{\delta_{\mathtt{NDRE}}}}{kh^{5/2}}\le \frac{1}{kh}$ and 
\begin{align*}
\EE\Big[\int(\hat p(y|X)-p^{\star}(y|X))^2dy\Big|\hat g\Big]\lesssim  \max\{k^{-\frac{1}{2m+1}},(\delta_{\mathtt{NDRE}})^{\frac{1}{2m+3}}\}^{2m}
\end{align*}
which completes the proof of the first part. The proof of the second part is similar to Proposition \ref{prop:conditional-moments}.
\end{proof}

\subsection{Proof of Proposition \ref{prop:conditional-score-estimation}}
\label{sec:proof-conditional-score-estimation}
\begin{proof}
Recall that \begin{align*}
\hat s(y|x)= \left\{\hat q(y|x)\right\} \Big/ \max\left\{\hat p(y|x),c_d\right\}
\end{align*}
where $\hat q(y|x)=\frac{1}{kh^2}\sum_{l=1}^k \tilde K'\left(\frac{y - \hat g(x,U_l)}{h}\right)$ is an estimator for $q^{\star}(y|x)\frac{d}{dy}[p^{\star}(y|x)]$, and $\hat p(y|x)=\frac{1}{kh}\sum_{l=1}^k K\left(\frac{y - \hat g(x,U_l)}{h}\right)$ is an estimator for $p^{\star}(y|x)$. 
As stated, $\tilde K$ is a $(m-1)$-order kernel satisfying
\begin{align*}
\tilde K(0)=1, \qquad  \int s^j \tilde K(s)ds=0 ~~ \forall j\in [m-2], \qquad \text{and} ~~ \int |s^{m-1} \tilde K(s)|ds< \infty. 
\end{align*}

By Condition \ref{cond:conditional-quantile}, $p^{\star}(y|x)\ge c_d$, and hence
\begin{align}
(\hat s(y|x)-s^{\star}(y|x))^2&=\frac{(\hat q p^{\star}-\max\{\hat p,c_d\} q^{\star})^2}{\max\{\hat p,c_d\}^2p^{*2}}\nonumber\\
&\le \frac{2(\hat q-q^{\star})^2}{\max\{\hat p,c_d\}^2}+\frac{2q^{*2}(p^{\star}-\hat p)^2}{\max\{\hat p,c_d\}^2p^{*2}}\nonumber\\
&\le \frac{2(\hat q-q^{\star})^2}{c_d^2}+\frac{2C_d(p^{\star}-\hat p)^2}{c_d^4}\label{equ:s-to-p-q}
\end{align}
where we abbreviate $\hat p=\hat p(y|x)$ and similarly for $\hat q,p^{\star},q^{\star}$.

It then boils down to bound $\EE\Big[\int(\hat q(y|X)-q^{\star}(y|X))^2dy\Big|\hat g\Big]$ and $\EE\Big[\int(\hat p(y|X)-p^{\star}(y|X))^2dy\Big|\hat g\Big]$.
The latter is bounded in Proposition \ref{prop:conditional-density-estimation} by $C\max\{k^{-\frac{1}{2m+1}},(\delta_{\mathtt{NDRE}})^{\frac{1}{2m+3}}\}^{2m}$, it suffices to bound $\EE\Big[\int(\hat q(y|X)-q^{\star}(y|X))^2dy\Big|\hat g\Big]$.

Fix a $x$, we have the following decomposition of the kernel estimator,
\begin{align*}
\hat q(y|x)-q^{\star}(y|x)=\hat q(y|x)-\EE[\hat q(Y|x)\big|
\hat g]+\EE[\hat q(Y|x)\big|
\hat g]-q^{\star}(y|x).
\end{align*}
For the bias term, we have 
\begin{align*}
\EE[\hat q(Y|x)\big|
\hat g]-q^{\star}(y|x)&=\frac{1}{h^2}\int \tilde K\left(\frac{t-y}{h}\right)dF_{\hat g}(x,t)-q^{\star}(y|x)\\
&=\frac{1}{h^2}\int \tilde K\left(\frac{t-y}{h}\right)d(F_{\hat g}(x,t)-F^{\star}(x,t))+ \int \tilde K\left(s\right)(\frac{d}{dy}p^{\star}(x,y+hs)-\frac{d}{dy}p^{\star}(x,y))ds
\\&=\frac{1}{h^3}\int \tilde K'\left(\frac{t-y}{h}\right)(F_{\hat g}(x,t)-F^{\star}(x,t))dt+ \int \tilde K\left(s\right)\frac{(hs)^{m-2}}{(m-2)!} \frac{\partial^{m}F^{\star}(x,\tilde y)}{\partial \tilde y^{m}}\Big|_{\tilde y=\xi(y,s)}ds
\\&\le \frac{1}{h^3}\sqrt{\int \tilde K'^2\left(\frac{t-y}{h}\right)dt\int (F_{\hat g}(x,t)-F^{\star}(x,t))^2dt}+C_x h^{m-1}\\&\le \frac{C_K}{h^{5/2}}\sqrt{\int (F_{\hat g}(x,t)-F^{\star}(x,t))^2dt}+C_x h^{m-1}
\end{align*}
where in the third equality, we use the integration by parts for the first term, Taylor expansion, and $\int \tilde K(s)s^ids=0$ for $i\le m-2$ for the second term. And we denote $C_K=\sqrt{\int \tilde K^2(t)dt}$, and $C_x=\frac{1}{h}\int \tilde K\left(s\right)\frac{s^{m-2}}{(m-2)!} \frac{\partial^{m}F^{\star}(x,\tilde y)}{\partial \tilde y^{m}}\Big|_{\tilde y=\xi(y,s)}ds$.
Then by the Lipschitzness of $\partial^{m}F^{\star}(x,y)/\partial y^{m}$.
\begin{align*}
C_x&=\frac{1}{h}\int \tilde K\left(s\right)\frac{s^{m-2}}{(m-2)!}\left(\frac{\partial^{m}F^{\star}(x,\tilde y)}{\partial \tilde y^{m}}\Big|_{\tilde y=\xi(y,s)}-\frac{\partial^{m}F^{\star}(x,\tilde y)}{\partial \tilde y^{m}}\Big|_{\tilde y=y}\right)ds\\&\le \frac{1}{h}\int \tilde K\left(s\right)\frac{s^{m-2}}{(m-2)!} l_f (hs)ds =\int \tilde K\left(s\right)\frac{s^{m-1}}{(m-2)!} l_f ds .
\end{align*}
Hence $C_x$ is bounded by a constant $C_f$.

For the variance term, it holds that
\begin{align*}
\mathsf{Var}[\hat q(Y|x)\big|
\hat g]&\le \frac{1}{kh^4}\int \tilde K'^2\left(\frac{t-y}{h}\right)dF_{\hat g}(x,t)\\
&=\frac{1}{kh^4}\int \tilde K'^2\left(\frac{t-y}{h}\right)dF^{\star}(x,t)+\frac{1}{kh^4}\int \frac{2}{h}\tilde K'\left(\frac{t-y}{h}\right)\tilde K''\left(\frac{t-y}{h}\right)(F_{\hat g}(x,t)-F^{\star}(x,t))dt\\
&\lesssim \frac{1}{kh^3}+\frac{1}{kh^4}\sqrt{\frac{\int (F_{\hat g}(x,t)-F^{\star}(x,t))^2dt}{h}}
\end{align*}
where we use the integral by parts in the equality and the Cauchy-Schwarz Inequality in the second inequality.
Hence we have
\begin{align*}
\EE\Big[\int(\hat q(y|X)-q^{\star}(y|X))^2dy\Big|\hat g\Big]&\le 2\EE\Big[\int(\hat q(y|X)-\EE[\hat q(Y|X)\big|
\hat g,X])^2dy\Big|\hat g\Big]\\&~~~~~~~~~~~~+2\EE\Big[\int(\EE[\hat q(Y|X)\big|
\hat g,X]-q^{\star}(y|X))^2dy\Big|\hat g\Big]\\&\lesssim \frac{1}{h^5}\EE\Big[\int (F_{\hat g}(X,t)-F^{\star}(X,t))^2dt\Big|\hat g\Big]+h^{2(m-1)}+\EE\left[\mathsf{Var}[\hat q(Y|X)\big|
\hat g,X]\Big| \hat g\right]\\&\le \frac{1}{h^5}\delta_{\mathtt{NDRE}}+h^{2(m-1)}+\frac{1}{kh^3}+\frac{\sqrt{\delta_{\mathtt{NDRE}}}}{kh^{9/2}}.
\end{align*}
Take $h\asymp \max\{k^{-\frac{1}{2m+1}},(\delta_{\mathtt{NDRE}})^{\frac{1}{2m+3}}\}$, we have that $\frac{\sqrt{\delta_{\mathtt{NDRE}}}}{kh^{9/2}}\le \frac{1}{kh^3}$ and 
\begin{align*}
\EE\Big[\int(\hat q(y|X)-q^{\star}(y|X))^2dy\Big|\hat g\Big]\lesssim h^{2(m-1)}\lesssim  \max\{k^{-\frac{1}{2m+1}},(\delta_{\mathtt{NDRE}})^{\frac{1}{2m+3}}\}^{2(m-1)}.
\end{align*}

Hence from \eqref{equ:s-to-p-q} and Proposition \ref{prop:conditional-density-estimation}, we have 
\begin{align*}
\EE\Big[\int(\hat s(y|X)-s^{\star}(y|X))^2dy\Big|\hat g\Big]\lesssim \max\{k^{-\frac{1}{2m+1}},(\delta_{\mathtt{NDRE}})^{\frac{1}{2m+3}}\}^{2(m-1)}
\end{align*}
which completes the proof of the first part. The proof of the second part is similar to Proposition \ref{prop:conditional-moments}.

\end{proof}

\section{General Kernel function}
\label{sec:general-kernel}
\subsection{Distributional Regression Using Kernel MMD Loss}
\label{sec:method-mmd}

Let $\mathcal{D}=\{(X_i, Y_i)\}_{i=1}^n$ be an i.i.d. sample from $\mu_0$. 
Let $(\mathcal{Y},\mathcal{B})$ be a measurable space and let 
$k:\mathcal{Y}\times \mathcal{Y}\to \mathbb{R}$ be a measurable, symmetric, positive definite kernel with associated reproducing kernel Hilbert space $\mathcal{H}_k$. 
We assume $k$ is bounded, i.e., $\sup_{y\in\mathcal{Y}} k(y,y) \leq \kappa^2 < \infty$.

For a fixed covariate value $x$, denote by $P_{Y|X=x}$ the conditional distribution of $Y$ given $X=x$. 
Given a stochastic generator $g:\mathcal{X}\times \mathcal{U}\to \mathcal{Y}$ and a noise variable $U\sim \mu_u$, define the induced conditional distribution
\[
P_{g,x} := \mathcal{L}\big(g(x,U)\big).
\]

We measure the discrepancy between $P_{g,x}$ and $P_{Y|X=x}$ using the squared maximum mean discrepancy, defined as
\begin{align}
\mathrm{MMD}_k^2(P,Q)
:=
\mathbb{E}_{Z,Z'\sim P}[k(Z,Z')]
+
\mathbb{E}_{W,W'\sim Q}[k(W,W')]
-
2\mathbb{E}_{Z\sim P,\,W\sim Q}[k(Z,W)].
\end{align}

The $\mu_u$-population-level objective is then given by
\begin{align}
\label{eq:mmd-loss-pop}
\mathsf{R}_n^{(k)}(g)
:=
\frac{1}{n}\sum_{i=1}^n 
\mathrm{MMD}_k^2\big(P_{g,X_i},\, P_{Y|X=X_i}\big).
\end{align}

Expanding the definition and noting that $Y_i$ is a draw from $P_{Y|X=X_i}$, we obtain
\begin{align}
\label{eq:mmd-loss-pop-expanded}
\mathsf{R}_n^{(k)}(g)
=
\frac{1}{n}\sum_{i=1}^n 
\Big\{
\mathbb{E}_{U,U'}\big[k(g(X_i,U), g(X_i,U'))\big]
-
2\mathbb{E}_{U}\big[k(g(X_i,U), Y_i)\big]
\Big\}
+ C_n,
\end{align}
where $C_n := \frac{1}{n}\sum_{i=1}^n k(Y_i,Y_i')$ is independent of $g$ and can be ignored in optimization.

Thus, minimizing $\mathsf{R}_n^{(k)}(g)$ corresponds to matching the conditional distribution of $g(X,U)$ to that of $Y|X$ in the RKHS induced by $k$. We term this approach \emph{kernel distributional regression}.

\vspace{0.5em}
\noindent\textbf{Empirical Objective.}
As in \eqref{eq:loss-n}, the expectations in \eqref{eq:mmd-loss-pop-expanded} are generally intractable for complex function classes such as neural networks. 
We approximate them via Monte Carlo sampling. 

Let $m=(m_1,m_2)$ and let $\mathcal{U}=\{U_{b,i,k}\}$ be i.i.d. samples from $\mu_u$, independent of $\mathcal{D}$. 
We define the empirical objective
\begin{align}
\label{eq:mmd-loss-emp}
\begin{split}
\mathsf{R}_{n,m}^{(k)}(g)
=
\frac{1}{m_1}\sum_{b=1}^{m_1}\frac{1}{n}\sum_{i=1}^n
\Bigg(
&
\frac{1}{m_2(m_2-1)}\sum_{k\neq k'}
k\big(g(X_i,U_{b,i,k}),\, g(X_i,U_{b,i,k'})\big)
\\
&\quad
-
\frac{2}{m_2}\sum_{k=1}^{m_2}
k\big(g(X_i,U_{b,i,k}),\, Y_i\big)
\Bigg).
\end{split}
\end{align}

We minimize $\mathsf{R}_{n,m}^{(k)}(g)$ over the stochastic neural network class 
$\mathcal{H}_{\mathtt{nn}}(d,L,N,\tilde d,M,B)$.

\vspace{0.5em}
\noindent\textbf{Computational Considerations.}
As in the energy loss formulation, the quadratic complexity in $m_2$ suggests choosing a small $m_2\ge 2$ and a larger $m_1$. 
The computations across $b\in[m_1]$ can be parallelized, while those across $k\in[m_2]$ cannot. 
The statistical accuracy depends on the total number of auxiliary samples $m_1 m_2$, making this decomposition computationally advantageous.

\vspace{0.5em}
\noindent\textbf{Estimator.}
Our estimator $\hat g$ approximately minimizes \eqref{eq:mmd-loss-emp} in the sense that
\begin{align}
\label{eq:mmd-estimator}
\mathsf{R}_{n,m}^{(k)}(\hat g)
\le 
\inf_{g\in \mathcal{H}_{\mathtt{nn}}(d,L,N,\tilde d,M,B)}
\mathsf{R}_{n,m}^{(k)}(g)
+ \delta_{\mathtt{opt}},
\end{align}
where $\delta_{\mathtt{opt}}$ denotes the optimization error.

\vspace{0.5em}
\noindent\textbf{Relation to Energy Loss.}
When $\mathcal{Y}\subseteq \mathbb{R}_+$ and $k(x,y)=\min(x,y)$, the MMD loss \eqref{eq:mmd-loss-pop} is equivalent to the energy loss in \eqref{eq:loss0}, and admits the representation
\[
\mathrm{MMD}_k^2(P,Q)
=
\int (F_P(t)-F_Q(t))^2\,dt,
\]
where $F_P$ and $F_Q$ are cumulative distribution functions. 
Thus, the energy loss can be viewed as a special case of kernel MMD regression corresponding to a first-order Sobolev kernel.


\subsection{Excess Risk Bound for General Kernel MMD Loss}
\label{sec:proof-general-kernel}

Let $k:\mathbb{R}\times \mathbb{R}\to \mathbb{R}$ be a measurable, symmetric, positive definite kernel with RKHS $\mathcal{H}_k$ and canonical feature map $\phi(y)=k(y,\cdot)\in \mathcal{H}_k$. Throughout this section, we impose the following condition.

\begin{condition}
\label{cond:kernel-bounded-lipschitz}
There exist constants $\kappa,L_k>0$ such that
\begin{align*}
    \sup_{y\in[-M,M]} k(y,y)\le \kappa^2,
    \qquad 
    \|\phi(y)-\phi(y')\|^2_{\mathcal H_k}\le L_k |y-y'|,
    \qquad \forall y,y'\in[-M,M].
\end{align*}
Moreover, $|Y|\le M$ almost surely and every $g\in\mathcal{H}_{\mathtt{nn}}(d+1,N,L,B,M)$ satisfies $|g(x,u)|\le M$ for all $(x,u)\in[0,1]^d\times[0,1]$.
\end{condition}

\begin{remark}
It is straightforward to verify that the first-order Sobolev kernel, linear kernel, and Gaussian kernel all satisfy this condition.
\end{remark}

For any $g\in \mathcal{H}_{\mathtt{nn}}(d+1,N,L,B,M)$, define its conditional kernel mean embedding
\begin{align}
\label{eq:def-mu-g}
    \mu_g(x):=\mathbb E_U[\phi(g(x,U))]\in\mathcal H_k,
\end{align}
and define the target conditional embedding
\begin{align}
\label{eq:def-mu-star}
    \mu^\star(x):=\mathbb E[\phi(Y)\mid X=x]\in\mathcal H_k.
\end{align}
We also write
\begin{align*}
    \| \mu_g-\mu_{\tilde g}\|_2^2
    :=
    \mathbb E_X\big[\|\mu_g(X)-\mu_{\tilde g}(X)\|_{\mathcal H_k}^2\big].
\end{align*}

Recall the empirical kernel objective
\begin{align}
\label{eq:kernel-loss-n}
\begin{split}
    \mathsf R_{n,m}^{(k)}(g)
    =
    \frac{1}{m_1}\sum_{b=1}^{m_1}\frac{1}{n}\sum_{i=1}^n
    \Bigg(
    \frac{1}{m_2(m_2-1)}\sum_{r\neq s}
    k(g(X_i,U_{b,i,r}),g(X_i,U_{b,i,s}))
    -
    \frac{2}{m_2}\sum_{r=1}^{m_2}k(g(X_i,U_{b,i,r}),Y_i)
    \Bigg).
\end{split}
\end{align}
Its population-level counterpart is
\begin{align}
\label{eq:kernel-loss-pop}
    \mathsf R_{\infty,\infty}^{(k)}(g)
    :=
    \mathbb E\Big[
        \|\mu_g(X)\|_{\mathcal H_k}^2
        -
        2\langle \phi(Y),\mu_g(X)\rangle_{\mathcal H_k}
    \Big].
\end{align}
Similarly, define the exact empirical counterpart
\begin{align}
\label{eq:kernel-loss-emp-exact}
    \mathsf R_{n,\infty}^{(k)}(g)
    :=
    \frac{1}{n}\sum_{i=1}^n
    \Big(
        \|\mu_g(X_i)\|_{\mathcal H_k}^2
        -
        2\langle \phi(Y_i),\mu_g(X_i)\rangle_{\mathcal H_k}
    \Big).
\end{align}

The following identity is the analogue of the energy identity.

\begin{lemma}
\label{lemma:kernel-risk-identity}
For every $g\in\mathcal{H}_{\mathtt{nn}}(d+1,N,L,B,M)$,
\begin{align}
\label{eq:kernel-excess-risk-identity}
    \mathsf R_{\infty,\infty}^{(k)}(g)-\mathsf R_{\infty,\infty}^{(k)}(g^\star)
    =
    \|\mu_g-\mu^\star\|_2^2,
\end{align}
where $g^\star$ denotes any measurable minimizer of $\mathsf R_{\infty,\infty}^{(k)}$ over all measurable generators, equivalently any measurable map satisfying $\mu_{g^\star}(x)=\mu^\star(x)$ for $\mu_X$-a.e.\ $x$.
\end{lemma}

\begin{proof}
Using \eqref{eq:def-mu-star}, for every measurable $m:[0,1]^d\to \mathcal H_k$ with $\mathbb E\|m(X)\|_{\mathcal H_k}^2<\infty$,
\begin{align*}
    \mathbb E\big[\|\phi(Y)-m(X)\|_{\mathcal H_k}^2\big]
    &=
    \mathbb E\big[\|\phi(Y)\|_{\mathcal H_k}^2\big]
    +
    \mathbb E\big[\|m(X)\|_{\mathcal H_k}^2\big]
    -
    2\mathbb E\big[\langle \phi(Y),m(X)\rangle_{\mathcal H_k}\big] \\
    &=
    \mathbb E\big[\|\phi(Y)\|_{\mathcal H_k}^2\big]
    +
    \mathbb E\big[\|m(X)\|_{\mathcal H_k}^2\big]
    -
    2\mathbb E\big[\langle \mu^\star(X),m(X)\rangle_{\mathcal H_k}\big].
\end{align*}
Taking $m=\mu_g$ gives
\begin{align*}
    \mathsf R_{\infty,\infty}^{(k)}(g)
    =
    \mathbb E\big[\|\phi(Y)-\mu_g(X)\|_{\mathcal H_k}^2\big]
    -
    \mathbb E\big[\|\phi(Y)\|_{\mathcal H_k}^2\big].
\end{align*}
Hence
\begin{align*}
    \mathsf R_{\infty,\infty}^{(k)}(g)-\mathsf R_{\infty,\infty}^{(k)}(g^\star)
    &=
    \mathbb E\big[\|\phi(Y)-\mu_g(X)\|_{\mathcal H_k}^2
    -
    \|\phi(Y)-\mu^\star(X)\|_{\mathcal H_k}^2\big] \\
    &=
    \mathbb E\big[\|\mu_g(X)-\mu^\star(X)\|_{\mathcal H_k}^2\big],
\end{align*}
where the last equality follows from the orthogonality relation
\[
\mathbb E\big[\langle \phi(Y)-\mu^\star(X),\, \mu_g(X)-\mu^\star(X)\rangle_{\mathcal H_k}\big]=0.
\]
This proves \eqref{eq:kernel-excess-risk-identity}.
\end{proof}

For any $g,\tilde g\in\mathcal H_{\mathtt{nn}}(d+1,N,L,B,M)$, define
\begin{align*}
    \Delta_n^{(k)}(g,\tilde g)
    &=
    \mathsf R_{n,\infty}^{(k)}(g)-\mathsf R_{n,\infty}^{(k)}(\tilde g)
    -
    \Big(\mathsf R_{\infty,\infty}^{(k)}(g)-\mathsf R_{\infty,\infty}^{(k)}(\tilde g)\Big), \\
    \chi_n^{(k)}(g)
    &=
    \mathsf R_{n,m}^{(k)}(g)-\mathsf R_{n,\infty}^{(k)}(g).
\end{align*}
As before, these two quantities separate the randomness from the original observations and the Monte Carlo noise.

We first record the two key propositions. Their proofs follow the same localization and chaining arguments as in Proposition~\ref{prop:instant-dependent-error-bound} and Proposition~\ref{prop:chainig-bernstein}, after replacing the scalar CDF discrepancy by the Hilbert-space discrepancy $\mu_g-\mu_{\tilde g}$. We therefore only use them in the proof of the theorem below.

\begin{proposition}
\label{prop:kernel-instance-dependent}
Under Condition~\ref{cond:bounded} and Condition~\ref{cond:kernel-bounded-lipschitz}, the following event
\begin{align*}
    |\Delta_n^{(k)}(g,\tilde g)|
    \le
    C\Big(
        \kappa^2 \delta_{n,t}^2
        +
        \kappa \delta_{n,t}\|\mu_g-\mu_{\tilde g}\|_2
    \Big),
    \qquad \forall g,\tilde g,
\end{align*}
occurs with probability at least $1-2e^{-t}$, where
\begin{align*}
    \delta_{n,t}:= NL\sqrt{\frac{\log n}{n}}+\sqrt{\frac{t}{n}}.
\end{align*}
\end{proposition}

\begin{proposition}
\label{prop:kernel-chaining-bernstein}
Under Condition~\ref{cond:bounded} and Condition~\ref{cond:kernel-bounded-lipschitz}, the following event
\begin{align*}
    |\chi_n^{(k)}(g)|
    \le
    C\kappa^2\left(
        \frac{\delta_{n,t}}{\sqrt{m_1m_2}}
        +
        \delta_{n,t}^2
    \right),
    \qquad \forall g,
\end{align*}
occurs with probability at least $1-2e^{-t}$.
\end{proposition}

We are now ready to state and prove the oracle inequality.

\begin{theorem}
\label{thm:oracle-general-kernel}
Suppose Condition~\ref{cond:bounded} and Condition~\ref{cond:kernel-bounded-lipschitz} hold, and let $\hat g$ satisfy
\begin{align}
\label{eq:kernel-approx-min}
    \mathsf R_{n,m}^{(k)}(\hat g)
    \le
    \inf_{g\in \mathcal H_{\mathtt{nn}}(d+1,N,L,B,M)} \mathsf R_{n,m}^{(k)}(g)
    +
    \delta_{\mathtt{opt}}.
\end{align}
Then there exists a constant $C>0$ such that, with probability at least $1-4e^{-t}$,
\begin{align}
\label{eq:oracle-general-kernel}
    \|\mu_{\hat g}-\mu^\star\|_2^2
    \le
    C\left(
        \inf_{g\in \mathcal H_{\mathtt{nn}}(d+1,N,L,B,M)} \|\mu_g-\mu^\star\|_2^2
        +
        \kappa^2\delta_{n,t}^2
        +
        \kappa^2\frac{\delta_{n,t}}{\sqrt{m_1m_2}}
        +
        \delta_{\mathtt{opt}}
    \right).
\end{align}
\end{theorem}

\begin{proof}
Fix any $\tilde g\in \mathcal H_{\mathtt{nn}}(d+1,N,L,B,M)$. By Lemma~\ref{lemma:kernel-risk-identity},
\begin{align*}
    \mathsf R_{\infty,\infty}^{(k)}(g)-\mathsf R_{\infty,\infty}^{(k)}(\tilde g)
    =
    \|\mu_g-\mu^\star\|_2^2-\|\mu_{\tilde g}-\mu^\star\|_2^2.
\end{align*}
Expanding the square yields
\begin{align*}
    \mathsf R_{\infty,\infty}^{(k)}(g)-\mathsf R_{\infty,\infty}^{(k)}(\tilde g)
    &=
    \|\mu_g-\mu_{\tilde g}\|_2^2
    +
    2\mathbb E_X\Big[
        \big\langle \mu_g(X)-\mu_{\tilde g}(X),\, \mu_{\tilde g}(X)-\mu^\star(X)\big\rangle_{\mathcal H_k}
    \Big] \\
    &\ge
    \frac12 \|\mu_g-\mu_{\tilde g}\|_2^2
    -
    2\|\mu_{\tilde g}-\mu^\star\|_2^2,
\end{align*}
where the last inequality follows from $2ab\le \frac12 a^2+2b^2$. Applying this with $g=\hat g$, we obtain
\begin{align}
\label{eq:kernel-risk-lower}
    \mathsf R_{\infty,\infty}^{(k)}(\hat g)-\mathsf R_{\infty,\infty}^{(k)}(\tilde g)
    \ge
    \frac12 \|\mu_{\hat g}-\mu_{\tilde g}\|_2^2
    -
    2\|\mu_{\tilde g}-\mu^\star\|_2^2.
\end{align}

On the other hand, by the definition of $\Delta_n^{(k)}$, $\chi_n^{(k)}$, and the approximate optimality \eqref{eq:kernel-approx-min},
\begin{align*}
    \mathsf R_{\infty,\infty}^{(k)}(\hat g)-\mathsf R_{\infty,\infty}^{(k)}(\tilde g)
    &=
    \mathsf R_{\infty,\infty}^{(k)}(\hat g)-\mathsf R_{n,\infty}^{(k)}(\hat g)
    +
    \mathsf R_{n,\infty}^{(k)}(\hat g)-\mathsf R_{n,\infty}^{(k)}(\tilde g)
    +
    \mathsf R_{n,\infty}^{(k)}(\tilde g)-\mathsf R_{\infty,\infty}^{(k)}(\tilde g) \\
    &\le
    |\Delta_n^{(k)}(\hat g,\tilde g)|
    +
    \mathsf R_{n,\infty}^{(k)}(\hat g)-\mathsf R_{n,\infty}^{(k)}(\tilde g) \\
    &\le
    |\Delta_n^{(k)}(\hat g,\tilde g)|
    +
    \mathsf R_{n,m}^{(k)}(\hat g)-\mathsf R_{n,m}^{(k)}(\tilde g)
    +
    |\chi_n^{(k)}(\hat g)|
    +
    |\chi_n^{(k)}(\tilde g)| \\
    &\le
    |\Delta_n^{(k)}(\hat g,\tilde g)|
    +
    |\chi_n^{(k)}(\hat g)|
    +
    |\chi_n^{(k)}(\tilde g)|
    +
    \delta_{\mathtt{opt}}.
\end{align*}
Now intersect the events in Proposition~\ref{prop:kernel-instance-dependent} and Proposition~\ref{prop:kernel-chaining-bernstein}; this event has probability at least $1-4e^{-t}$. On this event,
\begin{align*}
    \mathsf R_{\infty,\infty}^{(k)}(\hat g)-\mathsf R_{\infty,\infty}^{(k)}(\tilde g)
    &\le
    C\Big(
        \kappa^2\delta_{n,t}^2
        +
        \kappa\delta_{n,t}\|\mu_{\hat g}-\mu_{\tilde g}\|_2
    \Big)
    +
    2C\kappa^2\left(
        \frac{\delta_{n,t}}{\sqrt{m_1m_2}}
        +
        \delta_{n,t}^2
    \right)
    +
    \delta_{\mathtt{opt}} \\
    &\le
    \frac14 \|\mu_{\hat g}-\mu_{\tilde g}\|_2^2
    +
    C\kappa^2\delta_{n,t}^2
    +
    C\kappa^2\frac{\delta_{n,t}}{\sqrt{m_1m_2}}
    +
    \delta_{\mathtt{opt}},
\end{align*}
where in the last step we used
\[
C\kappa\delta_{n,t}\|\mu_{\hat g}-\mu_{\tilde g}\|_2
\le
\frac14 \|\mu_{\hat g}-\mu_{\tilde g}\|_2^2
+
C^2\kappa^2\delta_{n,t}^2.
\]
Combining this with \eqref{eq:kernel-risk-lower} gives
\begin{align}
\label{eq:kernel-midpoint-bound}
    \frac14 \|\mu_{\hat g}-\mu_{\tilde g}\|_2^2
    \le
    C\kappa^2\delta_{n,t}^2
    +
    C\kappa^2\frac{\delta_{n,t}}{\sqrt{m_1m_2}}
    +
    \delta_{\mathtt{opt}}
    +
    2\|\mu_{\tilde g}-\mu^\star\|_2^2.
\end{align}

Finally, by the triangle inequality in $L^2(\mu_X;\mathcal H_k)$,
\begin{align*}
    \|\mu_{\hat g}-\mu^\star\|_2^2
    &\le
    2\|\mu_{\hat g}-\mu_{\tilde g}\|_2^2
    +
    2\|\mu_{\tilde g}-\mu^\star\|_2^2 \\
    &\le
    C\left(
        \kappa^2\delta_{n,t}^2
        +
        \kappa^2\frac{\delta_{n,t}}{\sqrt{m_1m_2}}
        +
        \delta_{\mathtt{opt}}
        +
        \|\mu_{\tilde g}-\mu^\star\|_2^2
    \right),
\end{align*}
where we used \eqref{eq:kernel-midpoint-bound} in the second line. Since $\tilde g$ is arbitrary, taking the infimum over $\tilde g\in\mathcal H_{\mathtt{nn}}(d+1,N,L,B,M)$ completes the proof.
\end{proof}

\subsection{Auxiliary Results for General Kernel MMD Loss}
\label{sec:kernel-aux}

We establish the key empirical process bounds used in Theorem~\ref{thm:oracle-general-kernel}. 
The arguments follow the same localization and chaining strategy as in the energy loss case, after reducing the Hilbert-space quantities to scalar function classes.

\paragraph{Reduction to Scalar Function Classes.}

For any $g,\tilde g\in \mathcal H_{\mathtt{nn}}(d+1,N,L,B,M)$, define
\begin{align}
\label{eq:def-v-u-kernel}
\begin{split}
    v_{g,\tilde g}(x)
    &:=
    \|\mu_g(x)-\mu_{\tilde g}(x)\|_{\mathcal H_k}^2, \\
    u_{g,\tilde g}(x,y)
    &:=
    \big\langle 
        \mu_g(x)-\mu_{\tilde g}(x),\,
        \mu_{\tilde g}(x)-\phi(y)
    \big\rangle_{\mathcal H_k}.
\end{split}
\end{align}

Then we can rewrite
\begin{align}
\label{eq:kernel-decomposition}
\begin{split}
    \Delta_n^{(k)}(g,\tilde g)
    =
    (\mathbb P_n-\mathbb P)\big[ v_{g,\tilde g}(X) + 2u_{g,\tilde g}(X,Y)\big],
\end{split}
\end{align}
where $\mathbb P_n$ and $\mathbb P$ denote the empirical and population measures of $(X,Y)$.

Thus, it suffices to control the empirical processes indexed by the scalar function classes
\begin{align*}
    \mathcal V
    &:=
    \{v_{g,\tilde g}: g,\tilde g\in\mathcal H_{\mathtt{nn}}\}, \\
    \mathcal U
    &:=
    \{u_{g,\tilde g}: g,\tilde g\in\mathcal H_{\mathtt{nn}}\}.
\end{align*}

\paragraph{Uniform Boundedness and Lipschitz Reduction.}

Under Condition~\ref{cond:kernel-bounded-lipschitz}, we have for all $g,\tilde g$:
\begin{align}
\label{eq:kernel-bounds}
\begin{split}
    \|\mu_g(x)-\mu_{\tilde g}(x)\|^2_{\mathcal H_k}
    &\le
    L_k \mathbb E_U|g(x,U)-\tilde g(x,U)|, \\
    |v_{g,\tilde g}(x)|
    &\le
    4\kappa^2, \\
    |u_{g,\tilde g}(x,y)|
    &\le
    4\kappa^2.
\end{split}
\end{align}

Moreover, both $v_{g,\tilde g}$ and $u_{g,\tilde g}$ are Lipschitz functionals of $g,\tilde g$ through $\mu_g$, hence their covering numbers can be controlled by those of the neural network class.

\begin{lemma}
\label{lemma:kernel-covering}
Let $\mathcal G:=\mathcal H_{\mathtt{nn}}(d+1,N,L,B,M)$. Then for any $\epsilon>0$,
\begin{align*}
    \log \mathcal N(\epsilon, \mathcal V, \|\cdot\|_\infty)
    \lesssim
    NL \log\left(\frac{n}{\epsilon}\right),
    \qquad
    \log \mathcal N(\epsilon, \mathcal U, \|\cdot\|_\infty)
    \lesssim
    NL \log\left(\frac{n}{\epsilon}\right).
\end{align*}
\end{lemma}

\begin{proof}
This follows from the Lipschitz mapping $g\mapsto \mu_g$ together with standard covering number bounds for ReLU networks. 
The argument is identical to the covering construction in the energy-loss case and is therefore omitted.
\end{proof}

\paragraph{Instance-Dependent Bound.}

\begin{proposition}[Restatement of Proposition~\ref{prop:kernel-instance-dependent}, Kernel analogue of Proposition~\ref{prop:instant-dependent-error-bound}]
\label{prop:kernel-instance-dependent-full}
Under Condition~\ref{cond:bounded} and Condition~\ref{cond:kernel-bounded-lipschitz}, with probability at least $1-2e^{-t}$,
\begin{align}
\label{eq:kernel-instance-bound}
    |\Delta_n^{(k)}(g,\tilde g)|
    \le
    C\Big(
        \kappa^2 \delta_{n,t}^2
        +
        \kappa \delta_{n,t}
        \|\mu_g-\mu_{\tilde g}\|_2
    \Big),
    \qquad \forall g,\tilde g,
\end{align}
where $\delta_{n,t}=NL\sqrt{\frac{\log n}{n}}+\sqrt{\frac{t}{n}}$.
\end{proposition}

\begin{proof}
By \eqref{eq:kernel-decomposition}, the result reduces to bounding
\[
(\mathbb P_n-\mathbb P)f, \quad f\in \mathcal V \cup \mathcal U.
\]
Using Lemma~\ref{lemma:kernel-covering}, the boundedness \eqref{eq:kernel-bounds}, and Bernstein’s inequality combined with chaining, we obtain the same localized empirical process bound as in the energy-loss case. 

The only difference is that the variance proxy becomes
\[
\mathbb E[u_{g,\tilde g}^2],\mathbb E[v_{g,\tilde g}^2] \lesssim \kappa^2 \|\mu_g-\mu_{\tilde g}\|_2^2,
\]
which yields the mixed term in \eqref{eq:kernel-instance-bound}. 

The rest of the proof follows identically from Proposition~\ref{prop:instant-dependent-error-bound} and is omitted.
\end{proof}

\paragraph{Monte Carlo Approximation Error.}

\begin{proposition}[Restatement of Proposition~\ref{prop:kernel-chaining-bernstein}, Kernel analogue of Proposition~\ref{prop:chainig-bernstein}]
\label{prop:kernel-chaining-bernstein-full}
Under Condition~\ref{cond:bounded} and Condition~\ref{cond:kernel-bounded-lipschitz}, with probability at least $1-2e^{-t}$,
\begin{align}
\label{eq:kernel-mc-bound}
    |\chi_n^{(k)}(g)|
    \le
    C\kappa^2\left(
        \frac{\delta_{n,t}}{\sqrt{m_1 m_2}}
        +
        \delta_{n,t}^2
    \right),
    \qquad \forall g.
\end{align}
\end{proposition}

\begin{proof}
The Monte Carlo error can be written as an average of degenerate $U$-statistics indexed by $g$. 
Since $k$ is bounded, the kernel evaluations are uniformly bounded by $\kappa^2$, and the variance scales as $(m_1 m_2)^{-1}$. 

Applying Bernstein’s inequality together with a union bound over a covering net of $\mathcal G$, and then chaining, yields \eqref{eq:kernel-mc-bound}. 

The argument is identical to Proposition~\ref{prop:chainig-bernstein}, replacing $\min(g,g')$ by $k(g,g')$, and is therefore omitted.
\end{proof}

\section{A More General Architecture}
In this section, we consider a more general deep neural network architecture with ReLU activation $\sigma(\cdot) = \max\{0, \cdot\}$ that injects extra noise in each layer of a standard fully connected deep neural network, and call it \emph{stochastic deep ReLU network} for short. Let $L, N, \tilde d, N_s$ be four positive integers satisfying $N_s < N$, a \emph{deep ReLU network with depth $L$, ambient width $N$, and noise width $(\tilde d, N_s)$} admits the form that is defined recursively as~\footnote{{Here $\tilde d$ is the dimension of noise components injected at the first layer while $N_s$ is the dimension of noise components introduced in each intermediate layer, analogous to the dropout as an implicit regularization.}}
\begin{align}
\label{eq:nn-architecture}
\begin{split}
    g(x, u) &= T_{L+1} \circ [g^{(L)}(x, \{u_\ell\}_{\ell=1}^L)] ~~~ \text{ where } \\
    &\forall \ell \in [L], \qquad g^{(\ell)}(x, \{u_r\}_{r=1}^\ell) = 
        \bar{\sigma}_\ell \circ T_{\ell} \circ [g^{(\ell-1)}(x, \{u_r\}_{r=1}^{\ell-1}), u_\ell]
    ~~\text{ with }~~ g^{(0)}(x, \emptyset) = x. 
\end{split}
\end{align} Here $T_{l}(z) = W_l z + b_l: \mathbb{R}^{d_l} \to \mathbb{R}^{d_{l+1}}$ is a linear map with weight matrix $W_l \in \mathbb{R}^{(d_{l}-N_s 1_{\{l\neq L+1\}})\times d_{l-1}}$ and bias vector $b_{l} \in \mathbb{R}^{d_{l} - N_s 1_{\{l\neq L+1\}}}$, where $(d_0,d_1\ldots, d_{L-1}, d_L, d_{L+1}) = (d, N, \ldots, N, N-N_s, 1)$, and $\bar{\sigma}_l: \mathbb{R}^{d_l - N_s} \to \mathbb{R}^{d_l - N_s}$ applies the ReLU activation $\sigma(\cdot)$ to each entry of a $d_l$-dimensional vector in $l\in [L]$. The set of noise vectors is $u=\{u_\ell\}_{\ell=1}^L$ with $u_\ell \in \mathbb{R}^{N_s}$ for $\ell \ge 2$ and $u_1 \in \mathbb{R}^{\tilde d}$. Here, the equal width is for presentation simplicity. See an example of $L=3$, $N=7$, $d=3$ and $(N_s, \tilde d)=(3,2)$ in Fig. \ref{fig:snn}.

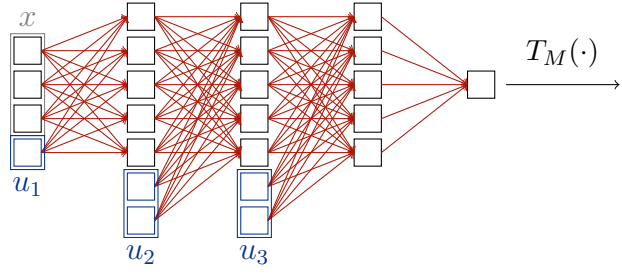
\begin{figure}
\begin{center}
\begin{tikzpicture}[scale=1.5]
\draw[gray] (-0.03, -0.33) rectangle +(0.24+0.03*2, 0.3*2+0.24+0.03*2);
\draw[gray] (0.12, 0.7) node {$x$};
\foreach \x in {-1, 0, 1}
{
	\draw[black] (0, \x*0.3) rectangle +(0.24, 0.24);
}
\foreach \l in {1, 2, 3}
\foreach \x in {-2, -1, 0, 1, 2}
{
	\draw[black] (\l, \x*0.3) rectangle +(0.24, 0.24);
}
\draw[black] (4, 0) rectangle+(0.24, 0.24);
\foreach \x in {-2, -1, 0, 1}
	\foreach \y in {-2, -1, 0, 1, 2}
{
	\draw[->, myred] (0+0.24, \x*0.3+0.12) -- (1, \y * 0.3+0.12);
}

\foreach \l in {2,3}
	\foreach \x in {-4, -3, -2, -1, 0, 1, 2}
		\foreach \y in {-2, -1, 0, 1, 2}
		{
			\draw[->, myred] (\l-1+0.24, \x*0.3+0.12) -- (\l, \y * 0.3+0.12);
		}
		
\foreach \x in {-2, -1, 0, 1, 2}
	{
		\draw[->, myred] (4-1+0.24, \x*0.3+0.12) -- (4, 0 * 0.3+0.12);
	}

\draw[black] (4+0.24+0.6, 0.4) node{\small $T_M(\cdot)$};
\draw[->] (4+0.24 + 0.1, 0.12) -- (4+0.24+1.1, 0.12);

\draw[myblue] (-0.03, -0.93+0.3) rectangle +(0.24+0.03*2, 0.24+0.03*2);
\foreach \x in {-2}
{
	\draw[myblue] (0, \x*0.3) rectangle +(0.24, 0.24);
}
\draw[myblue] (0.12, -0.93-0.15+0.3) node {\myblue{$u_1$}};

\foreach \l in {2, 3}
{
    \draw[myblue] (-0.03+\l-1, -0.93-0.3) rectangle +(0.24+0.03*2, 0.3*1+0.24+0.03*2);
    \foreach \x in {-4, -3}
    {
	   \draw[myblue] (\l-1, \x*0.3) rectangle +(0.24, 0.24);
    }
    \draw[myblue] (0.12+\l-1, -0.93-0.3-0.15) node {\myblue{$u_\l$}};
}
\end{tikzpicture} 
\end{center}
\caption{A visualization of the generalized stochastic neural network: the depth $L=3$, it takes $p=3$ dimensional covariate vector $x$ and $\tilde d=1$ dimensional noise vector $u_1$ as input in the input layer. Compared with a standard fully connected neural network, it also injects $N_s=2$ dimensional noise vector in hidden layers $\ell\in \{1,2\}$. } 
\label{fig:snn}
\end{figure}

\section{A variation of Algorithm \ref{algo1} with smaller $m_1$.}
\label{sec:algo2}

\begin{algorithm}
	\caption{A variation of Algorithm \ref{algo1} with smaller $m_1$.}
	\label{algo2}
	\begin{algorithmic}[1]
		\State \textbf{Hyper-parameter: } number of epoches $T$, batch size $B$, number of noise samples per data $K$.
		\State \textbf{Input:} data $\{(X_i, Y_i)\}_{i=1}^n$
		\State Initialize $\theta$ with random weights
		\myblue{\State Sample noises $\{{U}_{b,i,k}\}_{b\in [m_1], i\in [n], k\in [K]}$ i.i.d. from $\mu_u$.}
		\For {$t \in \{1,\ldots, T\}$} 
		\State Shuffled sample $\breve{\mathcal{D}} =  \{(\breve{X}_{i}, \breve{Y}_{i})\}_{i=1}^n = \{(X_{\pi(i)}, Y_{\pi(i)})\}_{i=1}^n$ by random permutation $\pi$ in $[n]$.
		\For {$\ell \in \{1,\ldots, \lfloor n/B\rfloor\}$}
		\State Get batch ${\mathcal{B}} = \{(\tilde{X}_{i}, \tilde{Y}_i)\}_{i=1}^{B} = \{(\breve{X}_{i+(\ell-1)B}, \breve{Y}_{i+(\ell-1)B})\}_{i=1}^{B}$.
		\myblue{\State Get noises $\{\tilde{U}_{i,k}\}_{i\in [B], k\in [K]} = \{{U}_{(t\mod m_1) + 1, \pi(i+(\ell-1)B),k}\}_{i\in [B], k\in [K]}$}
		\State Update $\theta$ by descending its gradient
		\begin{align*}
			\nabla_\theta \left[\frac{1}{B}\sum_{i=1}^B \left\{\frac{2}{K} \sum_{k=1}^K |\tilde{Y}_i-g_\theta(\tilde{X}_i,\tilde{U}_{i,k})| - \frac{1}{K(K-1)}\sum_{k\neq l} |g_\theta(\tilde{X}_i, \tilde{U}_{i,k}) - g_\theta(\tilde{X}_i, \tilde{U}_{i, l})|\right\}\right]
		\end{align*}
		\EndFor
		\EndFor
		\State \textbf{Output:} distributional regression estimate $g_\theta$. 
	\end{algorithmic}
\end{algorithm}


\end{document}